\documentclass[11pt,a4paper]{article}

\usepackage{amsmath,amsfonts,amssymb,amsthm,framed}
\usepackage[linesnumbered,algosection,noline]{algorithm2e}

\usepackage{algorithmic}
\newcommand{\alg}[1]{\mbox{\sf #1}}

 \usepackage[dvipsnames,svgnames,x11names]{xcolor}
  \definecolor{shadecolor}{gray}{0.9}
\usepackage{boxedminipage}
\usepackage{mdwlist}
\newtheorem*{theorem*}{Theorem}
\newtheorem*{lemma*}{Lemma}
\newtheorem*{definition*}{Definition}

\usepackage{todonotes,soul}
\usepackage{microtype}
\usepackage[dvipsnames,svgnames,x11names]{xcolor}

\newcommand{\spn}{\textrm{\rm span}}

\DeclareMathOperator{\operatorClassNP}{NP}
\newcommand{\classNP}{\ensuremath{\operatorClassNP}}

\DeclareMathOperator{\operatorClassFPT}{FPT}
\newcommand{\classFPT}{\ensuremath{\operatorClassFPT}}
\DeclareMathOperator{\operatorClassW}{W}
\newcommand{\classW}[1]{\ensuremath{\operatorClassW[#1]}}
\DeclareMathOperator{\operatorClassParaNP}{Para-NP}
\newcommand{\classParaNP}{\ensuremath{\operatorClassParaNP}}

\newcommand{\OO}{{\mathcal O}}
\DeclareMathOperator{\spanS}{span}
\DeclareMathOperator{\rank}{rank}
\DeclareMathOperator{\type}{type}
\DeclareMathOperator{\sumS}{sum}
\DeclareMathOperator{\edges}{edges}
\DeclareMathOperator{\expens}{expensive}
\DeclareMathOperator{\cheap}{cheap}
\DeclareMathOperator{\disR}{disr}
\DeclareMathOperator{\disC}{disc}
\DeclareMathOperator{\cont}{contribute}

\DeclareMathOperator{\children}{children}
\DeclareMathOperator{\subtree}{subtree}
\DeclareMathOperator{\image}{image}
\DeclareMathOperator{\domain}{domain}
\DeclareMathOperator{\nil}{nil}

\newcommand{\SSp}{\textsc{Space Cover}\xspace}
\newcommand{\probSteiner}{\textsc{Steiner Tree}\xspace}
\newcommand{\probSteinerF}{\textsc{Steiner Forest}\xspace}
\newcommand{\rrper}{$r$-rank perturbed\xspace}

\newcommand{\defproblemu}[3]{
  \vspace{1mm}
\noindent\fbox{
  \begin{minipage}{0.95\textwidth}
  #1 \\
  {\bf{Input:}} #2  \\
  {\bf{Question:}} #3
  \end{minipage}
  }
  \vspace{1mm}
}

\newcommand{\defproblemuTask}[3]{
  \vspace{1mm}
\noindent\fbox{
  \begin{minipage}{0.95\textwidth}
  #1 \\
  {\bf{Input:}} #2  \\
  {\bf{Task:}} #3
  \end{minipage}
  }
  \vspace{1mm}
}

\newcommand{\defparproblem}[4]{
  \vspace{1mm}
\noindent\fbox{
  \begin{minipage}{0.96\textwidth}
  \begin{tabular*}{\textwidth}{@{\extracolsep{\fill}}lr} #1  & {\bf{Parameter:}} #3 \\ \end{tabular*}
  {\bf{Input:}} #2  \\
  {\bf{Question:}} #4
  \end{minipage}
  }
  \vspace{1mm}
}

\usepackage[margin=1in]{geometry}

\newcommand{\Oh}{\mathcal{O}}

\newcommand{\bran}[1]{branchable\xspace}

\newcommand{\myparagraph}[1]{\smallskip\noindent{\textbf{\sffamily #1}}}

\newtheorem{theorem}{Theorem}
\newtheorem{lemma}{Lemma}[section]

\newtheorem{corollary}{Corollary}
\newtheorem{definition}{Definition}[section]
\newtheorem{observation}{Observation}[section]
\newtheorem{proposition}{Proposition}[section]

\theoremstyle{definition}

\usepackage{framed}

\usepackage{amsthm}
\usepackage{thmtools}

\usepackage{thm-restate}

\usepackage{thmtools}
\usepackage{thm-restate}

\usepackage{booktabs}

 \usepackage[ pdftex, plainpages = false, pdfpagelabels, 
                 bookmarks=false,
                 bookmarksopen = true,
                 bookmarksnumbered = true,
                 breaklinks = true,
                 linktocpage,
                 pagebackref,
                 colorlinks = true,  % was true
                 linkcolor = blue,
                 urlcolor  = blue,
                 citecolor = red,
                 anchorcolor = green,
                 hyperindex = true,
                 hyperfigures
                 ]{hyperref} 
\usepackage{cleveref}

\usepackage{todonotes}

\newcommand{\spcmain}{{\sc Space Cover}\xspace}
\newcommand{\spcgraphlong}{{\sc Space Cover on Perturbed Graphic Matroid}\xspace}
\newcommand{\spcgraph}{{\sc Space Cover on PGM}\xspace}
\newcommand{\spcduallong}{{\sc Space Cover on Dual of Perturbed Graphic Matroid}\xspace}
\newcommand{\spcdual}{{\sc Space Cover on Dual-PGM}\xspace}
\newcommand{\pattern}{{\sc Pattern Cover}\xspace}
\newcommand{\edc}{{\sc Edge-Set Cover}\xspace}
\newcommand{\aedc}{{\sc Annotated Edge-Set Cover}\xspace}



\begin{document}
%\begin{titlepage}
%\def\thepage{}
%\thispagestyle{empty}

\title{Covering Vectors by Spaces in Perturbed Graphic Matroids and Their Duals}

\author{
Fedor V. Fomin\thanks{
Department of Informatics, University of Bergen, Norway. {\ttfamily fomin@ii.uib.no}}
\and
Petr A. Golovach\thanks{
Department of Informatics, University of Bergen, Norway. {\ttfamily petr.golovach@uib.no}}
\and 
 Daniel Lokshtanov\thanks{University of California Santa Barbara, USA. { \ttfamily daniello@ucsb.edu}}
 \and
Saket Saurabh\thanks{Institute of Mathematical Sciences, HBNI, Chennai, India. {\ttfamily saket@imsc.res.in}}
\and
Meirav Zehavi\thanks{Ben-Gurion University, Israel. {\ttfamily meiravze@bgu.ac.il}}
}

\date{}

\maketitle

 \begin{abstract}
Perturbed graphic matroids are binary  matroids that can be obtained from a graphic matroid by  adding a noise of small rank. More precisely,  \rrper graphic matroid $M$ is a binary  matroid that can be  represented in the form $I +P$, where $I$ is the incidence matrix of some graph  and $P$ is a binary matrix of rank at most $r$. Such matroids naturally appear in a number of  theoretical and  applied settings. The main motivation behind our work is an attempt to understand which   parameterized  algorithms for various problems on graphs could be lifted to perturbed graphic matroids.  
 
We study the parameterized complexity of a natural generalization (for matroids) of   the following  fundamental problems on graphs: \textsc{Steiner Tree} and \textsc{Multiway Cut}. In this generalization, called the  \SSp problem, we are given a   binary matroid $M$ with a ground set $E$, a set of \emph{terminals} $T\subseteq E$, and a non-negative integer $k$. The task is to decide whether  $T$ can be spanned by a subset of $ E\setminus T$ of size at most $k$.

We prove  that  
on graphic matroid perturbations,  for every fixed $r$, \SSp is fixed-parameter tractable parameterized by $k$. On the other hand,  the problem becomes \classW1-hard when  parameterized by $r+k+|T|$ and it is \classNP-complete for $r\leq 2$ and $|T|\leq 2$.

On cographic matroids, that are the duals of graphic matroids,  \SSp  generalizes another fundamental and well-studied problem, namely \textsc{Multiway Cut}.
We show that on the duals of perturbed graphic matroids the  \SSp  problem is fixed-parameter tractable parameterized by $r+k$. 
\end{abstract}

 \section{Introduction}\label{sec:intro}
In this paper we develop parameterized algorithms on low-rank perturbations of graphic matroids and their duals. These matroids   and their matrices naturally appear in various settings. For example, 
in the emerging  Matroid Minors Project of   Geelen,   Gerards,   and Whittle~\cite{GeelenB14}, perturbed matroids play a  significant role  in  characterization of   proper minor-closed classes  of binary matroids. More precisely, for each proper minor-closed class $\mathcal{M}$ of binary matroids, there exist nonnegative integer  $r$ such that  every sufficiently highly connected matroid $M \in \mathcal{M}$,
is either a perturbation of graphic or cographic matroid. In other words, 
there exist matrices  $I,P \in \mbox{GF(2)}^{\ell\times n}$ such that $I$ is the incidence matrix of a graph, the rank of $P$ is at most $r$, and either $M$ or its dual $M^*$ is represented by $I + P$.  
Another example of closely related concept is the 
 robust Principal Component Analysis (PCA),  a popular approach to 
 robust subspace learning and tracking by decomposing the data matrix  into low-rank and sparse matrices. 
 Here  data  matrix $M$ is assumed to be a superposition of a low-rank perturbation component $P$ and a sparse component $I$, that
 is, $M=I+P.$ See   
Cand{\`{e}}s et al.   \cite{CandesLMW11}, Wright et al.  \cite{WrightGRPM09}, and 
 Chandrasekaran et al. \cite{ChandrasekaranSPW11} for further references on robust PCA. In particular, one of the well-studied, see e.g.  
 \cite{ShahidPKPV16,ZhaoKVRM18}, of the variants of robust PCA is when the structure of the sparse matrix $I$ is imposed from the structure of some graph.    
   Perturbed matroids  also come naturally in the settings when a structural  input is corrupted by a noise. 
  In graph algorithms, one of the questions studied in the literature about corrupted inputs is---what happens to   special graph classes when they are perturbed adversarially?
  For example, Magen and Moharrami \cite{MagenM09}, and  Bansal, Reichman, Umboh~\cite{Bansal0U17}, studied approximation algorithms on
  noisy minor-free graphs,  which are the graphs obtained from minor-free graphs by corrupting  a fraction of edges and vertices.

  \medskip\noindent\textbf{Our results.} We work with the following classes of binary matroids. 
 A binary matroid  $M$ such that   $M$ 
  can be   represented in the form 
$I +P$, where $I$ is the incidence matrix of some graph and $P$ is a binary matrix of rank at most $r$, is called 
   the \emph{\rrper graphic matroid}.  Similarly, when the  dual $M^*=I +P$ for some incidence matrix $I$ and $r$-rank matrix $P$, we refer to $M$ as to  \emph{\rrper cographic matroid}.
   
In this paper we study parameterized complexity  on binary perturbed matroids  of the following generic problem.
  Let us remind that in a matroid  $M$,  a set $F$ spans $T$, denoted by  $T\subseteq \spn(F)$, if the 
  sets $F$ and $T\cup F$ are of the same rank.
 
\defproblemu{\spcmain}%
{A binary matroid $M$ with a ground set $E$, 
a set of \emph{terminals} $T\subseteq E$, and a non-negative integer $k$.}%
{Is there a set $F\subseteq E\setminus T$ with $|F|\leq k$ such that $T\subseteq \spn(F)$?}

In other words,  \SSp is  the problem of covering a given set of vectors $T$  over GF(2) by   a minimum-dimension subspace of the space generated by vectors from  $E\setminus T$.
\SSp encompasses various problems   arising in different domains,  such as    coding theory,  machine learning, 
 and   graph algorithms. 
For example, 
  \spcmain {} is a natural generalization of \textsc{Matroid Girth}, the problem of finding a minimum set of dependent elements in a matroid.  
  \textsc{Matroid Girth}  can be reduced to \spcmain {}  by  computing for each element $t$ of $M$  a minimum set of elements of the remaining part of the matroid that covers $T=\{t\}$.

  On graphs (equivalently, special classes of binary matroids, namely graphic and cographic matroids),  \SSp
 generalizes    well-studied optimization problems      \probSteiner and  \textsc{Multiway Cut}. Various algorithmic techniques were developed for these problems, see e.g. \cite{CyganFKLMPPS15},  and it is very interesting to see which of these techniques, if any, can be lifted to matroids.

We obtain the following results about the complexity of \SSp  on {\rrper matroids}. (In all these results we assume  that   representation of \rrper matroid in the form $I+P$ is given.)
\begin{itemize}

\item Our first main algorithmic result  (Theorem~\ref{thmgraphic})  states the following: On \rrper graphic matroids, 
for every fixed $r$, \SSp  is fixed-parameter tractable (\classFPT)  when parameterized by $k$. 
\item We also show that a ``weaker'' parameterization makes the problem intractable. More precisely, we prove that on \rrper graphic matroids, \SSp is \classW{1}-hard when parameterized by $r+k+|T|$ (Theorem~\ref{thm:w-hard}).  We also  prove that the problem is \classNP-complete for $r\leq 2$ and $|T|\leq 2$ (Theorem~\ref{thm:paraNP-hard}). 
\item Our second main algorithmic result  (Theorem~\ref{thmdual}) concerns   \rrper cographic matroids. This theorem states that  \SSp is 
 \classFPT{} on \rrper cographic matroids when parameterized by $r+k$.  We find it a bit surprising that the parameterized complexity of  \SSp  is different  on  \rrper graphic and cographic matroids.
\end{itemize}

 \medskip\noindent\textbf{Previous work.}
   Geelen and Kapadia \cite{GeelenK15} studied the problem of computing the girth  of  a binary \rrper matroid. 
  (The girth of a matroid is the length of its shortest circuit.) 
   Geelen and Kapadia have proved that the girth of an  
  \rrper matroid is fixed-parameter tractable being parameterized by $r$.  Let us note that while \SSp generalizes  \textsc{Matroid Girth}, 
  our results are incomparable. In our FPT result for \rrper graphic matroids the parameter is $k$ while the parameter $r$ should be fixed. As Theorems~\ref{thm:w-hard} and~\ref{thm:paraNP-hard} show, the requirement that $r$ should be fixed and that $k$ should be  the parameter are, most likely,  unavoidable. 
  For binary matroids, \textsc{Matroid Girth} has several equivalent formulations. For example, it is equivalent to 
 the \textsc{Minimum Distance} problem from coding theory, which  ask for a minimum dependent set of columns in
a matrix  over GF(2). The complexity of this problem was open till 1997,   
  when Vardy showed it to be NP-complete~\cite{Vardy97}. On the other hand, Geelen, Gerards and Whittle in \cite{GeelenGW15} conjecture that 
  for any proper minor-closed class $\mathcal{M}$ of binary matroids, there is a polynomial-time algorithm for computing the girth of matroids in $\mathcal{M}$. 
 The parameterized version of  the 
   problem, namely \textsc{Even Set},  asks whether there is a dependent set $F\subseteq X$ of  size at most $k$. The parameterized complexity of  \textsc{Even Set} was a long-standing open question in the area, see e.g.  \cite{DowneyFbook13},  whose complexity was resolved only recently  \cite{BonnetEM16}.

 \SSp on  graphic and cographic matroids is a  generalization of  \probSteiner and  \textsc{Multiway Cut}, two very well-studied problems on graphs. 
 By the classical result of Dreyfus and Wagner \cite{DreyfusW71}, \probSteiner is 
   fixed-parameter tractable (FPT) parameterized by the  number of terminals $T$. Similar approach can be used to show that \SSp is FPT on graphic matroids. 
   On cographic matroids \SSp is equivalent to the \textsc{Restricted Edge-Subset Feedback Edge Set} introduced by Xiao and Nagamochi~\cite{XiaoN12} who also showed that  
the   problem is FPT parameterized by $k$. Due to its connection to \textsc{Multiway Cut}, the 
  \classNP-completeness result of   Dahlhaus et al.~\cite{DahlhausJPSY94} for \textsc{Multiway Cut} with three terminals implies that 
  \SSp
   is \classNP{}-hard even if  $|T|=3$ on  cographic matroids.  
Fomin et al. in~\cite{FominGLS17plus} extended the results for \SSp on graphic and cographic matroids to a more general class of binary matroids, namely, regular matroids, by providing an algorithm of running time 
  $2^{\Oh(k)}\cdot  ||M|| ^{\Oh(1)}$. While the class of regular matroids is a proper minor-class of binary matroids, this class of matroids is incomparable to the class of perturbed matroids.  It is also known that \SSp is hard on general class of binary matroids: 
By the result of    Downey et al.~\cite{DowneyFVW99}, \SSp is  
 \classW{1}-hard on  binary matroids when parameterized by $k$ even if restricted to the inputs with one terminal.

\section{Overview of Algorithmic Theorems}\label{sec:overview}
In this section, we give short descriptions of both of our algorithmic results. 
The formal definitions of  graph and matroid-related terms that are used in this section are given in Section~\ref{sec:prelims}. 

\subsection{Perturbed Graphic Matroids}
In this section, we give an overview of the proof of the first main result of the paper. Full proof is given in Section~\ref{section:graphic}.

In this case,  \rrper  matroid  $M$ is represented by the perturbed incidence matrix $I(G)$ of a (multi) graph $G$.
Formally we define the following problem. 
\medskip
 
\defproblemu{\spcgraphlong~(\spcgraph)}%
{A (multi) graph $G$ with $n$ vertices and $m$ edges, an $(n\times m)$-matrix $P$ over $GF(2)$ with $\rank(P)\leq r$, a set of 
\emph{terminals} $T\subseteq E$ where $E$ is the set of columns of the matrix $A=I(G)+P$,
and a non-negative integer $k$.}%
{Is there a set $F\subseteq E\setminus T$ with $|F|\leq k$ such that $T\subseteq \spn(F)$ in the binary matroid $M$ represented by $A$?}

 \begin{theorem} \label{thmgraphic}
 For any fixed constant $r$, \spcgraph{} is solvable in time $k^{\OO(k)}\cdot (n+m)^{\OO(1)}$. In particular, \spcgraph{} is \classFPT{} when parameterized by $k$ whenever $r$ is a constant.
 \end{theorem}

We underline that $r$ is a constant here, that is, the constants hidden behind the big-O notation in the running time depend on $r$.

Before proceedings with an overview,  it  is useful to discuss how \spcgraph 
 is  solvable when $r=0$, i.e. on graphic  matroids and what are the main challenges for solving the problem for $r>0$. 
On graphic matroids \SSp corresponds to the following problem. 
Given a set of terminal edges   $T=\{e_1, e_2, \dots, e_s\}$, we want to find a set of at most $k$ edges $F\subseteq E\setminus T$ such that for every $e_i$, graph  $G[F\cup e_i]$ has a cycle containing  $e_i$. This can be seen as a variant of  the \probSteiner, and more generally, of the 
\probSteinerF problem. Here  we are given a  graph $G$, 
a collection of pairs of distinct non-adjacent terminal vertices   $\{x_1,y_1\},\ldots,\{x_s,y_s\}$ of $G$, and a non-negative integer $k$.
The task is to decide whether there is 
  a set $F\subseteq E(G)$ with $|F|\leq k$ such that for each  $i\in\{1,\ldots,s\}$, graph $G[F]$ (which we can be assumed to be a forest) contains an $(x_i,y_i)$-path.
  The special case   when  $x_1=x_2=\cdots=x_s$, i.e. when edge set $F$ is a tree spanning  all demand vertices, is  the \probSteiner problem.
  To see that \probSteinerF is a special case of  \SSp,  we construct the following graph: For each $i\in\{1,\ldots,s\}$, we add a new edge $x_iy_i$ to $G$. Denote by $G'$ the obtained graph and let $T$ be the set of added edges and let $M(G')$ be the graphic matroid associated with $G'$. Then a set of edges $F\subseteq E(G)$  forms a graph containing all $(x_i,y_i)$-paths if an only if $T\subseteq \spn(F)$
 in $M(G')$.

 Similar to \probSteiner,   \probSteinerF is  fixed-parameter tractable  parameterized by the  number of terminals. This can be shown by applying a dynamic programming algorithm similar to the classical algorithm of  Dreyfus and Wagner \cite{DreyfusW71}. Notice that by Theorem~\ref{thm:paraNP-hard}, 
\spcgraph{} is \classNP-complete when restricted to the instances with $r\leq 2$ and $|T|\leq 2$. This  shows that for our problem the parameterization just by the number of terminals $|T|$ will not work; it also indicates that for matroids we should try a different approach. 
 To  show that  \probSteinerF   is FPT parameterized by the size $k$ of the forest $F$, one can use the following idea. Since the size of $F$ is at most $k$, there are $2^{\OO(k)}$ non-isomorphic forests, so we can 
 guess the structure of $F$. In other words, we can guess a forest $H$ on at most $k$ edges such that the solution $F$ to    \probSteinerF  is isomorphic to $H$. 
Thus for each guess of $H$,  the task is reduced to the following constraint variant of \textsc{Subgraph Isomorphism}: For given graph $G$ and forest $H$, decide whether $G$ contains a forest isomorphic to $H$ and spanning all terminal vertices of $G$ in the prescribed way.  This problem can be solved by combining a color coding technique of Alon, Zwick, and Yuster~\cite{AlonYZ} with dynamic programming. 

This is exactly the approach we want to push forward for  $r>0$. However in this case reduction to constraint  \textsc{Subgraph Isomorphism} is way more difficult. 
First,    while perturbation matrix $P$ is of  bounded rank, adding it to $I(G)$ can change an  unbounded number of its elements.    On the other hand, since the rank of perturbation matrix $P$ is bounded, we know that matrix $P$ contains only a small number of different columns. Thus while 
adding $P$ to $I(G)$ changes many elements of $I(G)$, the variety of these changes is bounded. We exploit this in order to   guess the structure of a solution.
Second,  for graphic matroids, the way a forest $H$ should be mapped into $G$  is very clear---for every terminal element $t$, adding $t$ to the solution should  create a cycle containing $t$. This defines the constrains how the edges of the guessed solution should be connected to terminal edges and allows us to reduce the problem to a constraint variant of \textsc{Subgraph Isomorphism}. For $r>0$,  adding $P$ to $I(G)$ completely destroys this nice property of the solution. Interestingly, the bounded rank of perturbation still allows us to establish the constrains expressed as parities of vertex degrees of a small number of vertices in $G$,  coloring of edges of $G$,  and some additional mappings. As a result,    by a sequence  of   reductions, we succeed in reducing the original problem to a   version of constraint  \textsc{Subgraph Isomorphism}. Due to the nature of constrains, the solution to this problem also requires new ideas on top of color coding and dynamic programming.

\medskip
We proceed with an overview of the proof of Theorem \ref{thmgraphic}. The proof consists of two main parts. The first part is an FPT-Turing reduction from \SSp to the following version of  \textsc{Subgraph Isomorphism}, which we call \pattern.

\defproblemu{\pattern}%
{A (multi) graph $G$ with $n$ vertices and $m$ edges, a non-negative integer $t$ that is a fixed constant, a function $\ell_G: E(G)\rightarrow \{1,2,\ldots,t\}$, a non-negative integer $k$, a forest $H$ with $k$ vertices, a function $\ell_H: E(H)\rightarrow \{1,2,\ldots,t\}$, a set $U\subseteq V(H)$ and an injective function $f: U\rightarrow V(G)$.}%
{
Is there an injective homomorphism 
 $g: V(H)\cup E(H)\rightarrow V(G)\cup E(G)$ such that  (i) for all $e\in E(H)$, it holds that $\ell_H(e)=\ell_G(g(e))$, and (ii) for all $v\in U$, it holds that $g(v)=f(v)$?}

In other words, we give a reduction that for an  input  $(G,P,T,k)$ of  \spcgraph  in time  $k^{\OO(k)}\cdot (n+m)^{\OO(1)}$ constructs $k^{\OO(k)}\cdot (n+m)^{\OO(1)}$ instances of \pattern such that  $(G,P,T,k)$ is a yes-instance if and only if at least one of the instances of \pattern is.

The second part of the proof is an algorithm solving \pattern in time $k^{\OO(k)}\cdot (n+m)^{\OO(1)}$. The combination of the two parts provides the proof of the theorem. 

\medskip 
In what follows, we provide a brief description of the FPT-Turing reduction. 
The reduction is done by a sequence of steps.
For simplicity, here we explain how to construct a reduction in time $2^{\OO(k^2)}\cdot (n+m)^{\OO(1)}$; in Section~\ref{section:graphic} we provide more precise arguments that allow to reduce the running time.

We start by bounding $|T|$ by $k$. 
In case the columns in $T$ are not linearly independent, we let $T'$ denote a basis of $T$, and else we denote $T'=T$. We remove the columns in $T\setminus T'$ from $I(G)$ and $P$, and let $(G',P',T',k)$ denote the resulting instance. Clearly, $(G,P,T,k)$ is a yes-instance if and only if $(G',P',T',k)$ is a yes-instance. Moreover, given a set $X$ of size $t$ of linearly independent vectors, for some $t\in\mathbb{N}$, there does not exist any set $Y$ of vectors of size smaller than $t$ such that $X\subseteq\spanS(Y)$. Thus, in case $|T'|>k$, the input instance is a no-instance. Therefore, from now onwards we implicitly assume that $|T|\leq k$.
We use the term {\em solution} to refer to any set $F\subseteq E\setminus T$ with $|F|\leq k$ such that $T\subseteq \spn(F)$ in the binary matroid $M$ represented by $A$.

Recall that $\disC(P)$ is defined as the {\em set} of the distinct vectors that correspond to the columns in $\{P^e: e\in E(G)\}$. Let us denote $|\disC(P)|=t$. Since the rank of $P$ is $r$, it is easy to see that it has at most  $ 2^r$ different columns, thus  $t\leq 2^r$.
 We say that an edge $e\in E(G)$ is of {\em type} $i$, $1\leq i\leq t$, if $P^e=C^i$ (as vectors). Given an edge $e\in E(G)$, we let $\type(e)$ denote its type. Given a set of edges $E'\subseteq E(G)$, we denote $\type(E',i)=|\{e\in E': \type(e)=i\}|\mod 2$. Towards to constructing the reduction to \pattern, we define $\ell_G\colon E(G)\rightarrow \{1,\ldots,t\}$ by setting $\ell_G(e)=\text{type}(e)$.

 We proceed by identifying a small graph that we can guess, and which will guide us how to find a solution. 
Let $F$ be an inclusion-wise minimal solution;  note that the minimality of $F$ implies that $F$ is an independent set. 
Consider the graph $H=G[\edges(F)]$. The crucial structural lemma that we use states that 
$H$ is ``almost'' a forest. More precisely, we show that $H$ has at most $2^t$ cycles.  
To see it, assume that $H$ has at least $2^t+1$ cycles. 
There are at most $t$ edge types in $H$. Hence
 by the pigeonhole principle, there are distinct sets of edges $C_1$ and $C_2$ of $R$ that 
 compose cycles and such that $\type(C_1,i)=\type(C_2,i)$ for  all $i\in\{1,\ldots, t\}$. Then for the symmetric difference $C=C_1\bigtriangleup C_2$, we obtain that $\type(C,i)=0$ for $i\in\{1,\ldots,t\}$. Thus the sum of the columns of $P$ corresponding to edges of $C$ is the zero-vector.  
  Notice that 
 since $C$ is the union of cycles of $H$, the sum of the columns of matrix $I(G)$ corresponding to its edges is also the zero-vector.  Hence, the sum of the corresponding vectors of $A$ is also zero; and thus the correspondent set of columns of $A$,  
  $\{A^e\mid e\in C\}\subseteq F$ is not independent. But this  contradicts  the minimality of $F$.

Let $\cal H$ denote the set of all non-isomorphic graphs with at most $k$ edges, at most $2^t$ cycles, and no isolated vertices. 
Thus  $(G,P,T,k)$ is a yes-instance of \spcgraph  if and only if  $(G,P,T,k)$ has a solution isomorphic to some    $H\in {\cal H}$.
It is possible to show that all non-isomorphic graphs in 
$\cal H$ can be enumerated within time  $ 2^{\OO(k)}$. 
Therefore, we may explicitly examine each graph $H\in \cal H$ and check whether we have a solution 
$F$ with subgraph of $G$, $G[\edges(F)]$, isomorphic to $H$.  In other words, we are looking for an injective homomorphism 
  $g: V(H)\cup E(H)\rightarrow V(G)\cup E(G)$\footnote{Since we handle {\em multi} graphs, we define the domain and image of $g$ to include edge-sets.} 
 such that $F=\{A^e\mid e\in g(E(H))\}$ is a solution. 
  This is an FPT-Turing reduction which reduces in time $ 2^{\OO(k)}$ the solution of  the original problem to the solution of  $ 2^{\OO(k)}$ new problems. We will use a less formal term \emph{guess} to refer to such type of reductions.
 So we guess graph $H$.

Next, we observe that we can guess the types of edges of $H$. Since $H$ has at most $k$ edges, there are at most $t^k=2^{\OO(k)}$ distinct functions $\ell_H\colon E(H)\rightarrow\{1,\ldots,t\}$. 
Then for 
each
 guess of function $\ell_H$, we want to decide  whether there 
 is an injective homomorphism $g$ such that $\ell_G(g(e))=\ell_H(e)$ for every $e\in E(H)$ and such that the set of columns $F$ of $A$ corresponding to the  image of $g$, which is  $F=\{A^e\mid e\in g(E(H))\}$,  is a solution. 

By definition, if $F=\{A^e\mid e\in g(E(H))\}$ is a solution, then  for each $W\in T$, there is $F_W\subseteq F$ such that 
\begin{equation}\label{eq:spnW}
W=\sum_{e\in F_W}A^e.
\end{equation}
(The summations here are modulo $2$.)
We denote by $E_W=g^{-1}(\edges(F_W))$ the   edge subset of  $H$ corresponding to $F_W$. Then by \eqref{eq:spnW}, 
\begin{equation*}\label{eq:spnWpr1}
W=\sum_{e\in g(E_W)}(I^e(G)+P^e)=\sum_{e\in g(E_W)}I^e(G)+\sum_{e\in g(E_W)}P^e .
\end{equation*}
Each column $P^e$ is equal to   vector $C^{\ell_H(e)}$  from partition $\disC(P)$. 
Thus 
 \begin{equation}\label{eq:spnWpr}
W= \sum_{e\in g(E_W)}I^e(G)+\sum_{e\in E_W}C^{\ell_H(e)}.
\end{equation}

Let $W'=W+\sum_{e\in E_W}C^{\ell_H(e)}$. 
 The rows of matrix $I(G)$ and thus the elements of $W'$ are indexed by  the vertices of $G$. For $v\in V(G)$, we denote by 
 $w_v$ the element of $W'$ indexed by $v$. Note that $w_v$ is either $0$ or $1$.  Let $V_{W}=\{v\in V(G)\mid w_v=1\}$. Observe that $V_W$ is uniquely defined by the choice of  $W$ and $E_W$.
 The crucial insight, which proof is given in Section~\ref{section:graphic}, that 
 \eqref{eq:spnWpr} and, therefore, \eqref{eq:spnW} holds if and only if $g$ acts as a bijection between $V_W$ and  vertices of $H[E_W]$ of odd degrees.
 This is the most important part of the reduction;  it allows to reduce the algebraic requirement  that every terminal vector should be in the span of the solution to 
 constrains in the form of  bijections, which can be guessed efficiently.

We exploit this property for the next set of  guesses. For each $W\in T$, we guess a set $E_W\subseteq E(H)$ and construct $V_W$ as described above. Since $|T|\leq k$ and $|E(H)|\leq k$, we have at most $2^{k^2}$ possible choices of the sets $E_W$. 
Then we find the set $U_W\subseteq V(H[E_W]))$ of vertices that have odd degrees in  $H[E_W]$. If $|V_W|\neq|U_W|$, we discard the choice. Otherwise, we set $U=\cup_{W\in T}U_W$. Notice that if our guesses correspond to a (potential) solution $F$, we have that corresponding injective homomorphism $g$ should map $U$ to $V'=\cup_{W\in T}V_W$ bijectively and, moreover, $g$ should act as bijection between each $U_W$ and $V_W$. We make all possible guesses  of  a bijection $f\colon U\rightarrow U'$.
Since $|U|\leq 2k$, we have at most $(2k)^{2k}$ possible choices.
Then for each $U$ and $f$, we are searching for  
an  injective homomorphism 
 $g: V(H)\cup E(H)\rightarrow V(G)\cup E(G)$ such that  (i) for all $e\in E(H)$,  $\ell_H(e)=\ell_G(g(e))$, and (ii) for each $v\in U$,   $g(v)=f(v)$.

Now we are ready for the final step of our reduction. Recall that  $H$ in the statement of \pattern is required to be a forest. The graph $H$ that was guessed so far does not have   this property, but it is  ``almost'' a forest, that is, it has at most $2^t$ cycles. To fix it, we guess a set of edges $S\subseteq E(H)$ of size at most $2^t$ such that the graph obtained from $H$ by the deletion of $S$ is a forest and set $H=H-S$. Since $|S|\leq 2^t$, and $t$ is a constant depending on $r$ only, 
we can make a polynomial number of  guesses how solution $g$ could map $S$ to $E(G)$; we have at most $|E(G)|^{2^t}=m^{\OO(1)}$ possibilities for such partial mappings. For each guess of mapping  $h: S\rightarrow V(G)$, we modify $U$ and $f$ respectively. Namely, we set $U=U\cup V(H[S])$ and define $f(v)=h(v)$ for $v\in V(H[S])$ as prescribed by our choice of the mapping $h$ of $S$.

This concludes the description of the construction of an instance of \pattern. It is possible to show that
 $(G,P,T,k)$ is a yes-instance of  \spcgraph if and only if for at least one of the described guesses of a
   forest $H$,   functions $\ell_H$, $\ell_G$,    set $U\subseteq V(H)$ and  function $f: U\rightarrow V(G)$, the instance  of \pattern with these parameters is a yes-instance.
Since the total number of guesses we make is 
$2^{\OO(k^2)}\cdot (n+m)^{\OO(1)}$, 
our construction is the required FPT-Turing reduction.

\medskip

In order to solve \pattern, and to complete
the proof of Theorem~\ref{thmgraphic},  we still have to solve \pattern. This is done by 
  a non-trivial application of  the color coding technique    combined with dynamic programming. We postpone all the details till Section~\ref{section:graphic}.

\subsection{Duals of Perturbed Graphic Matroids}\label{subsec:dual}
In this section, we give an overview of the proof of our second main  result. The detailed  proof of the theorem is given in Section~\ref{section:dual}. 

Formally, we define the following problem. 
\medskip

\defproblemu{\spcduallong (\spcdual)}%
{A (multi) graph $G$ with $n$ vertices and $m$ edges, an $(n\times m)$-matrix $P$ over $GF(2)$ with $\rank(P)\leq r$, a set of 
\emph{terminals} $T\subseteq E$ where $E$ is the set of columns of the matrix $A=I(G)+P$,
and a non-negative integer $k$.}%
{Is there a set $F\subseteq E\setminus T$ with $|F|\leq k$ such that $T\subseteq \spn(F)$ in the dual $M^*$ of the binary matroid $M$ represented by $A$?}

\smallskip

\begin{theorem}
\label{thmdual}
\spcdual{} is solvable in time $2^{2^{\OO((2^r+k^2)k)}}\cdot (n+m)^{\OO(1)}$. In particular, \spcdual{} is \classFPT{} when parameterized by $r+k$. 
\end{theorem}
As in the case with graphic matroids, it is useful to recall how \spcdual 
   is solvable for $r=0$, i.e. on cographic matroids. In a cographic matroid a circuit corresponds to a cut in the underlying graph $G$. 
   In this case the solution set $F$  should satisfy the following property: for every terminal element 
     $e\in T$ there is a partition (or a cut) $(X_e, \overline{X}_e)$ of the vertex set of $G$ such that this cut, i.e. the set of  edges between   $X_e $ and $\overline{X}_e$,
     is of the form $\{e\}\cup F_e$, where $F_e\subseteq F$. Thus $e$ is the only edge in the cut from $T$ and all other edges are from $F$.
  In graph theory this problem is known under name 
  \textsc{Edge Subset Feedback Edge Set}.   Xiao and Nagamochi~\cite{XiaoN12}  showed that this  problem is FPT parameterized by $k=|F|$.  
The algorithm for solving \textsc{Edge Subset Feedback Edge Set}, as well as its special case \textsc{Multiway Cut}, uses
the technique   of Marx  based on important separators~\cite{Marx06}.  The essence of this technique is that all required information about the cuts in a graph can 
be extracted from a carefully selected set of separators of size at most $k$. 
 However, we do not see how this approach can be shifted to more general matroids, even when the rank of perturbation matrix is $1$. The difficulty  in this case is that 
solution $F$ together with $T$ cannot be represented as the union of the sets of edges of cuts in $G$ anymore, and thus  
the sizes of important separators in $G$ cannot be bounded by a function of $k$ only. In order to overcome this challenge,  we have to apply   more powerful method of recursive understanding \cite{DBLP:journals/siamcomp/ChitnisCHPP16}.

On a general level, the structure of the proof of Theorem~\ref{thmdual} is similar to the structure of the proof of  
 Theorem~\ref{thmgraphic}. It consists of two parts. In the first part we give FPT-Turing reduction to a  cut problem on graphs and in the second part we use the method of recursive understanding to solve the problem. But here the similarities end. While on perturbation of graphic matroids \SSp is about subgraph isomorphisms, on perturbation of cographic matroids it is about collections of  cuts in graphs. This makes 
 both parts of the proof  of   Theorem~\ref{thmdual} much more challenging than in Theorem~\ref{thmgraphic}. 
  In order to introduce the graph-cut problem we reduce to, we need several definitions.

 \smallskip
\noindent 
{\bf Graph problem.} 
Let  $G$ be a graph with $n$ vertices and $m$ edges given together with  a set of terminal edges $T$ and a partition of  $V(G)=(V_1,V_2,\ldots,V_t)$.
In addition,   for every   $e\in T$ graph $G$ is provided with  a function $f_e: E(G)\rightarrow \{0,1\}$ and a binary vector $B^e=(b^e_1,b^e_2,\ldots,b^e_t)$.

For terminal edge $e\in T$ and a partition $(X,\bar{X})$ of $V(G)$, we say that an edge $e'\in E(G)$ {\em contributes to $(e,(X,\bar{X}))$} (with respect to $f_e$) if one of the  following conditions holds
\begin{enumerate}
\setlength{\itemsep}{-2pt}
\item Both endpoints of $e'$ belong to $X$ and $f_{e}(e')=1$.
\item Both endpoints of $e'$ belong to $\bar{X}$ and $f_{e}(e')=1$.
\item Exactly one of the endpoints of $e'$ belongs to $X$ and $f_{e}(e')=0$.
\end{enumerate}
Accordingly, we define $\cont(e,X)$ as the set of edges that contribute to $(e,(X,\bar{X}))$.

For partition $(X,\bar{X})$   of $V(G)$, and terminal edge $e\in T$, 
we say that $(X,\bar{X})$ {\em almost fits $e$} (with respect to $f_e$)
 if $T\cap\cont(e,X)=\{e\}$. Moreover, if $(X,\bar{X})$ almost fits $e$ and for all $1\leq i\leq t$, it holds that $|X\cap V_i|=b^e_i\mod 2$, then we say that $(X,\bar{X})$ fits $e$ (with respect to $f_e$ and $B^e$).

We are now ready to define our graph problem.

\defproblemu{\edc}%
{A (multi) graph $G$ with $n$ vertices and $m$ edges, non-negative integers $k$ and $t$, a partition $(V_1,V_2,\ldots,V_t)$ of $V(G)$, 
a set $  T\subseteq E(G) $, a binary vector $B^e=(b^e_1,b^e_2,\ldots,b^e_t)$ for $e\in T$, and a function $f_e: E(G)\rightarrow \{0,1\}$ for $e\in T$.}%
{Is there a set $F\subseteq E(G)\setminus T$ with $|F|\leq k$ such that for each $e\in T$, there exists a partition $(X_e,\bar{X}_e)$ of $V(G)$ that fits $e$ and such that $\cont(e,X_e)\setminus\{e\}\subseteq F$?}

In other words,  we a looking for a set of edges $F$ of size $k$, such that for every terminal edge $e$, there is a cut  $(X_e,\bar{X}_e)$ such that (i) the parities of the intersections of $X_e$ with sets $V_i$ constitute vector $B^e$, (ii) $e$ is the only terminal edge contributing to the cut and all other edges contributing to the cut are from $F$.

In the first part of the proof we  give a reduction that for an  input  $(G,P,T,k)$ of  \spcdual  in time   
$2^{\OO(k2^r)}\cdot (n+m)^{\OO(1)}$ constructs  $2^{\OO(k2^r)}\cdot (n+m)^{\OO(1)}$ 
instances of \edc such that  $(G,P,T,k)$ is a yes-instance if and only if at least one of the instances of \edc is.

\medskip
As in the case of perturbed graphic matroids, we can assume that $|T|\leq k$. Recall that $\disR(P)$ is the set of the distinct vectors corresponding the rows of $P$, and denote $|\disR(P)|=t$. 
  Since the rank of $P$ is $r$, it  has at most  $ 2^r$ different rows, hence  $t\leq 2^r$. 
 Moreover, denote $\disR(P)=\{R_1,R_2,\ldots,R_t\}$. Accordingly, we say that a vertex $v\in V(G)$ is of {\em type} $i$, $1\leq i\leq t$, if $P_v=R_i$. Given a vertex $v\in V(G)$, we let $\type(v)$ denote its type. For $i\in\{1,\ldots,t\}$, we denote by $V_i$ the set of vertices of type $i$.

\smallskip
\noindent
{\bf Characterization of solutions.}
For \spcdual, we use the term {\em solution} to refer to a set $F\subseteq E\setminus T$ with $|F|\leq k$ such that $T\subseteq \spn(F)$ in the dual $M^*$ of the binary matroid $M$ represented by $A$. Let $I$ be a binary vector with $m$ elements. Recall that given $F\subseteq E$, $\edges(F)$ denotes the set of all edges $e\in E(G)$ such that $A^e\in F$. Now, given a set $F\subseteq E$, we say that $I$ is the {\em characteristic vector} of $F$ if the $i^\mathrm{th}$ entry of $I$ is 1 if and only if $F$ contains the  $i^\mathrm{th}$ column of $A$. Moreover, a set $F\subseteq E$ is a {\em cocyle} in $M$ if and only if it is a cycle in ${M}^*$. We need the following folklore result (see, e.g., \cite{GeelenK15}) characterizing cocycles of binary matroids.

\begin{proposition}\label{lem:cocyclesmall}
Let $M$ be a binary matroid represented by an $(n\times m)$-matrix $A$, and let $F$ be a subset of $E$, where $E$ is the set of columns of $A$. Then, $F$ is a cocycle in $M$ if and only if the characteristic vector of $F$ belongs to $\spanS(V)$, where $V$ is the set of rows of $A$.
\end{proposition}

Note that a set $F\subseteq E\setminus T$ is a solution if and only if for each terminal $W\in T$, there exists a subset $F_W\subseteq F$ such that $F_W\cup\{W\}$ is a cocycle in $M$. Thus, in light of Proposition~\ref{lem:cocyclesmall}, we can think of a solution as follows:

\begin{observation}\label{lem:incidencesmall}
A set $F\subseteq E\setminus T$ is a solution if and only if $|F|\leq k$ and for each terminal $W\in T$, there exists a subset $F_W\subseteq F$ such that the characteristic vector of $F_W\cup\{W\}$ belongs to $\spanS(V)$, where $V$ is the set of rows of $A$.
\end{observation}

Let $F$ be a solution. For each $W\in T$, denote by $e(W)$ the edge of $G$ corresponding to the terminal $W$. 
By Observation~\ref{lem:incidencesmall}, for each $W\in T$, there is $F_W\subseteq F$ such that the characteristic vector $I_W$ of $F_W'=F_W\cup\{W\}$ belongs to $\spanS(V)$.
It means that there is a set of vertices $X_{e(W)}\subseteq V(G)$ such that $I_W=\sum_{v\in X_{e(W)}}A_v$. Hence, for each $W\in T$, we have the corresponding partition $(X_{e(W)},\bar{X}_{e(W)})$ of $V(G)$, and the solution can be represented as a collection of cuts $\{(X_{e(W)},\bar{X}_{e(W)})\mid W\in T\}$ of $G$.

For each $W\in T$ and $i\in \{1,\ldots,t\}$, we guess the parity of $|X_{e(W)}\cap V_i|$ and define the vector $B^{e(W)}=(b^{e(W)}_1,\ldots,b^{e(W)}_t)$ respectively by setting $b^{e(W)}_i=|X_{e(W)}\cap V_i|\mod 2$. Notice that we have at most $2^{tk}$ choices for $B^{e(W)}$, because $|T|\leq k$. For each guess, we now looking for a solution represented by a 
collection of cuts $\{(X_{e(W)},\bar{X}_{e(W)})\mid W\in T\}$ of $G$ such that $|X_{e(W)}\cap V_i| \mod 2 =b^{e(W)}_i$ for $W\in T$ and $i\in\{1,\ldots,t\}$.

Let  $I_W=\sum_{v\in X_{e(W)}}A_v$ and denote by $i^{W}_e$ for $e\in E(G)$ the elements of $I_W$. 
We have that 
\begin{equation}\label{eq:cv}
I_W=\sum_{v\in X_{e(W)}}(I_v(G)+P_v)=\sum_{v\in X_{e(W)}}I_v(G)+\sum_{v\in X_{e(W)}}P_v.
\end{equation}
Let $P_W=\sum_{v\in X_{e(W)}}P_v$.
Since  $|X_{e(W)}\cap V_i| \mod 2 =b^{e(W)}_i$ for $W\in T$ and $i\in\{1,\ldots,t\}$, we obtain that 
$P_W=\sum_{i=1}^tb_i^{e(W)} R_i$. Notice  that   vector $P_W$ is uniquely defined by the  choice of $B^{e(W)}$.
We define $f_{e(W)}\colon E(G)\rightarrow\{0,1\}$, by setting $f_{e(W)}(e)$ to be equal to the element of $P_W$ corresponding to $e$.

Recall that $I_W$ is the characteristic vector of the cocycle $F_W'$. It means that $A^e\in F_W'$ if and only if $i^W_e=1$. By  %the properties of the incidence matrix of a graph and
  making use of \eqref{eq:cv}, we are able to show that for each edge $e\in E(G)$, $A^e\in F_W'$ if and only if one of the following holds:
\begin{itemize}
\setlength{\itemsep}{-2pt}
\item Both endpoints of $e$ belong to $X_{e(W)}$ and $f_{e(W)}(e)=1$.
\item Both endpoints of $e$ belong to $\bar{X}_{e(W)}$ and $f_{e(W)}(e)=1$.
\item Exactly one of the endpoints of $e$ belongs to $X_{e(W)}$ and $f_{e(W)}(e)=0$.
\end{itemize}

We have that $W\in F_W'$ and $W$ is the unique element of $T$ in this set. It means that for   edge $e(W)$, cut $(X_{e(W)},\bar{X}_{e(W)})$ almost fits $e(W)$ with respect to $f_{e(W)}$. Since 
$|X_{e(W)}\cap V_i| \mod 2 =b^{e(W)}_i$ for each $W\in T$ and $i\in\{1,\ldots,t\}$, we have that $(X_{e(W)},\bar{X}_{e(W)})$  fits $e(W)$ with respect to $f_{e(W)}$ and $B^{e(W)}$.
Moreover, we prove that  for  each $W\in T$ and $i\in\{1,\ldots,t\}$ the following are equivalent
\begin{itemize}
\setlength{\itemsep}{-2pt}
\item  $F_W\rq{}$ is a cocycle of $M$ such that $|X_{e(W)}\cap V_i| \mod 2 =b^{e(W)}_i$ and such that 
  its characteristic vector is expressible as  $I_W=\sum_{v\in X_{e(W)}}A_v$;
  \item  
Cut $(X_{e(W)},\bar{X}_{e(W)})$  fits $e(W)$ with respect to $f_{e(W)}$ and $B^{e(W)}$.
\end{itemize}
We also  have that 
$F_W\rq{}= \cont(e(W),X_{e(W)})$.

Now we can complete the reduction to \edc. We consider the partition $(V_1,\ldots,V_t)$ of $V(G)$ and the set of terminal edges $T_G=\{e(W)\mid W\in T\}$. For each $W\in T$, we have a binary vector $B^{e(W)}$ and  a function $f_{e(W)}$. Together, all these parameters compose an instance of \edc.

\smallskip
\noindent
{\bf Solving \edc.} The algorithm for \edc is the most technical part of the paper. Here we  briefly highlight the approach.  On a high-level, we use the method of recursive understanding \cite{DBLP:journals/siamcomp/ChitnisCHPP16}, in which we incorporate various new, delicate subroutines. Informally, this means that at the basis, we are going to deal with a ``highly-connected'' or a small graph, and at each step where our graph is not highly-connected, we will break it using a very small number of edges into two graphs that are both neither too small nor too large.

Let $G$ be a connected graph, and let $p$ and $q$ be positive integers. A partition $(X,Y)$ of $V(G)$ is called {\em $(q,p)$-good edge separation}  if
  $|X|,|Y|>q$,  $|E(X,Y)|\leq p$, and 
 $G[X]$ and $G[Y]$ are connected graphs.

Roughly speaking, a graph $G$ is unbreakable if every partition of $V(G)$ with few edges going across must contain a large chunk of $V(G)$ in one of its two sets. Intuitively, this means that $G$ is ``highly-connected'': any attempt to ``break'' it severely by using only few edges is futile. Formally,  
A graph $G$ is {\em $(q,p)$-unbreakable} if it does not have a {\em $(q,p)$-good edge separation}.

If a graph $G$ is not $(q,p)$-unbreakable, we say that it is $(q,p)$-breakable. Chitnis et al.~\cite{DBLP:journals/siamcomp/ChitnisCHPP16} proved the following result.

\begin{proposition}[\cite{DBLP:journals/siamcomp/ChitnisCHPP16}]\label{prop:goodSep-p}
There exists a deterministic algorithm that given a connected graph $G$ along with integers $q$ and $p$, in time $\OO(2^{\min\{q,p\}\cdot\log(q+p)}\cdot (n+m)^3\log(n+m))$ either finds a $(q,p)$-good edge separation, or correctly concludes that $G$ is $(q,p)$-unbreakable.
\end{proposition}

In our case, we set $p=2(k+1)$ and $q=2^{2^{\lambda(t+k^2)|T|}}$ for some appropriate constant  $\lambda$. 
To apply the method of recursive understanding, we introduce a special variant of \edc called \aedc (see Section~\ref{section:dual} for the formal definition) that is tailored to apply recursion. We show that we can assume that the input graph $G$ is connected.  
If $G$ has bounded (by some function of $r$ and $k$) size, we solve \aedc directly.  Otherwise, we use Proposition~\ref{prop:goodSep-p} to check whether $G$ is $(q,p)$-unbreakable. 

If $G$ is not $(q,p)$-unbreakable, we find a $(q,p)$-good separation $(X,Y)$ of $G$. Then we solve a special  instance of \aedc for one of the graphs $G[X]$ and $G[Y]$ recursively. We use the obtained solution to construct a new instance of the problem for a graph $G'$ that has less vertices than $G$. Then we call our algorithm for this smaller instance.  

If $G$ is $(q,p)$-unbreakable, we obtain the crucial basic case that we briefly discuss here. For simplicity, we consider this case for \edc. 

Recall that in the definition of \edc, we ask about a set $F\subseteq E(G)\setminus T$ with $|F|\leq k$ such that for each $e\in T$, there exists a partition $(X_e,\bar{X}_e)$ of $V(G)$ that fits $e$ and such that $\cont(e,X_e)\setminus\{e\}\subseteq F$. We relax these conditions and look for a collection of partitions $\{(Y_e,\bar{Y}_e)\mid e\in T\}$ such that $(Y_e,\bar{Y}_e)$ almost fits $e$  and 
 $|\cont(e,Y_e)\setminus\{e\}|\leq k$ for $e\in T$. Then we can find such an auxiliary collection of partitions $\{(Y_e,\bar{Y}_e)\mid e\in T\}$ by reducing the relaxed problem to at most $k$ instances of  the {\sc Edge Odd Cycle Transversal} problem (also known as {\sc Edge Bipartization}). The latter problem could be solved by the results of  Guo et al.~\cite{GuoGHNW06}. 
Finally we  use  auxiliary  partitions  $\{(Y_e,\bar{Y}_e)\mid e\in T\}$ to construct the required 
 collection of partitions   
$\{(X_e,\bar{X}_e)\mid e\in T\}$ and a set $F$ of size at most $k$.  The final construction heavily exploit the  high connectivity of $G$  which allows to search only a ``small neighborhood'' of 
 $(Y_e,\bar{Y}_e)$.

\section{Preliminaries}\label{sec:prelims}

For standard graph and matroid-related terms that are not explicitly defined here, we refer to the books by Diestel \cite{Diestel} and Oxley \cite{oxley2006matroid}. Given two sets $A$ and $B$, we denote their symmetric difference by $A\bigtriangleup B$, that is, the set consisting of every element that is present in either $A$ or $B$, but not in both.

\paragraph{Graphs} Given a graph $G$, we let $V(G)$ and $E(G)$ denote its vertex set and edge set, respectively. Given a set of edges $E\subseteq E(G)$, we let $G[E]$ denote the subgraph of $G$  whose vertex-set consists of the endpoints of the edges in $E$ and whose edge-set is $E$. Moreover, $G-E$ denote the subgraph of $G$ obtained by deleting the edges in $E$ from $G$. Given a set of vertices $V\subseteq V(G)$, we let $G[V]$ denote the graph induced by $V$, that is, the graph whose vertex-set is $V$ and whose edge set consists of the edges in $E(G)$ with both endpoints in $V$. Given two vertices $u,v\in V(G)$, a {\em $(u,v)$-path} is a path in $G$ from $u$ to $v$.We say that two graphs, $G$ and $H$, are {\em isomorphic} is there exists a bijective function $f: V(G)\rightarrow V(H)$ such that $uv\in E(G)$ if and only if $f(u)f(v)\in E(H)$.

\paragraph{Vectors and Matrices.} Let $F$ be a set of vectors $V_1,V_2,\ldots,V_n$. We say that $F$ is {\em linearly dependent} over a field $\mathbb{F}$ if there exist $\lambda_1,\lambda_2,\ldots,\lambda_n\in\mathbb{F}$, not all equal $0$, such that $\sum_{i=1}^n\lambda_iV_i=0$. If $F$ is not linearly dependent, then it is called {\em linearly independent}. The {\em linear span of $F$} is the set of all finite linear combinations of the vectors in $F$, that is, $\spn(F)=\{\sum_{i=1}^n\lambda_i V_i: \lambda_1,\ldots,\lambda_n\in\mathbb{F}\}$. A {\em basis} of $\spn(F)$ is a subset $F'\subseteq \spn(F)$ that is linearly independent and such that $\spn(F')=\spn(F)$. The {\em rank} of a matrix is equal to the maximum number of linearly independent columns of the matrix (which is equal to the maximum number of linearly independent rows of the matrix). 

Let $G$ be a graph $G$. The {\em incidence matrix} of $G$, denoted by $I(G)$, is the binary matrix such that rows of $I(G)$ correspond to the vertices of $G$ and the columns of $I(G)$ correspond to the edges of $G$, and for all $v\in V(G)$ and $e\in E(G)$, the entry of $I(G)$ indexed by $v$ and $e$ is $1$ if $v$ is incident to $e$ and $0$ otherwise. When $G$ is clear from context, given any $(|V(G)|\times|E(G)|)$-matrix $P$, we index the elements of $P$ by the vertices and edges of $G$, and denote the elements by $p_{v,e}$ for $v\in V(G)$ and $e\in E(G)$. Moreover, we index the rows of $P$ by vertices and denote them by $P_v$ for $v\in V(G)$, and we index the columns of $P$ by edges and denote them by $P^e$ for $e\in E(G)$. The same notation is used also if the matrix is denoted by a different letter. For example, for a $(|V(G)|\times|E(G)|)$-matrix $A$, we denote the elements by $a_{v,e}$ for $v\in V(G)$ and $e\in E(G)$. Given a $(|V(G)|\times|E(G)|)$-matrix $A$ whose column set is denoted by $E$ and a subset $F\subseteq E$, we let $\edges(F)$ denote the set edges $e\in E(G)$ such that $A^e\in F$.

Given a $(|V(G)|\times|E(G)|)$-matrix $P$, we define $\disC(P)$ as the set of the distinct vectors that correspond to the columns in $\{P^e: e\in E(G)\}$. We stress that here, $\disC(P)$ is a set and not a multiset, and thus it contains each vector $W$ for which there exists $e\in E(G)$ such that $W=P^e$ exactly once (even if there exists more than one $e\in E(G)$ such that $V=P^e$). Note that generally, when we refer to a set $F$ of columns, it might contain several columns that correspond to the same vector (in other words, we treat columns as indexed vectors, and thus two columns that have different indices, although they might correspond to the same vector, are assumed to be distinct). Similarly, we define $\disR(P)$ as the set of the distinct vectors corresponding the {\em rows} in $\{P_v: v\in V(G)\}$.

\paragraph{Matroids} Let us begin by defining the notion of a matroid.

\begin{definition}\label{def:matroid}
A pair $M=(E,{\cal I})$, where $E$ is a ground set and $\cal I$ is a family of subsets (called {\em independent sets}) of $E$, is a {\em matroid} if it satisfies the following conditions, called {\em matroid axioms}:
\begin{enumerate}
\item[\rm (I1)]  $\phi \in \cal I$. 
\item[\rm (I2)]  If $A' \subseteq A $ and $A\in \cal I$ then $A' \in  \cal I$. 
\item[\rm (I3)] If $A, B  \in \cal I$  and $ |A| < |B| $, then there is $ e \in  (B \setminus A) $  such that $A\cup\{e\} \in \cal I$.
\end{enumerate}
\end{definition}
An inclusion-wise maximal set in $\cal I$ is a {\em basis} of the matroid $M$. Using axiom (I3) it is easy to show that all the bases of a matroid  are of  the same size. This size is called the {\em rank} of the matroid $M$, and it is denoted by $\rank(M)$. The {\em rank} of a subset $A\subseteq E$ is the maximum size of an independent set contained in $A$.
An inclusion-wise minimal non-independent set is called a \emph{circuit}. A \emph{cycle} is either the empty set or the union of pairwise disjoint circuits.

Let $A$ be a matrix over an arbitrary field $\mathbb F$, and let $E$ be the set of columns of $A$. We associate a matroid $M=(E,{\cal I})$ with $A$ as follows. A set $X \subseteq E$ is independent (that is, $X\in \cal I$) if the columns  in $X$ are linearly independent over $\mathbb F$. The matroids that can be defined by such a construction are called {\em linear matroids}, and if a matroid can be defined by a matrix $A$ over a field $\mathbb F$, then we say that the matroid is representable over $\mathbb F$. That is, a matroid $M=(E,{\cal I})$ of rank $d$ is representable over a field $\mathbb F$ if there exist vectors in $\mathbb{F}^d$  corresponding to the elements in $E$ such that  linearly independent sets of vectors  correspond to independent sets of the matroid. Then, a matroid $M=(E,{\cal I})$  is called {\em representable} or {\em linear} if it is representable over some field~$\mathbb F$. 
   
Given a graph $G$, a \emph{graphic matroid} $M=(E,{\cal I})$ is defined by setting the elements in $E$ to be the edges of $G$ (that is, $E=E(G)$), where $F\subseteq E(G)$ is in $\cal I$ if $G[F]$ is a forest.  The graphic matroid is representable over any field of size at least $2$. In particular,  consider the incidence matrix $I(G)$  of $G$ over $GF(2)$. This is precisely a representation of the graphic matroid of $G$ over $GF(2)$~\cite{oxley2006matroid}.

The {\em dual} of a matroid $M=(E,{\cal I})$ is the matroid $M^\star=(E,{\cal I}^\star)$ where a subset $A\subseteq E$ belongs to ${\cal I}^\star$ if and only if $M$ has a basis disjoint from $A$. Alternatively, the dual of a matroid $M$ is the matroid $M^\star$ whose basis sets are the complements of the basis sets of $M$. Note that duality is closed under involution, that is, $(M^{\star})^{\star}=M$. Bases, circuits and cycles of $M^\star$ are called \emph{cobases}, \emph{cocircuits} and \emph{cocycles} of $M$ respectively. 

To give the definition of a co-graphic matroid based on duality, let us consider a graph $G$. The {\em co-graphic matroid} corresponding to $G$, which is the dual of the graphic matroid corresponding to $G$, is defined as the matroid $M=(E,{\cal I})$ where $E=E(G)$ and ${\cal I} = \{F \subseteq E~\mid~G-F\ \mbox{is connected}\}$.

\paragraph{Parameterized Complexity.} A {\em parameterization} of a problem is the association of an integer $k$ with each input instance, which results in a {\em parameterized problem}. Let us first present the notion of {\em fixed-parameter tractability (\classFPT)}. Here, a parameterized problem $\Pi$ is said to be \classFPT{} if there is an algorithm that solves it in time $f(k)\cdot |I|^{\OO(1)}$, where $|I|$ is the size of the input and $f$ is a function that depends only on $k$. Such an algorithm is called a {\em parameterized algorithm}. In other words, the notion of \classFPT{} signifies that it is not necessary for the combinatorial explosion in the running time of an algorithm for $\Pi$ to depend on the input size, but it can be confined to the parameter $k$.
Parameterized Complexity also provides tools to refute the existence of parameterized algorithms for certain problems (under plausible complexity-theoretic assumptions), in which context the notion of \classW{1}-hard is a central one. It is widely believed that a problem that is \classW{1}-hard is unlikely to be \classFPT, and we refer the reader to the books \cite{DowneyFbook13,CyganFKLMPPS15} for more information. 

\section{FPT on Perturbed Graphic Matroids}
\label{section:graphic}
In this section, we prove Theorem~\ref{thmgraphic}. 
The proof  is done by a chain of statements, each adding constraints on the structures of the solutions we seek. Here, we give the technical details of the approach described in Section \ref{sec:overview}.

\subsection{Bounding $|T|$}

In case the columns in $T$ are not linearly independent, we let $T'$ denote a basis of $T$, and else we denote $T'=T$. We remove the columns in $T\setminus T'$ from $I(G)$ and $P$, and let $(G',P',T',k)$ denote the resulting instance. Clearly, $(G,P,T,k)$ is a yes-instance of \spcgraph{} if and only if $(G',P',T',k)$ is a yes-instance of \spcgraph{}. Moreover, given a set $X$ of size $t$ of linearly independent vectors, for some $t\in\mathbb{N}$, there does not exist any set $Y$ of vectors of size smaller than $t$ such that $X\subseteq\spanS(Y)$. Thus, in case $|T'|>k$, the input instance is a no-instance. Therefore, we have the following observation. 

\begin{observation}\label{obs:boundT}
Given an instance $(G,P,T,k)$ of \spcgraph{}, an equivalent instance of $(G',P',T',k)$ \spcgraph{} such that $|T'|\leq k$ can be computed in polynomial time.
\end{observation}

From now on, we implicitly assume that $|T|\leq k$.

\subsection{Assigning Colors to Edges}

Note that we can assume w.l.o.g.~that the columns of $A$ outside $T$ are distinct. We remind that given $F\subseteq E$ (recall that $E$ is the set of columns of $A$), we let $\edges(F)$ denote the set of each edge $e\in E(G)$ such that $A^e\in F$.

We start by associating a color with each column of $P$. To this end, we need the following folklore result.

\begin{proposition}[folklore]\label{lem:rank}
Let $A$ be a binary matrix. Then, the number of distinct rows in $A$ is bounded by $2^{\rank(A)}$. Similarly, the number of distinct columns in $A$ is bounded by $2^{\rank(A)}$.
\end{proposition}

Recall that $\disC(P)$ is defined as the {\em set} of the distinct vectors that correspond to the columns in $\{P^e: e\in E(G)\}$.
Let us denote $|\disC(P)|=t$. By Proposition \ref{lem:rank}, it holds that $t\leq 2^r$. Moreover, denote $\disC(P)=\{C^1,C^2,\ldots,C^t\}$. Accordingly, we say that an edge $e\in E(G)$ is of {\em type} $i$, $1\leq i\leq t$, if $P^e=C^i$ (as vectors). Given an edge $e\in E(G)$, we let $\type(e)$ denote its type. Given a set of edges $E'\subseteq E(G)$, we denote $\type(E',i)=|\{e\in E': \type(e)=i\}|\mod 2$.

\subsection{Backbone}

We remark that here, we use the term {\em solution} to refer to any set $F\subseteq E\setminus T$ with $|F|\leq k$ such that $T\subseteq \spn(F)$ in the binary matroid $M$ represented by $A$. We proceed by identifying a small graph that we can guess, and which will guide us how to find a solution as it will form its ``backbone''. We remind that given a set of edges $E'\subseteq E(G)$, $G[E']$ denotes the graph whose vertex-set contains all endpoints of the edges in $E'$ and whose edge-set is $E'$.

\begin{lemma}\label{lem:fewCyc}
If the input instance is a yes-instance, then there exists a solution $F$ such that $G[E_F]$, where $E_F=\edges(F)$, has at most $2^t$ cycles.\footnote{In the context of graphs, we use the term cycle to refer to a simple cycle, where two parallel edges are also considered to be a cycle.}
\end{lemma}

\begin{proof}
Suppose that the input instance is a yes-instance, and let $F$ be a {\em minimal} solution (i.e., no proper subset of $F$ is a solution). Suppose, by way of contradiction, that $G[E_F]$ contains more than $2^t$ cycles. By the pigeonhole principle, $G[E_F]$ contains two distinct cycles, say $C$ and $C'$, such that for all $1\leq i\leq t$, it holds that $\type(E(C),i)=\type(E(C'),i)$. Without loss of generality, we choose some edge $e\in E(C)\setminus E(C')$ (arbitrarily), and denote $W=A^e$. Note that $F\setminus\{W\}$ is not a solution, else we contradict the minimality of~$F$. 

Since $F\setminus\{W\}$ is not a solution, there exists a terminal $Q\in T$ such that $Q\notin\spanS(F\setminus\{W\})$. Thus, since $Q\in\spanS(F)$, there exists a subset $F'\subseteq F$ such that $W\in F'$ and $\sum_{W'\in F'}W'=Q\mod 2$. Let $E_C$  and $E_{C'}$ denote the set of columns of $A$ indexed by the edges in $E(C)$ and $E(C')$, respectively. Since $C$ and $C'$ are cycles, it holds that $\sum_{e\in E(C)}A^e=\sum_{e\in E(C)}P^e\mod 2$, and $\sum_{e\in E(C')}A^e=\sum_{e\in E(C')}P^e\mod 2$. Moreover, since for all $1\leq i\leq t$, it holds that $\type(E(C),i)=\type(E(C'),i)$, we have that $\sum_{e\in E(C)}P^e=\sum_{e\in E(C')}P^e\mod 2$. Therefore, $\sum_{W'\in E_C}W'=\sum_{W'\in E_{C'}}W'\mod 2$. We deduce that
\[\sum_{W'\in F'}W' + \sum_{W'\in E_C}W' + \sum_{W'\in E_{C'}}W' =Q\mod 2.\]
Let us denote $\widehat{E}=(E_C\setminus E_{C'})\cup (E_{C'}\setminus E_C)$. Note that $W\in\widehat{E}$.
We get that $\sum_{W'\in F'}W' + \sum_{W'\in \widehat{E}}W'=Q\mod 2$. Therefore, $Q\in\spanS(\widehat{F})$ where $\widehat{F} = (F'\setminus\widehat{E})\cup (\widehat{E}\setminus F')$. Note that $\widehat{F}\subseteq F$. Moreover, since $W\in F'\cap\widehat{E}$, we have that $W\notin \widehat{F}$. Thus, $Q\in\spanS(F\setminus\{W\})$, which is a contradiction.
\end{proof}

Let $\cal H$ denote the set of all non-isomorphic graphs with at most $k$ edges, at most $2^t$ cycles, and no isolated vertices. Note that each graph in $\cal H$ is not complicated in the following sense:

\begin{observation}\label{obs:decompH}
Given $H\in{\cal H}$, let $S$ be a subforest on $V(H)$ of maximum number of edges (spanning forest) of $H$. Then, $V(H)=V(S)$ and $|E(H)\setminus E(S)|\leq 2^t$.
\end{observation}

Note that $|{\cal H}|=2^{\OO(k)}$. Indeed, there are at most $2^{\OO(k)}$ distinct (up to isomorphism) forests on at most $k$ vertices, and by Observation \ref{obs:decompH}, each graph in ${\cal H}$ can be obtained by adding at most $2^t$ to one of these forests (so overall there are $2^{\OO(k)}$ options to examine). Therefore, we may explicitly examine each graph in $\cal H$.

\begin{definition}
Given $H\in {\cal H}$, we say that a solution $F$ is {\em compatible with $H$} if $G[E_F]$ is isomorphic to $H$ where $E_F=\edges(F)$.
\end{definition}

In light of Lemma \ref{lem:fewCyc}, to prove Theorem \ref{thmgraphic}, it is sufficient to prove the following lemma, on which we will focus next.

\begin{lemma}\label{lem:sufficeH}
Given $H\in {\cal H}$, it is possible to decide in time $k^{\OO(k)}\cdot (n+m)^{\OO(1)}$ whether there exists a solution that is compatible with $H$. 
\end{lemma}

From now on, we fix $H\in{\cal H}$. We also let $\widehat{H}$ denote a spanning forest of $H$.

\subsection{Enriching the Backbone}

Denote $\widetilde{E}=E(H)\setminus E(\widehat{H})$, and recall that $|\widetilde{E}|\leq 2^t$.
We start by guessing how we should map the edges in $\widetilde{E}$ to edges in $E(G)$. Let $\widetilde{V}$ be the set of each vertex in $E(H)$ that is an endpoint of at least one edge in $\widetilde{E}$. Let ${\cal F}$ be the set of all injective functions $f: \widetilde{V}\rightarrow V(G)$. Given a function $f\in{\cal F}$, let ${\cal F}_E$ denote the set of injective functions $f_E: \widetilde{E}\rightarrow E(G)\setminus\edges(T)$ that assign each edge $\{v,u\}\in\widetilde{E}$ an edge whose endpoints are $f(v)$ and $f(u)$ (since $G$ may contain parallel edges, ${\cal F}_E$ may contain more than one function). Now, we also label the edges in $E(\widehat{H})$. For this purpose, let ${\cal L}_E$ denote the set of all functions $\ell: E(\widehat{H})\rightarrow \{1,2,\ldots,t\}$. Accordingly, we define a notion of compatibility.

\begin{definition}
Given functions $f\in{\cal F}$, $f_E\in{\cal F}_E$ for the suitable set ${\cal F}_E$ and $\ell\in {\cal L}_E$, we say that a solution $F$ is {\em compatible with $(H,f,f_E,\ell)$} if there is an isomorphism $g: V(H)\cup E(H)\rightarrow V(G[E_F])\cup E_F$ between $H$ and $G[E_F]$,\footnote{Since we handle {\em multi} graphs, we define the domain and image of $g$ to include edge-sets.} where $E_F=\edges(F)$, such that the following conditions are satisfied.
\begin{itemize}
\item For all $e\in \widetilde{E}$, it holds that $g(e)=f_E(e)$; for all $v\in \widetilde{V}$, it holds that $g(v)=f(v)$.
\item For all $e\in E(\widehat{H})$, it holds that $\type(f_E(e))=\ell(e)$.
\end{itemize}
\end{definition}

Note that $|{\cal F}|,|{\cal F}_E|\leq (n+m)^{2\cdot 2^t} = (n+m)^{\OO(1)}$, $|{\cal L}|\leq t^k = 2^{\OO(k)}$, and for each solution $F$ compatible with $H$, there exists a triple $(H,f,f_E,\ell)$ with whom $F$ is compatible. Recall that we have argued that to prove Theorem \ref{thmgraphic}, it is sufficient to prove Lemma \ref{lem:sufficeH}. By examining every triple of a function in ${\cal F}$, a function in the suitable set ${\cal F}_E$ and a function in ${\cal L}$, we obtain that in order to prove Lemma \ref{lem:sufficeH}, it is actually sufficient to prove the following lemma, on which we focus next.

\begin{lemma}\label{lem:sufficeTriple}
Given $H\in {\cal H}$, $f\in{\cal F}$, $f_E\in{\cal F}_E$ for the suitable set ${\cal F}_E$ and $\ell\in {\cal L}$, it is possible to decide in time $k^{\OO(k)}\cdot (n+m)^{\OO(1)}$ whether there exists a solution that is compatible with $(H,f,f_E,\ell)$.
\end{lemma}

From now on, we fix $f\in{\cal F}$, $f_E\in{\cal F}_E$ for the suitable set ${\cal F}_E$ and $\ell\in {\cal L}$.

\subsection{Parities}

The ``guesses'' performed so far are insufficient for us to be able to detect a solution $F$. We need to define another set of functions, which captures information relevant to each terminal in $T$. This information is necessary since we need it to validate that each terminal indeed belongs to the span of the columns we are about to select. For this purpose, define $\widehat{B}=\{B=(b_1,b_2,\ldots,b_t): b_1,b_2,\ldots,b_t\in \{0,1\}\}$. That is $\widehat{B}$ is the set of binary vectors with $t$ elements. The size of this set is $2^t$. Moreover, we call a function $h: T\rightarrow \widehat{B}$ a {\em parity restriction}. That is, a parity restriction maps each terminal $W\in T$ to a vector $B_W\in \widehat{B}$. Note that there exist only $(2^t)^{|T|}=2^{\OO(k)}$ parity restrictions. To exploit  a parity restriction $h$, we need the following definitions.

\begin{definition}
Given a terminal $W\in T$ and a vector $B\in \widehat{B}$, we say that a solution $F$ is {\em compatible with $(W,B)$} if there exists a subset $X\subseteq F$ such that $W+\sum_{W'\in X}W'=0\mod 2$ and for all $1\leq i\leq t$, it holds that $\type(\edges(X),i)=b_i\mod 2$.
\end{definition}

\begin{definition}
Given a parity restriction $h$, we say that a solution $F$ is {\em compatible with $h$} if for all $W\in T$, it holds that $F$ is compatible with $(W,h(W))$.
\end{definition}

Note that for each solution, there exists a parity restriction with whom it is compatible. Indeed, consider some solution $F$. Then, for each $W\in T$, the fact that $W\in\spanS(F)$ implies that there exists $X_W\subseteq F$ such that $W+\sum_{W'\in X_W}W'=0\mod 2$. Thus, for each $W\in T$, we set $h(W)$ to be the vector $(b_1,b_2,\ldots,b_t)\in\widehat{B}$ such that $\type(\edges(X_W),i)=b_i\mod 2$. Now, note that there exist only $(2^t)^{|T|}=2^{\OO(k)}$ parity restrictions. Recall that we have already argued that to prove Theorem \ref{thmgraphic}, it is sufficient to prove Lemma \ref{lem:sufficeTriple}. Thus, we have that it is sufficient to show that given a parity restriction $h$, we can decide in time $k^{\OO(k)}\cdot (n+m)^{\OO(1)}$ whether there exists a solution that is compatible with both $(H,f,f_E\ell)$ and $h$. Therefore, from now on, we fix a parity restriction $h$. However, we are interested in imposing additional restrictions, after which we summarize (in a lemma) the objective of the rest of this section. For this purpose, we first need to introduce additional notation.

Given a vector $B\in\widehat{B}$, we let $\sumS(B)$ denote the binary vector with $n$ elements whose $i^{\mathrm{th}}$ element is $(\sum_{j=1}^tb_i\cdot c_{i,j})\mod 2$ where $c_{i,j}$ is the $j^{\mathrm{th}}$ element of $C_i$. Moreover, we index $\sumS(B)$ by the vertices in $V(G)$ (in the manner compatible with the one in which we index the rows of $P$). Let $\widehat{h}$ be the function whose domain is $T$, and for each $W\in T$, it is defined by $\widehat{h}(W)=\sumS(h(W))+W\mod 2$. Now, let $h^*$ be the function whose domain is $T$, and for each $W\in T$, it assigns the set of the vertices that are indices of the elements of $\widehat{h}(W)$ that equal 1.

Denote $V^*=(\bigcup_{W\in T}h^*(W))\cup \image(f)$. Moreover, denote ${\cal D}=\{D\subseteq V(H): \widetilde{V}\subseteq D, |D|=|V^*|\}$. Note that in case $|V^*|>|V(H)|$, it holds that ${\cal D}=\emptyset$. Given $D\in{\cal D}$, let ${\cal F}^*_D$ be the set of injective functions $f^*: D\rightarrow V^*$ such that for all $v\in \widetilde{V}$, it holds that ${f^*}(v)=f(v)$. Let $E^*_D$ denote the set of edges in $E(H)$ with both endpoints in $D$.
Given a function $f^*\in{\cal F}^*_D$, let $f^*_{E}: E^*_D\rightarrow E(G)\cup\{\nil\}$ denote the injective function that assigns each edge
$e\in E^*$, with endpoints $v$ and $u$, the edge between $f^*(v)$ and $f^*(u)$ in $E(G)$ of type $\ell(e)$ if one exists
and $\nil$ otherwise (since the columns of $A$ are distinct, there can be at most one such edge between $f^*(v)$ and $f^*(u)$ in $E(G)$ of type $\ell(e)$, and thus the function is well defined). We say that $f^*_E$ is {\em feasible} if $\image(f^*_{E})\subseteq E(G)\setminus\edges(T)$ and for all $e\in\widetilde{E}$, it holds that $f^*_E(e)=f_E(e)$.

Roughly speaking, the two following definitions capture the notion of compatibility which is associated with the notation we have just presented. Here, vertices of odd degree play a central role, which will be cleared at the proof of Lemma \ref{lem:existDf*} below. We remark that here, a type of an edge $e\in E(H)$ is defined as follows: if $e\in E(\widehat{H})$, then the type is $\ell(e)$, and otherwise the type is the same as the one of $f_E(e)$. Accordingly, we may use the notation $\type(E',i)=b_i$ also for a set $E'\subseteq E(H)$.

\begin{definition}\label{def:interestingW}
Given $W\in T$, $D\in{\cal D}$, $f^*\in{\cal F}^*_D$, we say that $(D,f^*)$ is {\em interesting with respect to $W$} if there exists $E_W\subseteq E(H)$ such that the following conditions are satisfied.
\begin{itemize}
\item Let $O_W$ be the set of vertices of odd degree in $H[E_W]$. Then, $O_W\subseteq\domain(f^*)$ and $f^*(O_W)=h^*(W)$.
\item Denote $h(W)=(b_1,b_2,\ldots,b_t)$. Then, for all $1\leq i\leq t$, it holds that $\type(E_W,i)=b_i$.
\end{itemize}
\end{definition}

\begin{definition}\label{def:interestingPair}
Given $D\in{\cal D}$ and $f^*\in{\cal F}^*_D$, we say that $(D,f^*)$ is {\em interesting} if $f^*_E$ is feasible and for all $W\in T$, it holds that $(D,f^*)$ is interesting with respect to $W$.
\end{definition}

Next, we give one definition which captures the notion of compatibility that we need, and subsumes previous such definitions.

\begin{definition}
Given $D\in{\cal D}$, $f^*\in{\cal F}^*_D$, we say that a solution $F$ is {\em compatible with $(H,\ell,D,f^*)$} if there is an isomorphism $g: V(H)\cup E(H)\rightarrow E(G[E_F])\cup E_F$ where $E_F=\edges(F)$ between $H$ and $G[E_F]$, such that the following conditions are satisfied.
\begin{itemize}
\item  For all $e\in E^*_D$, it holds that $g(e)=f^*_E(e)$; for all $v\in D$, it holds that $g(v)=f^*(v)$
\item For all $e\in E(\widehat{H})$, it holds that $\type(g(e))=\ell(e)$.
\end{itemize}
\end{definition}

Informally, we are now ready to prove that by using this notion of compatibility and also further restricting ourselves by Definition \ref{def:interestingPair}, we do not discard all solutions from the set of those which we currently seek (that is, the set of solutions compatible with both $(H,f,f_E,\ell)$ and $h$).

\begin{lemma}\label{lem:existDf*}
If there exists a solution $F$ that is compatible with both $(H,f,f_E,\ell)$ and $h$, then there exist $D\in{\cal D}$ and $f^*\in{\cal F}^*_D$ such that $(D,f^*)$ is interesting and $F$ is compatible with $(H,\ell,D,f^*)$.
\end{lemma}

\begin{proof}
Let $F$ be a solution that is compatible with both $(H,f,f_E,\ell)$ and $h$. Consider some $W\in T$, and denote $h(W)=(b^W_1,b^W_2,\ldots,b^W_t)$. Since $F$ is compatible with $h$, there exists $X_W\subseteq F$ such that $W+\sum_{W'\in X_W}W'=0\mod 2$ and for all $1\leq i\leq t$, it holds that $\type(E'_W,i)=b^W_i\mod 2$ where $E'_W=\edges(X_W)$. In particular, this means that $W+\sum_{e\in E'_W}P^e=R^*_W$ where $R^*_W=\widehat{h}(W)$. Let $O'_W$ be the set of each vertex $v\in V(G)$ that is incident to an odd number of edges in $E'_W$.
Let $\widetilde{R}(W)$ be the binary vector with $n$ elements, whose elements are indexed by vertices in the manner compatible with $A$, such that for all $v\in V(G)$, it holds that $\widetilde{r}(W)_v=1$ if and only if $v\in O'_W$. Then, $W+\sum_{W'\in X_W}W' = W+\widetilde{R}(W)+\sum_{e\in E'_W}P^e=\widetilde{R}(W)+R^*_W = 0\mod 2$. This means that every vertex which is an index of an element that is 1 in $R^*_W$ belongs to $V(G[E_F])$. Therefore, $\bigcup_{W\in T}h^*(W)\subseteq V(G[E_F])$

Since $F$ is compatible with $(H,f,f_E,\ell)$, there is an isomorphism $g: V(H)\cup E(H)\rightarrow E(G[E_F])\cup E_F$ between $H$ and $G[E_F]$, where $E_F=\edges(F)$, such that the following conditions hold.
\begin{itemize}
\item For all $e\in \widetilde{E}$, it holds that $g(e)=f_E(e)$; for all $v\in \widetilde{V}$, it holds that $g(v)=f(v)$.
\item For all $e\in E(\widehat{H})$, it holds that $\type(g(e))=\ell(e)$.
\end{itemize}

Overall, we get that $V^*\subseteq V(G[E_F])$. We define $D=\{v\in V(H): g(v)\in V^*\}$. Then, $\widetilde{V}\subseteq D$ and $|D|=|V^*|$. We now define an injective function $f^*: D\rightarrow V^*$ as follows. For all $v\in D$, we set $f^*(v)=g(v)$. Since for all $v\in \widetilde{V}$, it holds that $g(v)=f(v)$, we get that for all $v\in \widetilde{V}$, it holds that ${f^*}(v)=f(v)$. Thus, $f^*\in{\cal F}^*_D$. Since $F\cap T=\emptyset$ and $g$ is an isomorphism such that for all $e\in E(\widehat{H})$, it holds that $\type(g(e))=\ell(e)$, we also have that $f^*_E$ is feasible. Therefore, $F$ is compatible with $(H,\ell,D,f^*)$. For each terminal $W\in T$, define $E_W=\{e\in E(H): g(e)\in E'_W\}$. Now, consider some $W\in T$. Observe that for $O_W$, the set of vertices of odd degree in $H[E_W]$, it holds that $O'_W=\{g(v): v\in O_W\}$. Recall that $\widetilde{R}(W)+R^*_W = 0\mod 2$. Therefore, it holds that the element indexed by $v$ in $\widetilde{R}(W)$ is 1 (which is equivalent to $v\in f^*(O_W)$) if and only if this element is also 1 in $R^*_W$ (which is equivalent to $v\in h^*(W)$). Denote $h(W)=(b_1,b_2,\ldots,b_t)$. Then, for all $1\leq i\leq t$, it holds that $\type(E_W,i)=b_i$. Denote $h(W)=(b_1,b_2,\ldots,b_t)$. Recall that $\type(E'_W,i)=b_i\mod 2$. Therefore, for all $1\leq i\leq t$, it holds that $\type(E_W,i)=b^W_i$. We thus also conclude that $(D,f^*)$ is interesting.
\end{proof}

Note that $|{\cal D}|\leq 2^{|V(H)|}\leq 4^k$ and for each $D\in {\cal D}$, it holds that $|{\cal F}_D|\leq |D|^{|D|}\leq (2k)^{2k}$. Moreover, given a pair $(D,f^*)$, we can test whether it is interesting in time $2^{\OO(k)}\cdot (n+m)^{\OO(1)}$ as follows. First, we can clearly test whether $f^*_E$ is feasible in polynomial time. Then, for each $W\in T$, we examine each subset of edges of $E(H)$ (there are at most $2^k$ such subsets), where for each subset we check in polynomial time whether the conditions in Definition \ref{def:interestingW} are satisfied (recall that the type of each edge in $E(H)$ is either given by $\ell$ or determined by $f^*_E$). Thus, by exhaustively checking every interesting pair $(D,f^*)$, we have that in order to prove Lemma \ref{lem:sufficeTriple}, 
it is sufficient to prove the following lemma, on which we focus next.

\begin{lemma}\label{lem:sumGuesses}
Given $D\in{\cal D}$ and $f^*\in{\cal F}^*_D$ such that $(D,f^*)$ is interesting, it is possible to decide in time $k^{\OO(k)}\cdot (n+m)^{\OO(1)}$ whether there exists a solution that is compatible with $(H,\ell,D,f^*)$.
\end{lemma}

Thus, from now on, we fix $D\in{\cal D}$, $f^*\in{\cal F}^*_D$ such that $(D,f^*)$ is interesting.

\subsection{\pattern: Graph Problem}

We would next like to focus on a graph problem rather than our current problem (which is the case where we seek a solution compatible with $(H,\ell,D,f^*)$). For this purpose, we need the following notation and definition.

We denote $G'=G\setminus(\image(f^*_E)\cup\edges(T))$. Moreover, denote $H'=H\setminus E^*_D$.

\begin{definition}
Let $S$ be a subpgraph of $G'$. We say that $S$ is {\em excellent} if there exists an isomorphism $g: V(H')\cup E(H')\rightarrow V(S)\cup E(S)$  between $H'$ and $S$, such that the following conditions hold.
\begin{itemize}
\item For all $v\in D$, it holds that $g(v)=f^*(v)$.
\item For all $e\in E(H')$, it holds that $\type(g(e))=\ell(e)$.
\end{itemize}
\end{definition}

Note that since $H'$ is a forest, we have that an excellent subgraph is also a forest (which also implies that it is a simple graph).

\begin{lemma}
If there exists a solution $F$ that is compatible with $(H,\ell,D,f^*)$, then $G'$ has a subgraph $S$ that is excellent.
\end{lemma}

\begin{proof}
Suppose that there exists a solution $F$ that is compatible with $(H,\ell,D,f^*)$. Then, there is an isomorphism $g: V(H)\cup E(H)\rightarrow V(G[E_F])\cup E_F$ between $H$ and $G[E_F]$, where $E_F=\edges(F)$, such that the following conditions are satisfied.
\begin{itemize}
\item For all $e\in E^*_D$, it holds that $g(e)=f^*_E(e)$; for all $v\in D$, it holds that $g(v)=f^*(v)$.
\item For all $e\in E(\widehat{H})$, it holds that $\type(g(e))=\ell(e)$.
\end{itemize}

We define a graph $S$ as follows. We set $V(S)=V(G[E_F])$ and $E(S)=E_F\setminus\image(f^*_E)$. Since $F\cap T=\emptyset$, we get that $S$ is a subgraph of $G'$. By the existence of the isomorphism $g$, we conclude that $S$ is excellent.
\end{proof}

\begin{lemma}
If $G'$ has a subgraph $S$ that is excellent, then there exists a solution $F$ that is compatible with $(H,\ell,D,f^*)$.
\end{lemma}

\begin{proof}
Suppose that $G'$ has a subgraph $S'$ that is excellent. Then, there exists an isomorphism $g': V(H')\cup E(H')\rightarrow V(S')\cup E(S')$  between $H'$ and $S'$, such that the following conditions hold.
\begin{itemize}
\item For all $v\in D$, it holds that $g(v)=f^*(v)$.
\item For all $e\in E(H')$, it holds that $\type(g(e))=\ell(e)$.
\end{itemize}

Let us extend $S'$ to $S$ by adding the edges in $\image(f^*_E)$. Accordingly, by the definitions of $f^*$ and $f^*_E$, we also extend $g'$ to $g$. Then, since $f^*_E$ is feasible, $g$ is an isomorphism between $H$ and $S$. Let $E_F=\{g(e): e\in E(H)\}$. Now, define $F=\{A^e: e\in E_F\}$.  We claim that $F$ is a solution that is compatible with $(H,\ell,D,f^*)$. By the existence of the isomorphism $g$, we deduce that if $F$ is a solution, then it is compatible with $(H,\ell,D,f^*)$. Thus, we next focus on the proof that $F$ is a solution. % compatible with $h$.

Consider some terminal $W\in T$, and denote $h(W)=(b_1,b_2,\ldots,b_t)$. We need to prove that $W\in\spanS(F)$. Along the way, we also show that $F$ is compatible with $(W,h(W))$.
For this purpose, we show that there exists $X_W\subseteq F$ such that $W+\sum_{W'\in X_W}W'=0\mod 2$, and for all $1\leq i\leq t$, it holds that $\type(\edges(X_F),i)=b_i\mod 2$. Since $(D,f^*)$ is interesting with respect to $W$, there exists $E_W\subseteq E(H)$ such that the following conditions are satisfied.
\begin{itemize}
\item Let $O_W$ be the set of vertices of odd degree in $H[E_W]$. Then, $O_W\subseteq\domain(f^*)$ and $f^*(O_W)=h^*(W)$.
\item For all $1\leq i\leq t$, it holds that $\type(E_W,i)=b_i$.
\end{itemize}

We define $X_W=\{A^{g(e)}: e\in E_W\}$ and $O'_W=\{g(v): v\in O_W\}$. From the second item, we get that for all $1\leq i\leq t$, it holds that $\type(\edges(X_F),i)=b_i\mod 2$. It remains to show that $W+\sum_{W'\in X_W}W'=0\mod 2$.

Let $\widetilde{R}$ be the binary vector with $n$ elements, whose elements are indexed by vertices in the manner compatible with $A$, such that for all $v\in V(G)$, it holds that $\widetilde{r}_v=1$ if and only if $v\in O'_W$. Then, since $F$ is compatible with $(W,h(W))$, it holds that $W+\sum_{W'\in X_W}W' = W+\widetilde{R}+\sum_{w\in\edges(X_W)}P^e=\widetilde{R}+R^*$ where $R^*=\widehat{h}(W)$. Since $O'_W=f^*(O_W)=h^*(W)$, for any vertex $v\in V(G)$, the entry indexed by $v$ in $\widetilde{R}$ is 1 (which is equivalent to $v\in O_W$) if and only if this entry is also 1 in $R^*$ (which is equivalent to $v\in h^*(W)$). Thus, we conclude that $\widetilde{R}+R^*=0\mod 2$.
\end{proof}

In other words, to prove Lemma \ref{lem:sumGuesses}, it is now sufficient that we prove the following lemma, on which we focus next.

\begin{lemma}\label{lem:newsuffice}
It is possible to decide in time $2^{\OO(k)}\cdot (n+m)^{\OO(1)}$ whether $G'$ has an excellent subgraph.
\end{lemma}

We are now ready to translate our problem into the following graph problem.

\defproblemu{\pattern}%
{A (multi) graph $G$ with $n$ vertices and $m$ edges, a non-negative integer $t$ that is a fixed constant, a function $\ell_G: E(G)\rightarrow \{1,2,\ldots,t\}$, a non-negative integer $k$, a forest $H$ with $k$ vertices, a function $\ell_H: E(H)\rightarrow \{1,2,\ldots,t\}$, a set $U\subseteq V(H)$ and an injective function $f: U\rightarrow V(G)$.}%
{Are there a subgraph $S$ of $G$ and an isomorphism $g: V(H)\cup E(H)\rightarrow V(S)\cup E(S)$ that satisfies (i) for all $e\in E(H)$, it holds that $\ell_H(e)=\ell_G(g(e))$, and (ii) for all $v\in U$, it holds that $f(v)=g(v)$?}

Indeed, given an instance of our current problem where we need to determine whether $G'$ has an excellent subgraph, we can construct an equivalent instance of \pattern{}, $(\widetilde{G},\widetilde{t},\ell_G,\widetilde{k},\widetilde{H},\ell_H,U,\widetilde{f})$ as follows. We set $\widetilde{G}=G'$, $\widetilde{t}=t$, $\widetilde{k}=|V(H')|$, $\widetilde{H}=H'$, $\ell_H=\ell$, $U=D$ and $\widetilde{f}=f^*$. For all $e\in E(G)$, we set $\ell_G(e)=\type(e)$.
Thus, to prove Lemma~\ref{lem:newsuffice}, we can focus on the proof of the following result.

\begin{lemma}\label{lem:pattern}
\pattern{} is solvable in time $2^{\OO(k)}\cdot (n+m)^{\OO(1)}$.
\end{lemma}

We next assume that $k\leq n$, else we return that the input instance (of \pattern) is a no-instance.

\subsection{Color Coding}

In what follows, we rely on the method of color coding \cite{AlonYZ}. To obtain a deterministic algorithm, we need the following definition. Here, the notation $c|_S$ refers to the restriction of the function $c$ to the domain $S\subseteq\domain(c)$.

\begin{definition}\label{def:hash}
Given integers $k'\leq n'$, a set ${\cal C}$ of functions $c: \{1,2,\ldots,n'\}\rightarrow\{1,2,\ldots,k'\}$ is an {\em $(n',k')$-family of hash functions} if for every set $S\subseteq \{1,2,\ldots,n'\}$ of size $k'$, there exists a function $c\in {\cal C}$ such that $c|_S$ is an injective function.
\end{definition}

Alon et al.~\cite{AlonYZ} proved that small hash families can be computed efficiently:

\begin{proposition}[\cite{AlonYZ}]\label{prop:hash}
Given integers $k\leq n$, there exists an $(n',k')$-family of hash functions, $\cal C$, of size $2^{\OO(k')}\cdot\log 'n$, and this family is computable in time $2^{\OO(k')}\cdot n'\log n'$.
\end{proposition}

We proceed by employing Proposition \ref{prop:hash} to construct an $(n,k)$-family of hash functions, $\cal C$, of size $2^{\OO(k)}\cdot\log n$. Note that the term solution now refers to a pair $(S,g)$ where $S$ is a subgraph of $G$ and $g: V(H)\cup E(H)\rightarrow V(S)\cup E(S)$ is an isomorphism that satisfies (i) for all $e\in E(H)$, it holds that $\ell_H(e)=\ell_G(g(e))$, and (ii) for all $v\in U$, it holds that $f(v)=g(v)$.

\begin{definition}
Given a function $c\in{\cal C}$, we say that a solution $(S,g)$ is {\em colorful with respect to $c$} if for all distinct $v,u\in V(S)$, it holds that $c(g(v))\neq c(g(u))$.
\end{definition}

Next, the function $c$ will be clear from context, and therefore we simply use the term colorful. By Definition \ref{def:hash}, if the input instance is a yes-instance, then there exists a function $c\in{\cal C}$ such that there exists a colorful solution. Since we can examine each function in $\cal C$ individually, to prove Lemma \ref{lem:pattern} (and thus Theorem \ref{thmgraphic}), it is sufficient to prove the following lemma, on which we focus next.

\begin{lemma}\label{lem:sumColor}
Given $c\in{\cal C}$, it is possible to decide in time $2^{\OO(k)}\cdot (n+m)^{\OO(1)}$ whether there exists a colorful solution.
\end{lemma}

Thus, from now on, we fix some $c\in{\cal C}$.

\subsection{Dynamic Programming}

Here, we use standard dynamic programming to prove Lemma \ref{lem:sumColor}. The details are given for the sake of completeness. We define an arbitrary order on $V(H)$. We denote the connected components of $H$ by $H_1,H_2,\ldots,H_p$ for the appropriate $p$ (the order is arbitrary). Now, we root each tree $H_i$, $1\leq i\leq p$, at some arbitrary vertex $r_i\in V(H_i)$. Given $1\leq i\leq p$ and a vertex $v\in V(H_i)$, we let $\subtree(v)$ denote the subtree of $H_i$ that is rooted at $v$. Moreover, for all $1\leq j\leq |\children(v)|$, we let $\children(v,j)$ denote the $j^{\mathrm{th}}$ child of $v$. We define $\subtree(v,0)=\subtree(v)$.
Now, for all $1\leq j\leq |\children(v)|$, we define $\subtree(v,j)$ as the tree of $\subtree(v)$ from which we remove $\bigcup_{j'\leq j}\subtree(\children(v,j'))$.

We will use two table. The first table, M, has an entry $[i,v,x,j,C]$ for all $1\leq i\leq p$, $v\in V(H_i)$, $x\in V(G)$, $0\leq j\leq |\children(v)|$ and $C\subseteq \{1,2,\ldots,k\}$. The second table will be discussed later. To explain the purpose of an entry M$[i,v,x,j,C]$, we use the following definition.

\begin{definition}
Given $1\leq i\leq p$, $v\in V(H_i)$, $x\in V(G)$, $0\leq j\leq |\children(v)|$ and $C\subseteq \{1,2,\ldots,k\}$, we say that a pair $(S,g)$ of a subgraph $S$ of $G$ and an isomorphism $g: V(\subtree(v,j))\cup E(\subtree(v,j))\rightarrow V(S)\cup E(S)$ is an {\em $(i,v,x,j,C)$-solution} if the following conditions are satisfied.
\begin{enumerate}
\item For all $e\in E(\subtree(v,j))$, it holds that $\ell_H(e)=\ell_G(g(e))$.
\item For all $u\in U\cap V(\subtree(v,j))$, it holds that $f(u)=g(u)$.
\item $g(v)=x$.
\item $c(g(V(\subtree(v,j))))=C$ and for all $u,w\in V(\subtree(v,j))$, it holds that $c(g(u))\neq c(g(w))$.
\end{enumerate}
\end{definition}

Now, M$[i,v,x,j,C]$ should store a Boolean value, 0 or 1, and it should store 1 if and only if there exists an $(i,v,x,j,C)$-solution. The computation of M is given by the following formulas, whose correctness is straightforward.

\myparagraph{Basis.} The basis corresponds to the following items.
\begin{itemize}
\item If $v\in U$ and $x\neq f(v)$, then M$[i,v,x,j,C]$ is 0.
\item Else if $c(x)\notin C$ or $|\subtree(v,j)|\neq |C|$, then M$[i,v,x,j,C]$ is 0.
\item Else if $|\subtree(v,j)|=1$, then M$[i,v,x,j,C]$ is 1.
\end{itemize}

\myparagraph{Step.} Now, assume that we handle an entry that do not handled by any base case. Here, we let $u$ denote the $j^{\mathrm{th}}$ child of $v$. Moreover, $Y$ is the set of neighbors $y$ of $x$ such that there is an edge $e\in E(G)$ between $x$ and $y$ (there may be multiple edges) that satisfies $\ell_G(e)=\ell_H(\{v,u\})$. In case $u\in U$, we update $Y$ to $Y\cap \{f(u)\}$. If $Y\emptyset$, then M$[i,v,x,j,C]$ is 0. Otherwise, the computation is given by the following formulas.

\[\displaystyle{\mathrm{M}[i,v,x,j,C] = \max_{y\in Y}\max_{C'\cap C''=\emptyset, C'\cup C''=C}\mathrm{M}[i,v,x,j+1,C']\cdot\mathrm{M}[i,u,y,0,C''] }.\]

Clearly, the table M can be computed in time $2^{\OO(k)}\cdot (n+m)^{\OO(1)}$. We proceed to define another table, N. The table N has an entry $[i,C]$ for all $1\leq i\leq p$ and $C\subseteq \{1,2,\ldots,k\}$. We let $H^i$ denote the subforest of $H$ consisting of the trees $H_j$ for all $1\leq j\leq i$. To explain the purpose of an entry N$[i,C]$, we use the following definition.

\begin{definition}
Given $1\leq i\leq p$ and $C\subseteq \{1,2,\ldots,k\}$, we say that a pair $(S,g)$ of a subgraph $S$ of $G$ and an isomorphism $g: V(H^i)\cup E(H^i)\rightarrow V(S)\cup E(S)$ is an {\em $(i,C)$-solution} if the following conditions are satisfied.
\begin{enumerate}
\item For all $e\in E(H^i)$, it holds that $\ell_H(e)=\ell_G(g(e))$.
\item For all $u\in U\cap V(H^i)$, it holds that $f(u)=g(u)$.
\item $c(g(V(H^i)))=C$ and for all $u,w\in V(H^i)$, it holds that $c(g(u))\neq c(g(w))$.
\end{enumerate}
\end{definition}

Now, N$[i,C]$ should store a Boolean value, 0 or 1, and it should store 1 if and only if there exists an $(i,C)$-solution. The computation of N is given by the following formulas, whose correctness is straightforward.

\myparagraph{Basis.} The basis corresponds to the case where $i=1$. Then, the computation is given by the following formula.
\[\mathrm{N}[i,C] = \max_{x\in V(G)}\mathrm{M}[i,r_1,x,0,C].\]

\myparagraph{Step.} Now, assume that $i >1$. Then, the computation is given by the following formula.
\[\displaystyle{\mathrm{N}[i,C] = \max_{C'\cap C''=\emptyset, C'\cup C''=C}\mathrm{N}[i-1,C']\cdot(\max_{x\in V(G)}\mathrm{M}[i,r_i,x,0,C''])}.\]

At the end, we answer that there exists a colorful solution if and only if N$[p,\{1,2,\ldots,k\}]$ is~1.

\section{FPT on Duals of Perturbed Graphic Matroids}
\label{section:dual}
In this section, we prove Theorem~\ref{thmdual}.
 The proof of Theorem \ref{thmdual} is done by a chain of statements, each adding constraints on the structures of the solutions we seek. Here, we give the technical details of the approach described in Section \ref{sec:overview}.

\subsection{Bounding $|T|$ and Assigning Colors to Vertices}

As in the case of perturbed graphic matroids, we may assume that $|T|\leq k$.  Recall that $\disR(P)$ is defined as the set of the distinct vectors corresponding the {\em rows} in $\{P_v: v\in V(G)\}$. Let us denote $|\disR(P)|=t$. By Proposition \ref{lem:rank}, it holds that $t\leq 2^r$. Moreover, denote $\disR(P)=\{R_1,R_2,\ldots,R_t\}$. Accordingly, we say that a vertex $v\in V(G)$ is of {\em type} $i$, $1\leq i\leq t$, if $P_v=R_i$. Given a vertex $v\in V(G)$, we let $\type(v)$ denote its type.

\subsection{Characterization of Solutions}

In the context of \spcdual, we use the term {\em solution} to refer to a set $F\subseteq E\setminus T$ with $|F|\leq k$ such that $T\subseteq \spn(F)$ in the dual $M^*$ of the binary matroid $M$ represented by $A$. Let $I$ be a binary vector with $m$ elements. Recall that given $F\subseteq E$, we let $\edges(F)$ denote the set of each edge $e\in E(G)$ such that $A^e\in F$. Now, given a set $F\subseteq E$, we say that $I$ is the {\em characteristic vector} of $F$ if the $i^\mathrm{th}$ entry of $I$ is 1 if and only if $F$ contains the  $i^\mathrm{th}$ column of $A$. Moreover, given a set $F\subseteq E$, we say that $F$ is a {\em cocyle} in $M$ if and only if it is a cycle in ${M}^*$. In what follows, we rely on the following folklore result (see, e.g., \cite{GeelenK15}).

\begin{proposition}\label{lem:cocycle}
Let $M$ be a binary matroid represented by an $(n\times m)$-matrix $A$, and let $F$ be a subset of $E$, where $E$ is the set of columns of $A$. Then, $F$ is a cocycle in $M$ if and only if the characteristic vector of $F$ belongs to $\spanS(V)$ where $V$ is the set of rows of $A$.
\end{proposition}

Note that a set $F\subseteq E\setminus T$ is a solution if and only if for each terminal $W\in T$, there exists a subset $F'\subseteq F$ such that $F'\cup\{W\}$ is a cocycle in $M$. Thus, in light of Proposition~\ref{lem:cocycle}, we can think of a solution as follows:

\begin{observation}\label{lem:incidence}
A set $F\subseteq E\setminus T$ is a solution if and only if $|F|\leq k$ and for each terminal $W\in T$, there exists a subset $F'\subseteq F$ such that the characteristic vector of $F'\cup\{W\}$ belongs to $\spanS(V)$, where $V$ is the set of rows of $A$.
\end{observation}

\subsection{Parities and Cuts}

Define $\widehat{B}=\{B=(b_1,b_2,\ldots,b_t): b_1,b_2,\ldots,b_t\in \{0,1\}\}$. The size of this set is $2^t$. Given a vector $B\in\widehat{B}$, we let $\sumS(B)$ denote the binary vector with $m$ elements whose $i^{\mathrm{th}}$ element is $(\sum_{j=1}^tb_i\cdot r_{i,j})$ where $r_{i,j}$ is the $j^{\mathrm{th}}$ element of $R_i$. Moreover, we index $\sumS(B)$ by $E(G)$ (in the manner compatible with the one in which we index the columns of $P$).
To exploit $\widehat{B}$, we need several definitions.

\begin{definition}
Given a vector $B\in \widehat{B}$, we say that a set $X\subseteq V(G)$ is {\em compatible with $B$} if for all $1\leq i\leq t$, it holds that $|\{v\in V(G): \type(v)=i\}|=b_i\mod 2$.
\end{definition}

In the context of cographc matroids, it is natural to discuss cuts that correspond to circuits. However, here we do not deal with cographic matroids, but with duals of {\em perturbed} graphic matroids. For this reason, we need the following definition. The necessity of this precise definition will be clear from the proofs of Lemmata \ref{lem:cut1} and \ref{lem:cut2} ahead.

\begin{definition}\label{def:expensive}
Given a vector $B\in \widehat{B}$, let $(X,\bar{X})$ be a partition of $V(G)$, $e$ be an edge in $E(G)$, and $\widetilde{r}$ be the element of index $e$ in $\sumS(B)$. We say that an edge $e\in E(G)$ is {\em expensive with respect to $(B,(X,\bar{X}))$} if one of the three following conditions is true:
\begin{enumerate}
\item Both endpoints of $e$ belong to $X$ and $\widetilde{r}=1$.
\item Both endpoints of $e$ belong to $\bar{X}$ and $\widetilde{r}=1$.
\item Exactly one of the endpoints of $e$ belongs to $X$ and $\widetilde{r}=0$.
\end{enumerate}
We say that an edge $e\in E(G)$ is {\em cheap with respect to $(B,(X,\bar{X}))$} if it is not expensive with respect to $(X,\bar{X})$. 
\end{definition}

In the context of Definition \ref{def:expensive}, we remark that if ($B,(X,\bar{X}))$ is clear from context, we simply use the terms expensive and cheap. Moreover, we let $\expens(X)$ and $\cheap(X)$ denote the sets of expensive and cheap edges, respectively. Given a set $C$ of columns of $A$, let $\edges(C)$ denote the set of each edge in $E(G)$ that is an index of a column in $C$. 

\begin{definition}\label{def:partition}
Let $(X,\bar{X})$ be a partition of $V(G)$. 
Given a vector $B\in \widehat{B}$ and an edge $e\in \edges(T)$, we say that $(X,\bar{X})$ is {\em good with respect to $(B,e)$} if $X$ is compatible with $B$, $e$ is expensive and $\edges(T)\setminus\{e\}\subseteq\cheap(X)$.
\end{definition}

We call a function $h: T\rightarrow\widehat{B}$ a {\em parity restriction}. When $h$ is clear from context, we denote $h(W)=B_W$. We are now ready to present an alternative way to view a solution. The validity of this view is given by the two lemmata that follow it.

\begin{definition}
Given a parity restriction $h$, we say that a set $S\subseteq E(G)\setminus\edges(F)$ is a {\em solution compatible with $h$} if $|S|\leq k$ and for all $e\in \edges(T)$, there exists a partition $(X_e,\bar{X}_e)$ of $V(G)$ that is good with respect to $(B_W,e)$, where $W=A^e$, and such that $\expens(X_e)\setminus\{e\}\subseteq S$.
\end{definition}

\begin{lemma}\label{lem:cut1}
If there exist a parity restriction and a solution $S\subseteq E(G)\setminus\edges(F)$ compatible with it, then the input instance is a yes-instance.
\end{lemma}

\begin{proof}
Suppose that there exist a parity restriction $h$ and a solution $S\subseteq E(G)$ compatible with it. Then, it holds that $|S|\leq k$ and for all $e\in \edges(T)$, there exists a partition $(X_e,\bar{X}_e)$ of $V(G)$ that is good with respect to $(B_W,e)$, where $W=A^e$, and such that $\expens(X_e)\setminus\{e\}\subseteq S$.
By Lemma \ref{lem:incidence}, to show that the input instance is a yes-instance, it is sufficient to show that there exists $F\subseteq E\setminus T$ such that $|F|\leq k$ and for each terminal $W\in T$, there exists a subset $F_W\subseteq F$ such that the characteristic vector of $F_W\cup\{W\}$ belongs to $\spanS(V)$.

For each terminal $W\in T$, we define $F_W = \{A^e: e\in\expens_{e_W}(X)\}\setminus\{W\}$ where $W=A^{e_W}$. Accordingly, we define $F=\bigcup_{W\in T}F_W$. Since $|S|\leq k$ and for all $e\in \edges(T)$, it holds that $\expens(X_e)\setminus\{e\}\subseteq S$, we have that $|F|\leq k$. Next, fix some terminal $W\in T$, and let $I_W$ be the characteristic vector of $F_W\cup\{W\}$. To show that $I_W\in\spanS(V)$, it is sufficient to show that $\sum_{v\in X_{e_W}}A_v = I_W\mod 2$. Let $R^*$ be the binary vector with $m$ elements, whose elements are indexed by edges in the manner compatible with $A$, such that for all $e\in E(G)$, it holds that $\widetilde{r}_e=1$ if and only if exactly one of the endpoints of $e$ is in $X_{e_W}$.
Note that $\sum_{v\in X_{e_W}}A_v = R^* + \sum_{v\in X_{e_W}}P_v\mod 2$. Moreover, since $X_{e_W}\subseteq V(G)$ is compatible with $B_W$, it holds that $\widetilde{R} = \sum_{v\in X_{e_W}}P_v\mod 2$, where $\widetilde{R}=\sumS(B_W)$. Thus, to show that $\sum_{v\in X_{e_W}}A_v = I_W\mod 2$, it is sufficient to prove that $R^*+\widetilde{R}=I_W\mod 2$. Thus, we need to show that the following two conditions are satisfied.
\begin{enumerate}
\item For each column $C^e\in F_W\cup\{W\}$ of $A$, it holds that $r^*_e + \widetilde{r}_e = 1 \mod 2$.
\item For each column $C^e\notin F_W\cup\{W\}$ of $A$, it holds that $r^*_e + \widetilde{r}_e = 0 \mod 2$.
\end{enumerate}

First, let us consider a column $C^e\in F_W\cup\{W\}$ of $A$. 
If $r^*_e=1$, then exactly one of the endpoints of $e$ belongs to $X_{e_W}$. In this case, since $e$ is expensive, we have that $\widetilde{r}_e=0$, and therefore $r^*_e + \widetilde{r}_e = 1 \mod 2$. Now, if $r^*_e=0$, then both endpoints of $e$ belong to either $X_{e_W}$ or $\bar{X}_{e_W}$. In this case, since $e$ is expensive, we have that $\widetilde{r}_e=1$, and therefore $r^*_e + \widetilde{r}_e = 1 \mod 2$.

Second, let us consider a column $C^e\notin F_W\cup\{W\}$ of $A$. 
If $r^*_e=1$, then exactly one of the endpoints of $e$ belongs to $X_{e_W}$. In this case, since $e$ is cheap, we have that $\widetilde{r}_e=1$, and therefore $r^*_e + \widetilde{r}_e = 0 \mod 2$. Now, if $r^*_e=0$, then both endpoints of $e$ belong to either $X_{e_W}$ or $\bar{X}_{e_W}$. In this case, since $e$ is cheap, we have that $\widetilde{r}_e=0$, and therefore $r^*_e + \widetilde{r}_e = 0 \mod 2$.
\end{proof}

\begin{lemma}\label{lem:cut2}
If the input instance is a yes-instance, then there exist a parity restriction and a solution $S\subseteq E(G)\setminus\edges(F)$ compatible with it.
\end{lemma}

\begin{proof}
Suppose that the input instance is a yes-instance. Then, by Lemma \ref{lem:incidence}, there exists $F\subseteq E\setminus T$ such that $|F|\leq k$ and for each terminal $W\in T$, there exists a subset $F_W\subseteq F$ such that the characteristic vector of $F_W\cup\{W\}$, denoted by $I_W$, belongs to $\spanS(V)$. Since we consider the field $GF(2)$, this implies that there exists $V_W\subseteq V$ such that $\sum_{U\in V_W}U = I_W\mod 2$. First, we define $S=\edges(F)$. For each $W\in T$, we define $S_W=\edges(F_W)\cup\{e_W\}$, where we let $e_W$ denote the edge satisfying $W=A^{e_W}$. Moreover, we let $X_W$ denote the set of vertices that are the indices of the rows in $V_W$. Accordingly, we set $\bar{X}_W=V(G)\setminus X_W$. Finally, we let $B_W$ denote the binary vector with $t$ elements whose $i^{\mathrm{th}}$ element is $|\{v\in X_W: \type(v)=i\}|\mod 2$. We have thus defined a parity restriction~$h$. 

Since $|F|\leq k$, it holds that $|S|\leq k$. It remains to show that for all $W\in T$, the partition $(X_W,\bar{X}_W)$ of $V(G)$ is good with respect to $(B_W,e_W)$ and $\expens(X_W)\setminus\{e_W\}\subseteq S$. By our definition of $h$, it is clear that for all $W\in T$, it holds that $X_W$ is compatible with $B_W$. Moreover, since $F\cap T=\emptyset$, we have that $S\cap \edges(T)=\emptyset$. Thus, it is sufficient to show that for all $W\in T$, it holds that $\expens(X_W)=S_W$.

Next, fix some terminal $W\in T$. Let $R^*$ be the binary vectors with $m$ elements, whose elements are indexed by edges in the manner compatible with $A$, such that for all $e\in E(G)$, it holds that $r^*_e=1$ if and only if exactly one of the endpoints of $e$ is in $X_{e_W}$. Moreover, let $\widetilde{R}=\sumS(B_W)$. Recall that $\sum_{U\in V_W}U = I_W\mod 2$. Now, note that $\sum_{U\in V_W}U = \sum_{v\in X_W}A_v = R^* + \sum_{v\in X_W}P_v = R^* + \widetilde{R}\mod 2$. Therefore, we have that $R^* + \widetilde{R} = I_W\mod 2$.

First, consider an edge $e\in E(G)$ with both endpoints in either $X_W$ or $\bar{X}_W$. In this case, $r^*_e=0$. On the one hand, if $e$ is expensive, then $\widetilde{r}_e=1$. In this case, since $R^* + \widetilde{R} = I_W\mod 2$, we have that $e\in S_W$. On the other hand, if $e$ is cheap, then $\widetilde{r}_e=0$. In this case, since $R^* + \widetilde{R} = I_W\mod 2$, we have that $e\notin S_W$.

Now, consider an edge $e\in E(G)$ with exactly one endpoint in $X_W$. In this case, $r^*_e=1$. On the one hand, if $e$ is expensive, then $\widetilde{r}_e=0$. In this case, since $R^* + \widetilde{R} = I_W\mod 2$, we have that $e\in S_W$. On the other hand, if $e$ is cheap, then $\widetilde{r}_e=1$. In this case, since $R^* + \widetilde{R} = I_W\mod 2$, we have that $e\notin S_W$.
\end{proof}

In other words, since there are only ${2^t}^{|T|}\leq 2^{k 2^r}$ parity restrictions, to prove Theorem \ref{thmdual}, it is sufficient to prove the following lemma, on which we focus next.

\begin{lemma}\label{lem:dual}
Given a parity restriction, it is possible to decide in time $2^{2^{\OO((2^r+k^2)k)}}\cdot (n+m)^{\OO(1)}$ whether there exists a solution that is compatible with it.
\end{lemma}

From now on, we fix a parity restriction $h$.

\subsection{\edc: Graph Problem}

Overall, we observe that we can now translate our input instance to an instance of a graph problem. To formulate this observation precisely, assume that we are given a graph $G$ with $n$ vertices and $m$ edges, non-negative integers $k$ and $t$, a partition $V(G)=(V_1,V_2,\ldots,V_t)$, a set $T\subseteq E(G)\cup D$  where $D$ is a set of dummy elements (which are not edges in $E(G)$), a binary vector $(b^e_1,b^e_2,\ldots,b^e_t)$ for $e\in T$, and a function $f_e: E(G)\rightarrow \{0,1\}$ for $e\in T$. The necessity of introducing the dummy-set $\{1,2,\ldots,d\}$ will be cleared in Section \ref{sec:connectivity}. We need the following definitions, translating previous definitions to the given simpler setting.

\begin{definition}
Given $e\in T$ and a partition $(X,\bar{X})$ of $V(G)$, we say that an edge $e'\in E(G)$ {\em contributes to $(e,(X,\bar{X}))$} if one of the three following conditions is true:
\begin{enumerate}
\item Both endpoints of $e'$ belong to $X$ and $f_{e}(e')=1$.
\item Both endpoints of $e'$ belong to $\bar{X}$ and $f_{e}(e')=1$.
\item Exactly one of the endpoints of $e'$ belongs to $X$ and $f_{e}(e')=0$.
\end{enumerate}
Accordingly, we define $\cont(e,X)$ as the set of edges that contribute to $(e,(X,\bar{X}))$.
\end{definition}

\begin{definition}
Let $(X,\bar{X})$ be a partition of $V(G)$, and let $e\in T$.
We say that $(X,\bar{X})$ {\em almost fits $e$} if $T\cap\cont(e,X)=\{e\}$. Moreover, if $(X,\bar{X})$ almost fits $e$ and for all $1\leq i\leq t$, it holds that $|X\cap V_i|=b^e_i\mod 2$, then we say that $(X,\bar{X})$ fits $e$.
\end{definition}

We are now ready to define our graph problem.

\defproblemu{\edc}%
{A (multi) graph $G$ with $n$ vertices and $m$ edges, non-negative integers $k$ and $t$, a partition $(V_1,V_2,\ldots,V_t)$ of $V(G)$, a set $D\subseteq T\subseteq E(G)\cup D$ where $D$ is a set of dummy elements, a binary vector $B^e=(b^e_1,b^e_2,\ldots,b^e_t)$ for $e\in T$, and a function $f_e: E(G)\rightarrow \{0,1\}$ for $e\in T$.}%
{Is there a set $F\subseteq E(G)\setminus T$ with $|F|\leq k$ such that for each $e\in T$, there exists a partition $(X_e,\bar{X}_e)$ of $V(G)$ that fits $e$ and such that $\cont(e,X_e)\setminus\{e\}\subseteq F$?}

Indeed, given an instance of our current problem where we seek a solution compatible with $h$, we can construct an equivalent instance of \edc{}, $(G',k',t',(V_1,V_2,\ldots,V_{t'}),T',$ $\{B^e=(b^e_1,b^e_2,\ldots,b^e_{t'})\}|_{e\in T'},\{f_e\}_{e\in T'})$, by setting $G'=G$, $k'=k$, $t'=t$, for all $1\leq i\leq t'$, $V_i=\{v\in V(G): \type(v)=i\}$, $T'=\{e\in E(G): A^e\in T\}$, for all $e\in T'$, $B^e=B_W$ for $W=A^e$, and for all $e\in T'$ and $e'\in E(G)$, $f_e(e')=\widetilde{r}$ where $\widetilde{r}$ is the element of index $e'$ in $\sumS(B_W)$ for $W=A^e$. Thus, to prove Lemma \ref{lem:dual}, we can next focus on the proof of the following result.

\begin{lemma}\label{lem:edc}
\edc{} is solvable in time $2^{2^{\OO((t+k^2)|T|)}}\cdot (n+m)^{\OO(1)}$. 
\end{lemma}

We remark that in our input instance, it holds that $|T|\leq k$. However, in instances that we construct later, $k$ may become smaller, and therefore it may no longer hold that $|T|\leq k$.

\subsection{Disconnected Graphs}\label{sec:connectivity}

In what follows, it would be convenient to assume that $G$ is a connected graph.
In case $G$ is not a connected graph, we let ${\cal C}=\{C_1,C_2,\ldots,C_q\}$ for the appropriate integer $q$ denote the set of connected components of $G$, and we perform the following procedure. 
Let ${\cal G}$ be the set of all mappings $g$ that assign each $e\in T$ a binary vector $({b^e}'_1,{b^e}'_2,\ldots,{b^e}'_t)$. Note that $|{\cal G}|\leq 2^{t|T|}$.  Then, for each $C\in{\cal C}$, $k'\leq k$ and $g\in{\cal G}$, we define the instance $I(C,k',g)=(C,k',t'=t,(V_1\cap V(C),V_2\cap V(C),\ldots,V_t\cap V(C)),T,\{{B^e}'=g(e)\}|_{e\in T},\{f_e\}_{e\in T})$. Here, the dummy elements are the edges in $T$ that do not belong to $E(C)$. Assuming that we can solve these instances (which is the objective of the rest of this section), we show that we can solve the original instance. For this purpose, we only need the computation based on dynamic programming that is described next.

Let $y(C,g)$ be the smallest $k'\leq k$ such that $I(C,k',g)$ is a yes-instance, where if such $k'$ does not exist, set $y(C,g)=k+1$. Moreover, given $g,g',g''\in{\cal G}$, we denote $g=g'+g''$ if for all $e\in T$, it holds that $g(e)=g'(e)+g''(e)\mod 2$.
We construct a matrix M with the entry $[i,g]$ for all $1\leq i\leq q$ and $g\in{\cal G}$. Now, we compute M$[i,g]$ as follows. At the basis, M$[1,g]=y(C_1,g)$. Now, suppose that $2\leq i\leq q$. Then, M$[i,g] = \min_{g',g''\in {\cal G}, g=g'+g''}\mathrm{M}[i-1,g'] + y(C_i,g'')$. It is straightforward to verify that the solution is yes if and only if M$[q,g]\leq k$ where $g$ is the function mapping each $e$ to $B^e$.

From now on, we implicitly assume that $G$ is a connected graph. Moreover, all of the problems instances we will construct later will correspond to connected graphs. In other words, the above process that handles disconnected graphs is only done once.

\subsection{The Recursive Procedure}

To solve \edc, we intend to apply a recursive procedure. When constructing a solution to the current instance, we need to integrate solutions for subproblems returned by recursive calls. However, to be able to integrate such solutions, we cannot just use arbitrary solutions, but we need to have solutions of certain kinds.
To this end, we need to restrict the solutions for our subproblems, or in other words, to solve a more general problem that allows us to impose such restrictions by specifying them as part of the input. For this purpose, we define the \aedc{} problem. The output will be a set of ``solutions'', one for ``restriction''. First, we formally define the meaning of a solution. Here, recall that $\widehat{B}$ is the set of all binary vectors with $t$ elements, and that a function $h:T\rightarrow\widehat{B}$ is called a parity restriction.

\begin{definition}
Let $G$ be a graph, $k$ and $t$ be two non-negative integers, $(V_1,V_2,\ldots,V_t)$ be a partition of $V(G)$, $D\subseteq T\subseteq E(G)\cup D$ be a set where $D$ is a set of dummy elements, $f_e: E(G)\rightarrow \{0,1\}$ be a function for $e\in T$, $W$ be a set, and $W^e_1,W^e_2\subseteq V(G)$ be two sets for $e\in T$.

Then, given a parity restriction $h$, together with set of partitions, $\{(L^e,R^e)\}|_{e\in T}$, of $W$, we say that a set $F\subseteq E(G)\setminus T$ with $|F|\leq k$ and a partition $(X_e,\bar{X}_e)$ of $V(G)$ for $e\in T$ is an {\em $(h,\{(L^e,R^e)\}|_{e\in T})$-solution} if for each $e\in T$, it holds that $(X_e,\bar{X}_e)$ fits $e$, $L^e\subseteq X_e$, $R^e\subseteq \bar{X}_e$, $\cont(e,X_e)\setminus\{e\}\subseteq F$, $W^e_1\subseteq X_e$ and $W^e_2\subseteq \bar{X}_e$.
We say that an $(h,\{(L^e,R^e)\}|_{e\in T})$-solution $(F,\{(X_e,\bar{X}_e)\}|_{e\in T})$ is {\em optimal} if there is no $(h,\{({L^e},{R^e})\}|_{e\in T})$-solution $(F',\{(X'_e,\bar{X}'_e)\}|_{e\in T})$ such that $|F'|<|F|$.
\end{definition}

Let us fix $p=2(k+1)$.

\defproblemuTask{\aedc}%
{A (multi) graph $G$ with $n$ vertices and $m$ edges, non-negative integers $k$ and $t$, a partition $(V_1,V_2,\ldots,V_t)$ of $V(G)$,
a set $D\subseteq T\subseteq E(G)\cup D$ where $D$ is a set of dummy elements, 
a function $f_e: E(G)\rightarrow \{0,1\}$ for $e\in T$, a set $W\subseteq V(G)$ with $|W|\leq 2p$, and sets $W^e_1,W^e_2\subseteq V(G)$ for $e\in T$.}%
{For every parity restriction $h$ and a set of partitions, $\{(L^e,R^e)\}|_{e\in T}$, of $W$, determine whether there exists an $(h,\{(L^e,R^e)\}|_{e\in T})$-solution, and if the answer is positive, return an optimal one.}

We call our recursive procedure, designed in the rest of this section, \alg{Recurs}. The form of a call to this procedure is \alg{Recurs}$(G,k,t,(V_1,V_2,\ldots,V_t),T,\{f_e\}|_{e\in T},W,\{(W^e_1,W^e_2)\}|_{e\in T})$. Before we proceed to describe a recursive call, we remark that the solution for our original problem is the answer returned by $(G,k,t,(V_1,V_2,\ldots,V_t),T,\{f_e\}|_{e\in T},\emptyset,\{(\emptyset,\emptyset)\}|_{e\in T})$ with respect to the parity restriction $h$ such that $h(e)=B^e$ for $e\in T$, where $(G,k,t,(V_1,V_2,\ldots,V_t),T,\{B^e\}|_{e\in T},\{f_e\}|_{e\in T})$ is the input instance. Let $\alpha$ be a fixed constant whose exact value is determined later (more precisely, we give several lower bounds for $\alpha$ throughout the paper, and the exact value of $\alpha$ is the maximum among them). From now on, to prove Lemma \ref{lem:edc}, we focus on the proof of the following result.

\begin{lemma}\label{lem:aedc}
\aedc{} is solvable in time $\tau(n,m,k,t,|T|)=2^{2^{\alpha(t+k^2)|T|}}\cdot (n+m)^{\alpha}$.
\end{lemma}

At a given call, we will ensure that each recursive call that is made corresponds to a graph smaller than the current one (therefore, the depth of the recursion will be bounded by $(n+m)^{O(1)}$). In particular, this means that to prove that the procedure is correct and runs in the desired time (according to Lemma \ref{lem:aedc}), at each given call, we assume (as the inductive hypothesis), that the solutions returned by further recursive calls are correct and that these additional calls are performed in the desired time.

\subsection{Good Separation}

On a high-level, we use the method of recursive understanding \cite{DBLP:journals/siamcomp/ChitnisCHPP16}, in which we incorporate various new, delicate subroutines. Informally, this means that at the basis, we are going to deal with a ``highly-connected'' or a small graph, and at each step where our graph is not highly-connected, we will break it using a very small number of edges into two graphs that are both neither too small nor too large.

\begin{definition}
Let $G$ be a connected graph. A partition $(X,Y)$ of $V(G)$ is called {\em $(q,p)$-good edge separation} if
\begin{itemize}
\item $|X|,|Y|>q$.
\item $|E(X,Y)|\leq p$.
\item $G[X]$ and $G[Y]$ are connected graphs.
\end{itemize}
\end{definition}

Roughly speaking, a graph $G$ is unbreakable if every partition of $V(G)$ with few edges going across must contain a large chunk of $V(G)$ in one of its two sets. Intuitively, this means that $G$ is ``highly-connected'': any attempt to ``break'' it severely by using only few edges is futile. Formally, an unbreakable graph is defined as follows.

\begin{definition}\label{def:unbreakable}
A graph $G$ is {\em $(q,p)$-unbreakable} if it does not have a {\em $(q,p)$-good edge separation}.
\end{definition}

If a graph $G$ is not $(q,p)$-unbreakable, we say that it is $(q,p)$-breakable. Chitnis et al.~\cite{DBLP:journals/siamcomp/ChitnisCHPP16} proved the following result for the appropriate constant $\beta$.

\begin{proposition}[\cite{DBLP:journals/siamcomp/ChitnisCHPP16}]\label{prop:goodSep}
There exists a deterministic algorithm that given a connected graph $G$ along with integers $q$ and $p$, in time $\beta\cdot 2^{\min\{q,p\}\cdot\log(q+p)}\cdot (n+m)^3\log(n+m)$ either finds a $(q,p)$-good edge separation, or correctly concludes that $G$ is $(q,p)$-unbreakable.
\end{proposition}

Next, we denote $q=2^{2^{\lambda(t+k^2)|T|}}$ where $\lambda$ is fixed in Section \ref{sec:replace}, and recall that $p=2(k+1)$. The reason for the choice of these values will be cleared in the rest of the design of \alg{Recurs}. Notice that $\beta\cdot 2^{\min\{q,p\}\cdot\log(q+p)}=2^{2^{\widehat{\lambda}(t+k^2)|T|}}$ for the appropriate constant $\widehat{\lambda}$.

\subsection{Small Graphs}\label{sec:small}

Denote $s=q^4$. Here we handle the case where $n\leq s$. This is the base case of our recursion. For this base case, we prove the following result, where $\eta$ is a constant (to be determined in this section).

\begin{lemma}\label{lem:small}
\aedc{} where $n\leq s$ is solvable in time $\tau_\mathrm{s}(n,m,k,t,$ $|T|)=2^{2^{\eta(t+k^2)|T|}}\cdot (n+m)^{\eta}$.
\end{lemma}

Given two vertices $v,u\in V(G)$ and a function $g: T\rightarrow\{0,1\}$, let $E^g(v,u)$ denote the set of all edges $e'$ in $E(G)\setminus T$ between $v$ and $u$ such that for all $e\in T$, it holds that $f_e(e')=g(e)$. We start with the following lemma. This lemma does not rely on the assumption $|V(G)|\leq s$, and therefore will also be able to use it at a later point.

\begin{lemma}\label{lem:multiplicity}
Let $v,u\in V(G)$ and $g: T\rightarrow\{0,1\}$ such that $|E^g(v,u)|>k$. Then, by removing an edge $e$ from $E^g(v,u)$ from $G$ (updating the instance accordingly), we obtain an equivalent instance.
\end{lemma}

\begin{proof}
Consider some parity restriction $h$ and a set of partitions $\{(L^{e'},R^{e'})\}|_{e'\in T}$ of $W$. On the one hand, suppose that there exists an $(h,\{(L^e,R^e)\}|_{e\in T})$-solution for the original instance, and consider some such solution $(F,\{(X_{e'},\bar{X}_{e'})\}|_{e'\in T})$. Then, it is clear that $(F\setminus\{e\},\{(X_{e'},\bar{X}_{e'})\}|_{e'\in T})$ is a solution to the new instance.

On the other hand, suppose that there exists an $(h,\{(L^{e'},R^{e'})\}|_{e'\in T})$-solution for the new instance, and consider some such {\em optimal} solution $(F,\{(X_{e'},\bar{X}_{e'})\}|_{e'\in T})$. If is sufficient to prove that $e\notin F$, since then $(F,\{(X_{e'},\bar{X}_{e'})\}|_{e'\in T})$ is also a solution to the original instance. Suppose, by way of contradiction, that $e\in F$. Then, since $(F,\{(X_{e'},\bar{X}_{e'})\}|_{e'\in T})$ is optimal, there exists $e'\in T$ such that $e\in\cont(e',X_{e'})\setminus\{e'\}$. However, for all $\widehat{e}\in E^g(v,u)$, it holds that $f_{e'}(\widehat{e})=f_{e'}(e)$. Thus, $E^g(v,u)\subseteq \cont(e',X_{e'})\setminus\{e'\}$. However, $\cont(e',X_{e'})\setminus\{e'\}\subseteq F$, $|E^g(v,u)|>k$ and $|F|\leq k$, and therefore we have reached a contradiction.
\end{proof}

In particular, by exhaustive application of Lemma \ref{lem:multiplicity}, we have that between every pair of vertices there are at most $k2^{|T|}$ non-terminal edges. Since $n\leq s$, we have the following observation (for the current graph).

\begin{observation}
$|E(G)\setminus T|\leq k2^{|T|}\cdot s^2$.
\end{observation}

Next, let ${\cal F}$ be the family of all subsets of $E(G)\setminus T$ of size at most $k$. Notice that $|{\cal F}|\leq (k2^{|T|}\cdot s^2)^k$. Thus, we can examine each set $F\in{\cal F}$. In what follows, we also examine each pair of a parity restriction $h$ and a set partitions $\{(L^e,R^e)\}|_{e\in T}$ of $W$ separately. This means that we can next fix some $F\in{\cal F}$, parity restriction $h$ and a set of partitions $\{(L^e,R^e)\}|_{e\in T}$ of $W$ (there are at most $(k2^{|T|}\cdot s^2)^k\cdot 2^{t|T|+2p|T|}$ such triples), and to prove Lemma \ref{lem:small}, it is sufficient that we next prove the following result, where $\eta'$ is a constant implicitly fixed by the discussion below (according to which the value of $\eta$ is fixed -- that is, given the following result, it is clear that we can choose a large enough $\eta$ so that Lemma \ref{lem:small} is proved).

\begin{lemma}\label{lem:fSuffice}
It is possible to determine in time $2^{2^{\eta'(t+k^2)|T|}}\cdot (n+m)^{\eta'}$ whether there exists an $(h,\{(L^e,R^e)\}|_{e\in T})$-solution $(F',\{(X'_e,\bar{X}'_e)\}|_{e\in T})$ such that $f'=F$, and if the answer is positive, return such a solution.
\end{lemma}

Since $G$ is a connected graph and $|F|\leq k$, we have that $G\setminus (F\cup T)$ contains at most $|T|+k+1$ connected components, which we denote by $C_1,C_2,\ldots,C_r$ for the appropriate $r$. Moreover, for all $1\leq i\leq r$, arbitrarily choose a vertex $v_i\in V(C_i)$. Let ${\cal G}$ be the family of functions $g: T\times\{1,2,\ldots,r\}\rightarrow\{0,1\}$. Now, we need the following definition.

\begin{definition}
Given a function $g\in{\cal G}$, we say that an $(h,\{(L^e,R^e)\}|_{e\in T})$-solution, $(F',\{(X'_e,$ $\bar{X}'_e)\}|_{e\in T})$, is {\em nice} if $F'=F$ and for all $e\in T$ and $1\leq i\leq r$, it holds that if $g(e,i)=0$ then $v_i\in X'_e$ and otherwise $v_i\in \bar{X}'_e$.
\end{definition}

Since $|{\cal G}|\leq 2^{r|T|}\leq 2^{(k+|T|)|T|}$,\footnote{We remark that here, since it may be true $|T|$ is significantly larger than $k$, the bound may be larger than the running time we desire to obtain. However, we only decrease $k$ when we handle disconnected graphs, and thus by defining the problem as an optimization problem a priori, it is possible to assume that $|T|\leq k$. Therefore, for the sake of clarity, we ignore this issue.} we can examine every function $g\in{\cal G}$. Thus, we can fix some $g\in{\cal G}$, and to prove Lemma \ref{lem:fSuffice}, we can next focus on proving the following lemma.

\begin{lemma}\label{lem:sufficeNice}
It is possible to determine whether there exists a nice $(h,\{(L^e,R^e)\}|_{e\in T})$-solution, and if the answer is positive, return such a solution, in polynomial time.
\end{lemma}

Now, by the definition of an $(h,\{(L^e,R^e)\}|_{e\in T})$-solution, we have the following observation.

\begin{observation}\label{obs:contrNice}
Let $(F,\{(X_e,\bar{X}_e)\}|_{e\in T})$ be a nice solution. Then, for all $e\in T$, it holds that $\cont(e,X_e)\setminus\{e\}\subseteq F$.
\end{observation}

We also need the following insight.

\begin{lemma}\label{lem:onePartition}
Given $e\in T$ and $1\leq i\leq r$, there exists at most one partition $(X_e,\bar{X}_e)$ of $V(C_i)$ such that $\cont(e,X_e)\setminus\{e\}\subseteq F$ and if $g(e,i)=0$ then $v_i\in X_e$ and otherwise $v_i\in \bar{X}_e$. Moreover, in polynomial time we can detect whether such a partition exists, and if the answer is positive, return such a partition.
\end{lemma}

\begin{proof}
Let us fix $e\in T$ and $1\leq i\leq r$. We compute a spanning tree of $C_i$, and root it at $v_i$. Now, we traverse this tree from the root to the leaves. Clearly, if $g(e,i)=0$ then we determine that $v_i\in X_e$ and otherwise we determine that $v_i\in \bar{X}_e$. When we reach a vertex $v$ that is not the root $v_i$, we assume that we have already correctly determined whether its parent lies in $X_e$ or $\bar{X}_e$. Now, suppose that we currently examine some vertex $v$, and let $u$ be its parent. We have the following cases.
\begin{itemize}
\item There exist two edges, $e'$ and $e''$, between $v$ and $u$ such that $f_e(e')\neq f_e(e'')$. In this case, there does not exist a partition $(X_e,\bar{X}_e)$ of $V(C_i)$ such that $\cont(e,X_e)\setminus\{e\}\subseteq F$ (recall that $C_i$ does not contains any edge from $F\cup T$).
\item For every edge $e'$ between $v$ and $u$, it holds that $f_e(e')=0$. In this case, if $u\in X_e$, then we must insert $v$ into $X_e$, and otherwise we must insert $v$ into $\bar{X}_e$.
\item For every edge $e'$ between $v$ and $u$, it holds that $f_e(e')=1$. In this case, if $u\in X_e$, then we must insert $v$ into $\bar{X}_e$, and otherwise we must insert $v$ into $X_e$.
\end{itemize}

After this process is finished, we have either obtained a partition $(X_e,\bar{X}_e)$ of $V(C_i)$ such that if $g(e,i)=0$ then $v_i\in X_e$ and otherwise $v_i\in \bar{X}_e$ or correctly determined that a partition of the desired type does not exist. Suppose that the former case is true. Then, if there exists a partition $(X'_e,\bar{X}'_e)$ of $V(C_i)$ such that $\cont(e,X'_e)\setminus\{e\}\subseteq F$ and if $g(e,i)=0$ then $v_i\in X_e$ and otherwise $v_i\in \bar{X}_e$, then it must hold that $(X'_e,\bar{X}'_e)=(X_e,\bar{X}_e)$. Thus, it remains to check whether $\cont(e,X_e)\setminus\{e\}\subseteq F$, which can be done in polynomial time.
\end{proof}

Thus, to prove Lemma \ref{lem:sufficeNice}, we apply Lemma \ref{lem:onePartition} for each $e\in T$ and $1\leq i\leq r$. By Observation \ref{obs:contrNice}, if in at least one of the applications, a partition is not found, we deduce that there does not exist a nice solution. Otherwise, by combining the partitions we obtained, we get a pair $(F,\{(X_e,\bar{X}_e)\}|_{e\in T})$, and we can clearly determine in polynomial time whether it is indeed a nice solution. This concludes the proof of Lemma \ref{lem:sufficeNice}, and therefore also of Lemma \ref{lem:small}.

\subsection{Unbreakable Graphs}\label{sec:unbrekable}

Let us first handle the case where $G$ is $(q,p)$-unbreakable. By Lemma \ref{prop:goodSep}, we can detect that $G$ is $(q,p)$-unbreakable in time $2^{2^{\widehat{\lambda}(t+p)|T|}}\cdot (n+m)^3\log(n+m)$. In this case, we only make recursive calls that are associated with instances handled in Section \ref{sec:small}. In what follows, we examine each pair of a parity restriction $h$ and a set of partitions $\{(L^e,R^e)\}|_{e\in T}$ of $W$ separately. Then, for every pair $(h,\{(L^e,R^e)\}|_{e\in T})$, by choosing a large enough $\alpha$ in advance, to handle the current case, it is sufficient that we next prove the following lemma where $\alpha^*=\alpha/2$.

\begin{lemma}\label{lem:sufficeunbreak}
In time $\tau^*(n,m,k,t,|T|)=2^{2^{\alpha^*(t+k^2)|T|}}\cdot (n+m)^{\alpha^*}$, determine whether there exists an $(h,\{(L^e,R^e)\}|_{e\in T})$-solution, and if the answer is positive, return an optimal one.
\end{lemma}

\subsubsection{Preliminary Partitions}\label{sec:prelimPart}

For each $e\in T$, we start by obtaining a partition $(Y_e,\bar{Y}_e)$ of $V(G)$ that has only some of the properties that we want. These properties are given by the following definition.

\begin{definition}\label{def:ePrelim}
Given $e\in T$, we say that a partition $(Y_e,\bar{Y}_e)$ of $V(G)$ is {\em $e$-preliminary} if it almost fits $e$ and $|\cont(e,Y_e)\setminus\{e\}|\leq k$.
\end{definition}

For each $e\in T$, the computation of an $e$-preliminary partition $(Y_e,\bar{Y}_e)$ (if one exists) is given by the following lemma.

\begin{lemma}\label{lem:perlimPartition}
Given $e\in T$, an $e$-preliminary partition $(Y_e,\bar{Y}_e)$ can be computed in time $\beta\cdot 2^{k}\cdot(n+m)^{2}$ (if one exists) for some fixed constant $\beta$.
\end{lemma}

The rest of Section \ref{sec:prelimPart} is devoted to the proof of Lemma \ref{lem:perlimPartition}. To compute $(Y_e,\bar{Y}_e)$, we construct an instance of the {\sc Edge Odd Cycle Transversal (EOCT)} problem (also known as {\sc Edge Bipartization}), which is defined as follows.

\defproblemu{{\sc Edge Odd Cycle Transversal (EOCT)}}%
{A (multi) graph $G'$ with $n'$ vertices and $m'$ edges, and a parameter $k'$.}%
{Is there a subset $S\subseteq E(G')$ of size at most $k'$ such that $G'\setminus S$ is a bipartite graph?}

We construct an instance $(G',k')$ of {\sc EOCT} as follows. First, we set $k'=k$. Now, we initialize $G'$ to be the graph $G\setminus T$, and then perform the following modifications.
\begin{itemize}
\item Define $E_0=\{e'\in E(G)\setminus T: f_e(e')=0\}$. Then, for all $e'\in E_0$, subdivide $e'$. Let $e'_1$ and $e'_2$ be the two resulting edges.% Let $v_{e'}$ denote the newly added vertex.
\item If $e\in E(G)$, we have two cases:
	\begin{itemize}
	\item If $f_e(e)=0$, then add $k+1$ parallel edges between the endpoints of $e$ in $G'$. Note that while $e$ does not exist in $G'$, its endpoints do exist, and therefore this operation is well-defined. Let $N_e$ denote the set of newly added edges.
	\item If $f_e(e)=1$, then add a new vertex, $v_e$, and for each endpoint of $e$, add $k+1$ parallel edges between this endpoint and $v_e$. Let $N_e$ denote the set of newly added edges.
	\end{itemize}
\item For each $e'\in (T\cap E(G))\setminus\{e\}$, we have two cases:
	\begin{itemize}
	\item If $f_e(e')=0$, then add a new vertex, $v_e$, and for each endpoint of $e$, add $k+1$ parallel edges between this endpoint and $v_e$. Let $N_{e'}$ denote the set of newly added edges.
	\item If $f_e(e')=1$, then then add $k+1$ parallel edges between the endpoints of $e$ in $G'$. Note that while $e'$ does not exist in $G'$, its endpoints do exist, and therefore this operation is well-defined. Let $N_{e'}$ denote the set of newly added edges.
	\end{itemize}
\end{itemize}
 
Let $(Y_e,\bar{Y}_e)$ be an $e$-preliminary partition. Let $U_Y$ denote the set of each vertex in $V(G')\setminus V(G)$ such that the two vertices adjacent to it in $G'$ (which by our construction belong to $V(G)$), belong to $\bar{Y}_e$. Symmetrically, let $U_Y$ denote the set of each vertex in $V(G')\setminus V(G)$ such that the two vertices adjacent to it in $G'$ belong to $Y_e$. We define a partition $(A_Y,B_Y)$ of $V(G')$ as follows. We set $A_Y=Y_e\cup U_Y$ and $B=\bar{Y}_e\cup \bar{U}_Y\cup (V(G')\setminus (V(G)\cup U_Y\cup \bar{U}_Y))$.
Now, let $S_Y$ denote the set of each edge in $E(G')$ that has both its endpoints in either $A_Y$ or $B_Y$.

\begin{lemma}\label{lem:partToEOCT}
It holds that $|S_Y|\leq k$ and the graph $G'\setminus S_Y$ is a bipartite graph.
\end{lemma}

\begin{proof}
By the definition of $S_Y$, it is clear that $G'\setminus S_Y$ is a bipartite graph. To prove that $|S_Y|\leq k$, it is sufficient to construct an injective function $g: S_Y\rightarrow \cont(e,Y_e)\setminus\{e\}$. Towards this, we first need the following claims.
\begin{enumerate}
\item\label{arg:eoct1} $S_Y\cap T=\emptyset$: Since $E(G')\cap T=\emptyset$, this claim is true.
\item\label{arg:eoct2} For all $e'\in E_0$, it holds that $|\{e'_1,e'_2\}\cap S_Y|\leq 1$, and if $|\{e'_1,e'_2\}\cap S_Y|=1$ then $e'\in\cont(e,Y_e)$: Consider some edge $e'\in E_0$. If both of its endpoints belong to $Y_e$, then both of these vertices belong to $A_Y$ and the new vertex common to $e'_1$ and $e'_2$ belongs to $\bar{U}_Y\subseteq B_Y$, and hence $\{e'_1,e'_2\}\cap S_Y=\emptyset$. Symmetrically, if both of endpoints of $e'$ belong to $Y_e$ then again $\{e'_1,e'_2\}\cap S_Y=\emptyset$. Now, suppose that one endpoint of $e'$, $x$, is in $Y_e$ and the other, $y$, is in $\bar{Y}_e$. Then, $e'\in\cont(e,Y_e)$. Moreover, the new vertex common to $e'_1$ and $e'_2$ belongs to $B_Y$, and therefore only one edge in $\{e'_1,e'_2\}\cap S_Y$ --- the one between this new vertex and $y$ --- belongs to $S_Y$.
\item\label{arg:eoct3} If $e\in E(G)$ then $S_Y\cap N_e=\emptyset$: Note that since $(Y_e,\bar{Y}_e)$ is an $e$-preliminary partition, it holds that $e\in\cont(e,Y_e)$. First, suppose that $f_e(e)=0$. Then, since $e\in\cont(e,Y_e)$, one endpoint of $e$ is in $Y_e$ while the other is in $\bar{Y}_e$, and therefore $S_Y\cap N_e=\emptyset$. Now, suppose that $f_e(e)=0$. Then, since $e\in\cont(e,Y_e)$, both endpoints of $e$ are either in $Y_e$ or in $\bar{Y}_e$. If both endpoints are in $Y_e$, then the newly added vertex is in $\bar{U}_Y\subseteq B_Y$ and therefore $S_Y\cap N_e=\emptyset$, while if both endpoints are in $\bar{Y}_e$, then the newly added vertex is in $U_Y\subseteq A_Y$ and therefore again $S_Y\cap N_e=\emptyset$.
\item\label{arg:eoct4} For all $e'\in (T\cap E(G))\setminus\{e\}$, it holds that $N_{e'}\cap S_Y=\emptyset$: Consider some edge $e'\in (T\cap E(G))\setminus\{e\}$. The proof is symmetric to the one of Claim \ref{arg:eoct3}, but we present it for the sake of completeness. Note that since $(Y_e,\bar{Y}_e)$ is an $e$-preliminary partition, it holds that $e'\notin\cont(e,Y_e)$. First, suppose that $f_e(e')=1$. Then, since $e'\notin\cont(e,Y_e)$, one endpoint of $e'$ is in $Y_e$ while the other is in $\bar{Y}_e$, and therefore $S_Y\cap N_{e'}=\emptyset$. Now, suppose that $f_e(e')=0$. Then, since $e'\notin\cont(e,Y_e)$, both endpoints of $e'$ are either in $Y_e$ or in $\bar{Y}_e$. If both endpoints are in $Y_e$, then the newly added vertex is in $\bar{U}_Y\subseteq B_Y$ and therefore $S_Y\cap N_{e'}=\emptyset$, while if both endpoints are in $\bar{Y}_e$, then the newly added vertex is in $U_Y\subseteq A_Y$ and therefore again $S_Y\cap N_{e'}=\emptyset$.
\end{enumerate}

For each edge $\widehat{e}\in S_Y$, we define $g(\widehat{e})$ as follows. The exhaustiveness of these two cases is implied by Claims \ref{arg:eoct3} and \ref{arg:eoct4}.
\begin{itemize}
\item If $\widehat{e}\in E(G)\setminus T$, then we define $g(\widehat{e})=\widehat{e}$. We claim that $g(\widehat{e})\in\cont(e,Y_e)$. Since $\widehat{e}\in E(G)\setminus T$, it holds that $f_e(\widehat{e})=1$. Since $\widehat{e}\in S_Y$, it holds that both endpoints of $\widehat{e}$ are either in $Y_e$ or $\bar{Y}_e$. In either case, we deduce that $\widehat{e}\in\cont(e,Y_e)$. 
\item If $\widehat{e}$ is an edge of the form $e'_1$ or $e'_2$ for some $e'\in E_0$, then we define $g(\widehat{e})=e'$. By Claim \ref{arg:eoct2}, $g(\widehat{e})\in\cont(e,Y_e)$.
\end{itemize}

We have thus shown that $\image(g)\subseteq \cont(e,Y_e)$. By Claim \ref{arg:eoct1}, it holds that $e\notin\image(g)$. It remains to show that $g$ is injective. In the first case, $g$ is an identity function and its image lies in $E(G)\setminus (T\cup E_0)$. In the second case, the image lies in $E_0$, and by Claim \ref{arg:eoct2}, we thus conclude that $g$ is injective.
\end{proof}

Let $S$ be a subset of $E(G')$ of size at most $k'$ such that $G'\setminus S$ is a bipartite graph, and let $(A_S,B_S)$ be a bipartition of the vertex-set of $G'\setminus S$. Then, denote $Y_e(S)=A_S\cap V(G)$ and $\bar{Y}_e=B_S\cap V(G)$.

\begin{lemma}\label{lem:eoctToPart}
The partition $(Y_e(S),\bar{Y}_e(S))$ is an $e$-preliminary partition.
\end{lemma}

\begin{proof}
Let use denote $(Y_e,\bar{Y}_e)=(Y_e(S),\bar{Y}_e(S))$. We first argue that if $e\in E(G)$, then $e\in\cont(e,Y_e)$. To this end, we suppose that $e\in E(G)$ and consider two cases.
\begin{itemize}
\item First, suppose that $f_e(e)=0$. Since $|S|\leq k$ and there exist $k+1$ parallel edges in $G'$ between the endpoints of $e$, there exists an edge in $G'\setminus S$ between the endpoints of $e$. Thus, one endpoint of $e$ belongs to $A_S$ while the other to $B_S$. Therefore, one endpoint of $e$ belongs to $Y_e$ while the other to $\bar{Y}_e$. We thus deduce that $e\in\cont(e,Y_e)$.
\item Second, suppose that $f_e(e)=1$. Since $|S|\leq k$ and there exist $k+1$ parallel edges in $G'$ between each of the endpoints of $e$ and the newly added vertex, there exists a path on two edges in $G'\setminus S$ between the endpoints of $e$. Therefore,  both endpoints of $e$ either belong to $A_S$ or to $B_S$, which implies that both of these endpoints either belong to $Y_e$ or to $\bar{Y}_e$. We thus deduce that $e\in\cont(e,Y_e)$.
\end{itemize}

Next, we argue that for all $e'\in T\setminus\{e\}$, it holds that $e'\notin\cont(e,Y_e)$. For this purpose, consider some $e'\in  T\setminus\{e\}$. In case $e'\notin E(G)$, it is clear that $e'\notin\cont(e,Y_e)$, and therefore we next suppose that $e'\in E(G)$. We consider two cases. The arguments are symmetric to those given above, but we present them for the sake of completeness.
\begin{itemize}
\item First, suppose that $f_e(e')=1$. Since $|S|\leq k$ and there exist $k+1$ parallel edges in $G'$ between the endpoints of $e'$, there exists an edge in $G'\setminus S$ between the endpoints of $e'$. Thus, one endpoint of $e'$ belongs to $A_S$ while the other to $B_S$. Therefore, one endpoint of $e'$ belongs to $Y_e$ while the other to $\bar{Y}_e$. We thus deduce that $e'\notin\cont(e,Y_e)$.
\item Second, suppose that $f_e(e')=0$. Since $|S|\leq k$ and there exist $k+1$ parallel edges in $G'$ between each of the endpoints of $e'$ and the newly added vertex, there exists a path on two edges in $G'\setminus S$ between the endpoints of $e'$. Therefore,  both endpoints of $e'$ either belong to $A_S$ or to $B_S$, which implies that both of these endpoints either belong to $Y_e$ or to $\bar{Y}_e$. We thus deduce that $e'\notin\cont(e,Y_e)$.
\end{itemize}

So far we have shown that $(Y_e,\bar{Y}_e)$ fits $e$. It remains to show that $|\cont(e,Y_e)\setminus\{e\}|\leq k$. For this purpose, it is sufficient to construct an injective function $g: \cont(e,Y_e)\setminus\{e\}\rightarrow S_Y$. We first claim that for each $e'\in (\cont(e,Y_e)\setminus\{e\})\cap E_0$, it holds that $S_Y\cap\{e'_1,e'_2\}\neq\emptyset$. Indeed, since $e'\in (\cont(e,Y_e)\setminus\{e\})\cap E_0$, it holds that one endpoint of $e'$ is in $Y_e$ while the other is in $\bar{Y}_e$. Since $G'\setminus S$ is a bipartite graph, and $G'$ has a path on two edges, $e'_1$ and $e'_2$, between the endpoint of $e'$, we deduce that $S_Y\cap\{e'_1,e'_2\}\neq\emptyset$.

Now, for each edge $e'\in \cont(e,Y_e)\setminus\{e\}$, we define $g(e')$ as follows.
\begin{enumerate}
\item If $e'\in E_0$, then we define $g(e')$ to be an edge in $S_Y\cap\{e'_1,e'_2\}$ (recall that we have shown that $S_Y\cap\{e'_1,e'_2\}\neq\emptyset$), where if there is more then one choice, we choose one edge arbitrarily.

\item If $e'\notin E_0$, then we define $g(e')=e'$. We claim that $g(e')\in S_Y$. Since $e'\in (\cont(e,Y_e)$ $\setminus\{e\})\setminus E_0$, it holds that both endpoints of $e'$ are either in $Y_e$ or in $\bar{Y}_e$. This implies that both endpoints of $e'$ are either in $A_S$ or in $B_S$. Since $G'\setminus S$ is a bipartite graph, we have that $g(e')\in S_Y$.
\end{enumerate}

It is clear that $g$ is an injective function, and we have argued that $\image(g)\subseteq S_Y$. Overall, we conclude that $(Y_e(S),\bar{Y}_e(S))$ is an $e$-preliminary partition.
\end{proof}

Guo et al.~\cite{GuoGHNW06} (see also Pilipczuk et al.~\cite{DBLP:journals/corr/PilipczukPW15}) showed that {\sc EOCT} can be solved in time $\OO(2^k\cdot(n'+m')^{2})$. We employ this algorithm to solve our instance $(G',k')$ of {\sc EOCT}. By Lemma \ref{lem:partToEOCT}, if the answer is no, it is correct to return no. If the answer yes, by self-reduction, we also obtain a set $S$ of size at most $k'$ such that $G'\setminus S$ is a bipartite graph. Then, by Lemma \ref{lem:eoctToPart}, $(Y_e(S),\bar{Y}_e)$ is an $e$-preliminary partition. By choosing an appropriate (large enough) $\beta$, this concludes the proof of Lemma \ref{lem:perlimPartition}.

\subsubsection{Closeness of Partitions}

By using Lemma \ref{lem:perlimPartition}, \alg{Recurs} first computes an $e$-preliminary partition $(Y_e,\bar{Y}_e)$ for $e\in T$ in time $\beta\cdot 2^k\cdot(n+m)^2$. The relevance of these partitions stems from the insight that they are quite close to the actual partitions we would like to obtain. This insight is formalized in the following definition and lemma.

\begin{definition}\label{def:close}
Let $e\in T$,  and let $(Z_e,\bar{Z}_e)$ and $(Z'_e,\bar{Z}'_e)$ be two $e$-preliminary partitions. We say that $(Z_e,\bar{Z}_e)$ and $(Z'_e,\bar{Z}'_e)$ are {\em $e$-aligned} if $|(Z_e\setminus Z'_e)\cup(\bar{Z}_e\setminus \bar{Z}'_e)|\leq q$ and $|E(Z_e\setminus Z'_e)\cup(\bar{Z}_e\setminus \bar{Z}'_e),(Z_e\cap Z'_e)\cup(\bar{Z}_e\cap \bar{Z}'_e)|\leq 2(k+1)$. Moreover, we say that $(Z_e,\bar{Z}_e)$ and $(Z'_e,\bar{Z}'_e)$ are {\em $e$-opposite} if $(\bar{Z}_e,Z_e)$ and $(Z'_e,\bar{Z}'_e)$ are $e$-aligned. Finally, we say that $(Z_e,\bar{Z}_e)$ and $(Z'_e,\bar{Z}'_e)$ are {\em $e$-close} if they are $e$-aligned or $e$-opposite (or both).
\end{definition}

Note that two $e$-preliminary partitions $(Z_e,\bar{Z}_e)$ and $(Z'_e,\bar{Z}'_e)$ are $e$-opposite if and only if $|(Z_e\cap Z'_e)\cup(\bar{Z}_e\cap \bar{Z}'_e)|\leq q$ and $|E(Z_e\setminus Z'_e)\cup(\bar{Z}_e\setminus \bar{Z}'_e),(Z_e\cap Z'_e)\cup(\bar{Z}_e\cap \bar{Z}'_e)|\leq 2(k+1)$.

\begin{lemma}\label{lem:partitClose}
Let $e\in T$. Then, any two $e$-preliminary partitions $(Z_e,\bar{Z}_e)$ and $(Z'_e,\bar{Z}'_e)$ are $e$-close.
\end{lemma}

\begin{proof}
Let $(Z_e,\bar{Z}_e)$ and $(Z'_e,\bar{Z}'_e)$ be two $e$-preliminary partitions. Let us denote $A=Z_e\setminus Z'_e$, $B=\bar{Z}_e\setminus \bar{Z}'_e$, $C=Z_e\cap Z'_e$ and $D=\bar{Z}_e\cap \bar{Z}'_e$. Now, for distinct $X,Y\in\{A,B,C,D\}$, we denote $E_{XY}=E(X,Y)$, $E^0_{XY}=E_{XY}\cap\{e'\in E(G): f_e(e')=0\}$ and $E^1_{XY}=E_{XY}\setminus E^0_{XY}$. Since $(Z_e,\bar{Z}_e)$ is an $e$-preliminary partition, it holds that $|\cont(e,Z_e)\setminus\{e\}|\leq k$. Note that $E^0_{AD}\cup E^1_{AC}\cup E^0_{BC}\cup E^1_{BD}\subseteq \cont(e,Z_e)$. Therefore, $|E^0_{AD}\cup E^1_{AC}\cup E^0_{BC}\cup E^1_{BD}|\leq k+1$. Moreover, since $(Z'_e,\bar{Z}'_e)$ is an $e$-preliminary partitions, it holds that $|\cont(e,Z'_e)\setminus\{e\}|\leq k$. Note that $E^1_{AD}\cup E^0_{AC}\cup E^1_{BC}\cup E^0_{BD}\subseteq \cont(e,Z_e)$. Therefore, $|E^1_{AD}\cup E^0_{AC}\cup E^1_{BC}\cup E^1_{BD}|\leq k+1$. Observe that $E(A\cup B, C\cup D) = E_{AD}\cup E_{AC}\cup E_{BC}\cup E_{BD}$. Thus, we have that $|E(A\cup B, C\cup D)|\leq 2(k+1)$. Substituting $A,B,C$ and $D$, we have that $|E(Z_e\setminus Z'_e)\cup(\bar{Z}_e\setminus \bar{Z}'_e),(Z_e\cap Z'_e)\cup(\bar{Z}_e\cap \bar{Z}'_e)|\leq 2(k+1)$. Since $G$ is $(q,p)$-unbreakable, we deduce that $|A\cup B|\leq q$ or $|C\cup D|\leq q$ (or both). Substituting $A,B,C$ and $D$, we have that $|(Z_e\setminus Z'_e)\cup(\bar{Z}_e\setminus \bar{Z}'_e)|\leq q$ or $|(Z_e\cap Z'_e)\cup(\bar{Z}_e\cap \bar{Z}'_e)|\leq q$ (or both). Therefore, $(Z_e,\bar{Z}_e)$ and $(Z'_e,\bar{Z}'_e)$ are $e$-close.
\end{proof}

From Lemma \ref{lem:partitClose} we obtain the following corollary, which provides insight how to exploit the partitions computed in Section \ref{sec:prelimPart}.

\begin{corollary}\label{cor:partitClose}
For any parity restriction $h$ and a set of partitions $\{(L^e,R^e)\}|_{e\in T}$ of $W$, if there exists an $(h,\{(L^e,R^e)\}|_{e\in T})$-solution $(F,\{(X_e,\bar{X}_e)\}|_{e\in T})$, then for all $e\in T$, it holds that $(X_e,\bar{X}_e)$ and $(Y_e,Y'_e)$ are $e$-close.
\end{corollary}

By exhaustive search, we can examine each set $\{(Y'_e,\bar{Y}'_e)\}|_{e\in T}$ such that for all $e\in T$, it holds that $(Y'_e,\bar{Y}'_e)\in\{(Y_e,\bar{Y}_e),(\bar{Y}_e,Y_e)\}$. Indeed, there are only $2^{|T|}$ such sets. Now, fix some such set ${\cal Y}=\{(Y'_e,\bar{Y}'_e)\}|_{e\in T}$. Accordingly, we have the following definition.

\begin{definition}
We say that an $(h,\{(L^e,R^e)\}|_{e\in T})$-solution $(F,\{(X_e,\bar{X}_e)\}|_{e\in T})$ is {\em aligned} if for all $e\in T$, it holds that $(X_e,\bar{X}_e)$ and $(Y_e,Y_e)$ are $e$-aligned.
Moreover, we say that an $(h,\{(L^e,R^e)\}|_{e\in T})$-solution $(F,\{(X_e,\bar{X}_e)\}|_{e\in T})$ is {\em alignment optimal} is there is no aligned $(h,\{(L^e,R^e)\}|_{e\in T})$-solution $(F',\{(X'_e,\bar{X}'_e)\}|_{e\in T})$ such that $|F'|<|F|$.
\end{definition}

Therefore, to solve the current instance in time $\tau(n,m,k,t,|T|)$, by choosing a large enough $\alpha$ in advance, it is sufficient that we next prove the following lemma, where $\alpha'=\alpha^*/2$.

\begin{lemma}\label{lem:alignSuffice}
The following computation can be performed in time $\tau'(n,m,$ $k,t,|T|)=2^{2^{\alpha'(t+k^2)|T|}}\cdot (n+m)^{\alpha'}$.
If there exists no $(h,\{(L^e,R^e)\}|_{e\in T})$-solution, then answer no; if there exists an aligned $(h,\{(L^e,R^e)\}|_{e\in T})$-solution, then return an alignment optimal $(h,\{(L^e,R^e)\}|_{e\in T})$-solution; else, either return no or an $(h,\{(L^e,R^e)\}|_{e\in T})$-solution.
\end{lemma}

\subsubsection{Highlighting the Solution}

We are now going to color the vertices in $V(G)$ so that if an aligned $(h,\{(L^e,R^e)\}|_{e\in T})$-solution exists, it would be easy to detect an alignment optimal $(h,\{(L^e,R^e)\}|_{e\in T})$-solution. For this purpose, we need the following definition and result, which present the tool we will use in order to color the graph.

\begin{definition}
Given integers $k'\leq n'$, a set ${\cal C}'$ of functions $c: \{1,2,\ldots,n'\}\rightarrow\{1,2,\ldots,k'\}$ is an {\em $(n',k')$-family of hash functions} if for every set $S\subseteq \{1,2,\ldots,n'\}$ of size $k'$, there exists a function $c\in {\cal C}'$ such that $c|_S$ is an injective function.
\end{definition}

\begin{definition}\label{dfn:universal_set}
Let ${\cal C}'$ be a set of functions $c:\{1,2,\ldots,n'\}\rightarrow \{0,1\}$. We say that ${\cal C}'$ is an $(n',k',p')$-universal set if for every subset
$I\subseteq \{1,2,\ldots,n'\}$ of size $k'$ and a function $c':I\rightarrow\{0,1\}$ that assigns $1$ to exactly $p'$ indices, there is a function $c\in{\cal C}'$ such that for all $i\in I$, $c(i)=c'(i)$.
\end{definition}

Naor et al.~\cite{NaorSS95} (see also~\cite{DBLP:journals/siamcomp/ChitnisCHPP16}) proved that small universal sets can be computed efficiently.

\begin{proposition}[\cite{NaorSS95}]\label{theorem:splitter}
Given integers $k'\leq n'$, there exists an $(n',k',p')$-universal set, ${\cal C}'$, of size $\OO(2^{p'\log k'}\cdot\log n')$, and this family is computable in time $\OO(2^{p'\log k'}\cdot n'\log n')$.
\end{proposition}

We now define the objective we aim to achieve by coloring the vertices in $V(G)$. For this purpose, given an $(h,\{(L^e,R^e)\}|_{e\in T})$-solution 
$S=(F,\{(X_e,\bar{X}_e)\}|_{e\in T})$, for all $e\in T$, we denote $S_e=(Y_e\setminus X_e)\cup(\bar{Y}_e\setminus \bar{X}_e)$ and $N_e=N((Y_e\setminus X_e)\cup(\bar{Y}_e\setminus \bar{X}_e))$.

\begin{definition}\label{def:niceCol}
Let $c:\{1,2,\ldots,n\}\rightarrow \{0,1\}$ be a function. Then, we say that an aligned $(h,\{(L^e,R^e)\}|_{e\in T})$-solution $S=(F,\{(X_e,\bar{X}_e)\}|_{e\in T})$ is {\em colored with respect to $c$} if the two following conditions are satisfied.
\begin{itemize}
\item For all $v\in \bigcup_{e\in T} S_e$, it holds that $c(v)=0$.
\item For all $v\in (\bigcup_{e\in T}N_e)\setminus (\bigcup_{e\in T} S_e)$, it holds that $c(v)=1$.
\end{itemize}
\end{definition}

Next, the function $c$ will be clear from context, and therefore we simply use the term colored. By using Proposition \ref{theorem:splitter}, we compute an $(n,(q+2(k+1))|T|,2(k+1)|T|)$-universal set, ${\cal C}$, of size $2^{\gamma\cdot k|T|\log q}\cdot\log n$ in time $2^{\gamma\cdot k|T|\log q}\cdot n\log n$ for the appropriate constant $\gamma$.
Note that if there exists an aligned $(h,\{(L^e,R^e)\}|_{e\in T})$-solution $S=(F,\{(X_e,\bar{X}_e)\}|_{e\in T})$, then for all $e\in T$, it holds that $|S_e|\leq q$ and $|N_e|\leq 2(k+1)|T|$. This implies that $|\bigcup_{e\in T} S_e|\leq q|T|$ and $|(\bigcup_{e\in T}N_e)\setminus (\bigcup_{e\in T} S_e)|\leq 2(k+1)|T|$. Thus, by Definition \ref{dfn:universal_set}, there exists $c\in {\cal C}$ such that for all $v\in \bigcup_{e\in T} S_e$, it holds that $c(v)=0$, and for all $v\in (\bigcup_{e\in T}N_e)\setminus (\bigcup_{e\in T} S_e)$, it holds that $c(v)=1$. Therefore, there exists $c\in{\cal C}$ such that $(F,\{(X_e,\bar{X}_e)\}|_{e\in T})$ is colored. We can exhaustively examine every function in $\cal C$. Indeed, we thus only examine $2^{\gamma\cdot k|T|\log q}\cdot\log n$ functions. Therefore, we next fix some $c\in{\cal C}$. Now, we need the following definition.

\begin{definition}
We say that an $(h,\{(L^e,R^e)\}|_{e\in T})$-solution $(F,\{(X_e,\bar{X}_e)\}|_{e\in T})$ is {\em color optimal} is there is no colored aligned $(h,(L,R))$-solution $(F',\{(X'_e,\bar{X}'_e)\}|_{e\in T})$ such that $|F'|<|F|$.
\end{definition}

Therefore, to solve the current instance in time $\tau'(n,m,k,t,|T|)$, by choosing a large enough $\alpha$ in advance, it is sufficient that we next prove the following lemma, where $\widehat{\alpha}=\alpha'/2$.

\begin{lemma}\label{lem:colSuffice}
The following computation can be performed in time $\widehat{\tau}(n,m,k,t,|T|)=2^{2^{\widehat{\alpha}(t+k^2)|T|}}\cdot (n+m)^{\widehat{\alpha}}$.
If there exists no $(h,\{(L^e,R^e)\}|_{e\in T})$-solution, then answer no; if there exists a colored aligned $(h,\{(L^e,R^e)\}|_{e\in T})$-solution, then return a color optimal $(h,\{(L^e,R^e)\}|_{e\in T})$-solution; else, either return no or an $(h,\{(L^e,R^e)\}|_{e\in T})$-solution.
\end{lemma}

\subsubsection{Independence}\label{sec:independence}

Denote $P=\{v\in V(G): c(v)=1\}$. Moreover, denote $\widehat{G}=G\setminus P$. Now, let ${\cal D}'$ denote the set of connected components of $\widehat{G}$. Let us fix for now some hypothetical $(h,\{(L^e,R^e)\}|_{e\in T})$-solution $S=(F,\{(X_e,\bar{X}_e)\}|_{e\in T})$ that is aligned, colored and color optimal (if one exists). First, Definition \ref{def:niceCol} directly implies the correctness of the following observation.

\begin{observation}\label{obs:independence}
For all $e\in T$ and $v\in P$, it holds that $v\notin S_e$.
\end{observation}

Informally, this observation implies that we can handle the components in ${\cal D}'$ in a roughly independent manner: there are no edges between the components in ${\cal D}'$, and the ``location'' (that is, the association to either $X_e$ or $\bar{X}_e$) of the vertices that have neighbors in these components but do not belong to them, is known.

The next lemma indicates that among the components in ${\cal D}'$, we only need the handle those that are small.

\begin{lemma}\label{lem:onlyHandleSmall}
Let $D\in{\cal D}'$ such that $|V(D)|>q|T|$. Then, for all $e\in T$ and $v\in V(D)$, it holds that $v\notin S_e$.
\end{lemma}

\begin{proof}
Suppose, by way of contradiction, that there exists $e\in T$ and $v\in V(D)$ such that $v\in S_e$. Let $D^*$ denote the connected component of $G[\bigcup_{e'\in T}S_{e'}]$ that contains $v$. Then, $V(D^*)\subseteq V(D)$ (since $S$ is colored). However, in $G$, $N(V(D^*))\subseteq P$ (again, since $S$ is colored). This implies that $D=D^*$. However, $|V(D^*)|\leq |\bigcup_{e'\in T}S_{e'}|\leq q|T|$, and thus we have reached a contradiction.
\end{proof}

Let ${\cal D}=\{D_1,D_2,\ldots,D_r\}$, for the appropriate $r$, denote the set of each $D\in{\cal D}'$ such that $|V(D)|\leq q|T|$. Now, for each $1\leq i\leq r$, \alg{Recurs} calls itself recursively on the instance $(D_i,k,t,(V_1\cap V(D_i),V_2\cap V(D_i),\ldots,V_t\cap V(D_i)),T,\{f_e\}|_{e\in T},\emptyset,\{(W^e_1\cup L^e\cup(P\cap Y_e))\cap V(D_i), (W^e_2\cup R^e\cup (P\cap \bar{Y}_e))\cap V(D_i))\}|_{e\in T})$. These instances are of the form solved in Section \ref{sec:small}. Thus by the Lemma \ref{lem:small}, for each triple $(i,\widehat{h},(\emptyset,\emptyset))$ for $1\leq i\leq r$ and a parity restriction $\widehat{h}$, we obtain a correct answer for the constructed instance, which we denote by $S(i,\widehat{h})$, that is either no or an optimal $(\widehat{h},\{(L^e\cap V(D_i)\cup (P\cap Y_e),R^e\cap V(D_i)\cup (P\cap \bar{Y}_e))\}|_{e\in T})$-solution. For convenience, let us denote $S(i,\widehat{h})=\nil$ if the answer is no, and otherwise $S(i,\widehat{h})=(F^{i,\widehat{h}},\{(X^{i,\widehat{h}}_e,\bar{X}^{i,\widehat{h}}_e)\}|_{e\in T})$. Moreover, by Lemma \ref{lem:small}, the time that we spent for the recursive calls is bounded by $r\cdot\tau_\mathrm{s}(n,m,k,t,|T|)$.

\subsubsection{Dynamic Programming}

By relying on Observation \ref{obs:independence} and Lemma \ref{lem:onlyHandleSmall}, we can prove Lemma \ref{lem:colSuffice} by using dynamic programming. Roughly speaking, here we need to be slightly careful as the neighborhoods of our ``independent components'' may intersect. For the sake of completeness, we briefly describe the table and the formulas that compute it. We use a table M that has an entry M$[i,h]$ for all $0\leq i\leq r$ and $\widehat{h}:T\rightarrow \widehat{B}$. To explain the purpose of this entry, we need the following definition.

\begin{definition}
Let $0\leq i\leq r$, $\widehat{h}:T\rightarrow \widehat{B}$. Then, the set ${\cal S}(i,\widehat{h})$ is defined as $(F,\{(X_e,\bar{X}_e)\}|_{e\in T})$ such that the following conditions are satisfied.
\begin{itemize}
\item For all $e\in T$ and $1\leq i\leq t$, it holds that $|X_e\cap V_i|=b^e_i\mod 2$ where $\widehat{h}(e)=(b^e_1,b^e_2,\ldots,b^e_i)$.
\item For all $1\leq j\leq o$, there exist $h^j$ and $S(j,h^j)=(F^j,\{(X^j_e,\bar{X}^j_e)\}|_{e\in T})$ such that $F=\bigcup_{1\leq j\leq i}F^j$, $X^j_e=\bigcup_{1\leq i\leq r}X^j_e$ and $\bar{X}^j_e=\bigcup_{1\leq j\leq i}\bar{X}^j_e$.
\end{itemize}
\end{definition}

The purpose of the entry M$[i,\widehat{h}]$ is to store a pair $(F,\{(X_e,\bar{X}_e)\}|_{e\in T})\in{\cal S}(i,\widehat{h})$ that minimizes $|F|$. In the basis, we initialize each entry where $i=0$: if $\image(\widehat{h})=\{(0,0,\ldots,0)\}$ then M$[i,\widehat{h}]=(\emptyset,\{(\emptyset,\emptyset)\}|_{e\in T})$ and otherwise M$[i,\widehat{h}]=\nil$. Now, suppose that $i\geq 1$. Then, we compute an entry M$[i,\widehat{h}]$ as follows.
\begin{itemize}
\item Denote $Z=V(D_i)\cap(\bigcup_{j=1}^iV(D_j))$. Now, for all $e\in T$, define $B^e=(b^e_1,b^e_2,\ldots,b^e_t)$ as the binary vector such that for all $1\leq j\leq t$, we have $b^e_j=|Z\cap V_j|\mod 2$.
\item For all $h'$ and $h''$ such that for all $e\in T$, it holds that $h'(e)+h''(e)+B^e=\widehat{h}(e)\mod 2$: Denote by $(F',\{(X'_e,\bar{X}'_e)\}|_{e\in T})$ (or $\nil$) the value in M$[i-1,h']$. If this value is $\nil$ or $S(i,h'')$ is $\nil$, then denote $M(h',h'')=\nil$. Else, let $M(h',h'')$ denote $(F^{h',h''},\{(X^{h',h''}_e,\bar{X}^{h',h''}_e)\}|_{e\in T})=(F'\cup F^{i,h''},\{(X'_e\cup X_e^{i,h''}, \bar{X}_e\cup \bar{X}_e^{i,h''})\}|_{e\in T})$ (here we use the values computed in Section \ref{sec:independence}).
\item Let M$[i,\widehat{h}]$ store a pair $M(h',h'')=(F^{h',h''},\{(X^{h',h''}_e,\bar{X}^{h',h''}_e)\}|_{e\in T})$ that minimizes $|F^{h',h''}|$ (if there are several options, the choice is arbitrary). If there is no such pair (since all values are $\nil$), let it store $\nil$.
\end{itemize}

Having computed M, we need the following computation.
\begin{itemize}
\item For all $e\in T$, denote $Z_e=(P\cap Y_e)\setminus (\bigcup_{j=1}^rV(D_i))$ and $\bar{Z}_e=(P\cap \bar{Y}_e)\setminus (\bigcup_{j=1}^rV(D_i))$
\item $F_Z=\bigcup_{e\in T}\cont(e,Z_e)$ where the operator $\cont$ refers the graph $G[Z_e\cup\bar{Z}_e]$.
\item For all $e\in T$, define $B^e=(b^e_1,b^e_2,\ldots,b^e_t)$ as the binary vector such that for all $1\leq j\leq t$, we have $b^e_j=|Z^e\cap V_j|\mod 2$.
\end{itemize}

Now, define $h':T\rightarrow\widehat{B}$ as follows. For all $e\in T$, define $h'(e)=B^e+h(e)\mod 2$. If M$[r,h']$ is $\nil$, then return $\nil$. Else, denote by $(F',\{(X'_e,\bar{X}'_e)\}|_{e\in T})$ the value in M$[r,h']$. Then, we define $F=F'\cup F_Z$ and $(\{(X_e=X'_e\cup Z_e,\bar{X}_e\cup\bar{Z}_e)\}|_{e\in T})$. If we have thus obtained an $(\widehat{h},\{(L^e,R^e)\}|_{e\in T})$-solution (which can be checked in polynomial time), then we return it, and else we return $\nil$.

Overall, by choosing a large enough $\alpha$ in advance, we obtain that the computation in Section \ref{sec:independence} together with the dynamic programming computation  in this section can be performed in the desired time $\widehat{\tau}(n,m,k,t,|T|)$.

\subsection{Breakable Graphs}

We now handle the case where $G$ is $(q,p)$-breakable. Algorithm \alg{Recurs} first uses the algorithm given by Proposition \ref{prop:goodSep} to obtain a partition $(Q,P)$ of $V(G)$ that is a $(q,p)$-good edge separation. Without loss of generality, we assume that $|Q\cap W|\leq |P\cap W|$.

\subsubsection{Collecting Solutions for the Side $Q$}\label{sec:collecting}

Let $U$ denote the set of each vertex in $Q$ that is an endpoint of an edge in $E(Q,P)$. Since $(Q,P)$ is a $(q,p)$-good edge separation, it holds that $|U|\leq p$. Roughly speaking, we first attempt to ``fully understand'' $G[Q]$. This means that for any optimal $(h,\{(L^e,R^e)\}|_{e\in T})$-solution $(F,\{(X_e,\bar{X}_e)\}|_{e\in T})$, given its restriction to $G[P]$, we will be able to extend it to a ``complete'' optimal $(h,\{(L^e,R^e)\}|_{e\in T})$-solution (which might be different than $(F,\{(X_e,\bar{X}_e)\}|_{e\in T})$).

Let us denote $W_Q=U\cup (W\cap Q)$. To ``fully understand'' $G[Q]$, \alg{Recurs} calls itself recursively on the instance $(G,k,t,(V_1\cap Q,V_2\cap Q,\ldots,V_t\cap Q),T,\{f_e\}|_{e\in T},W_Q,\{(W^e_1\cap Q,W^e_2\cap Q)\}|_{e\in T})$. Since $|Q\cap W|\leq |P\cap W|$, $|W|\leq 2p$ and $|U|\leq p$, we have that $|W_Q|\leq 2p$, and therefore the call is valid. Moreover, note that $|Q|<|V(G)|$. Thus by the inductive hypothesis, for each pair $(h,\{(L^e,R^e)\}|_{e\in T})$ of a parity restriction $h$ and a set of partitions $\{(L^e,R^e)\}|_{e\in T}$ of $W_Q$, we obtain a correct answer for the constructed instance, which we denote by $S(h,\{(L^e,R^e)\}|_{e\in T})$, that is either $\nil$ or an optimal $(h,\{(L^e,R^e)\}|_{e\in T})$-solution. For convenience, let us denote $S(h,\{(L^e,R^e)\}|_{e\in T})=\nil$ if the answer is $\nil$, and otherwise $S(h,\{(L^e,R^e)\}|_{e\in T})=(F^{h,\{(L^e,R^e)\}|_{e\in T}},\{(X_e^{h,\{(L^e,R^e)\}|_{e\in T}},\bar{X}_e^{h,\{(L^e,R^e)\}|_{e\in T}})\}|_{e\in T})$. Moreover, by the inductive hypothesis, the time that we spent so far is bounded by $\tau(|Q|,|E(G[Q])|,k,t,|T|)+(n+m+k+t+|T|)^{\delta'}$ for the appropriate constant $\delta'$. If all of the recursive calls returned $\nil$, then we can safely return $\nil$. Thus, in the following Sections \ref{sec:replace} and \ref{sec:solveEnd}, we assume that at least one call did not return $\nil$.

\subsubsection{Replacing the Side $Q$}\label{sec:replace}

Next, we would like to replace $G[Q]$ by a ``small'' graph. For this purpose, we need to identify vertices that provide redundant information. To this end, we let $V_1$ denote the set of vertices $v$ such that there exists a pair $(h,\{(L^e,R^e)\}|_{e\in T})$ of a parity restriction $h$ and a set of partitions $\{(L^e,R^e)\}|_{e\in T}$ of $W_Q$ such that  $v$ is an endpoint of an edge in $F^{h,\{(L^e,R^e)\}|_{e\in T}}$. Note that $|V_1|\leq 2k\cdot 2^{t|T|+2p|T|}$. Moreover, let $V_2$ denote the set of each vertex in $Q$ that is an endpoint of a terminal in $T$. Finally, denote $V^*=V_1\cup V_2\cup W_Q$. The following observation is immediate.

\begin{observation}\label{obs:sizeV*}
$|V^*|\leq 2k\cdot 2^{t|T|+2p} + 2(|T|+p)$.
\end{observation}

We also need the following definitions. Roughly speaking, these definitions capture equivalence classes of vertices such that vertices that belong to the same equivalence class supply, in some sense, redundant information.

\begin{definition}\label{def:redPair}
Let $v,u\in Q\setminus V^*$. We say that $(v,u)$ is a {\em redundant pair} if all of the following conditions are satisfied.
\begin{enumerate}
\item There exists $1\leq i\leq t$ such that $v,u\in V_i$.
\item For each terminal $e\in T$, one of the following conditions is satisfied.
	\begin{enumerate}
	\item $v,u\notin W^e_1\cup W^e_2$.
	\item $v,u\in W^e_1$.
	\item $v,u\in W^e_2$.
	\end{enumerate}
\item For each pair $(h,\{(L^e,R^e)\}|_{e\in T})$ of a parity restriction $h$ and a set of partitions $\{(L^e,R^e)\}|_{e\in T}$ of $W_Q$, and for each terminal $e\in T$, one the following conditions is satisfied.
	\begin{enumerate}
	\item $S(h,\{(L^e,R^e)\}|_{e\in T})=\nil$.
	\item $v,u\in X^{h,\{(L^e,R^e)\}|_{e\in T}}_e$.
	\item $v,u\in \bar{X}^{h,\{(L^e,R^e)\}|_{e\in T}}_e$.
	\end{enumerate}
\end{enumerate}
\end{definition}

More generally, we have the following definition.

\begin{definition}\label{def:redSet}
We say that a subset $U\subseteq Q\setminus V^*$ is a {\em redundant set} if for all $v,u\in U$, it holds that $(v,u)$ is a redundant pair.
\end{definition}

Let ${\cal U}$ denote the family of all non-empty redundant sets.% , and denote $\widehat{U}=Q\setminus(\bigcup_{U\in{\cal U}}U\cup W_Q)$.
 Note that each set $U\in{\cal U}$ is characterized by a type $1\leq i\leq t$, an integer $j^e\in\{0,1,2\}$ for $e\in T$ indicating whether the vertices in $U$ belong to $Q\setminus(W^e_1\cup W^e_2)$, $W^e_1$ or $W^e_2$, and a Boolean variable $b^{h,\{(L^e,R^e)\}|_{e\in T},e'}$ for each triple $(h,\{(L^e,R^e)\}|_{e\in T},e)'$ of a parity restriction $h$, set of partitions $\{(L^e,R^e)\}|_{e\in T}$ of $W_Q$ and terminal $e'\in T$, indicating whether the vertices in $U$ belong to $X^{h,\{(L^e,R^e)\}|_{e\in T}}_{e'}$ or $\bar{X}^{h,\{(L^e,R^e)\}|_{e\in T}}_{e'}$. Thus, the next observation follows directly from Definitions \ref{def:redPair} and \ref{def:redSet}.

\begin{observation}\label{obs:sizeFam}
The set ${\cal U}$ is a partition of $Q\setminus V^*$, and $|{\cal U}|\leq t\cdot 3^{|T|}\cdot 2^{2^{t|T|+2p|T|}\cdot |T|}$.
\end{observation}

We aim to decrease the size of each set in $\cal U$. First, from each set in $\cal U$ of even size, let us remove one vertex. Let $U_R$ denote the set of vertices that were removed. Now, the size of each set in $\cal U$ is odd. For each set $U\in{\cal U}$, choose (arbitrarily) a vertex $v_U$. Denote $Z=(\bigcup_{U\in{\cal U}}U\setminus\{v_U\})$
 We need the following definition, which, informally, presents our ``compact'' instance which is the result of removal o redundant information that we managed to identify (together with the insertion of a few other ingredients to ensure equivalence).

\begin{definition}
We define the instance $I^*=(G^*,k,t,(V_1\setminus Z,V_2\setminus Z,\ldots,V_t\setminus Z),T,\{f^*_e\}|_{e\in T},W_Q,$ $\{{W^*}^e_1=W^e_1\setminus Z,{W^*}^e_2=W^e_2\setminus Z\}|_{e\in T})$ of \aedc{} as follows.
\begin{itemize}
\item $V(G^*)=V(G)\setminus Z$.
\item For every $e\in E(G)$:
	\begin{itemize}
	\item If both endpoints of $e$, $v$ and $u$, belong to $Z$: Let $U\in{\cal U}$ such that $v\in U$, and let $U'\in{\cal U}$ such that $u\in U'$. If $U\neq U'$, then insert $k+1$ new edges $e'$ between $v_U$ and $v_{U'}$ into $E(G^*)$. Moreover, in this case, for each new edge $e'$ and $\widehat{e}\in T$, define $f^*_{\widehat{e}}(e')=f_{\widehat{e}}(e)$.
	\item If exactly one endpoint of $e$ is a vertex, $v$, in $Z$: Let $u$ denote the other endpoint of $e$. Let $U\in{\cal U}$ such that $v\in U$. If $u\notin U$,
	then insert $k+1$ new edges $e'$ between $v_U$ and $u$ into $E(G^*)$. Moreover, in this case, for each new edge $e'$ and $\widehat{e}\in T$, define $f^*_{\widehat{e}}(e')=f_{\widehat{e}}(e)$.
	\item Otherwise: Insert $e$ into $E(G^*)$. Moreover, for all $\widehat{e}\in T$, define $f^*_{\widehat{e}}(e)=f_{\widehat{e}}(e)$.
	\end{itemize}
\end{itemize} 
\end{definition}

By Observations \ref{obs:sizeV*} and \ref{obs:sizeFam}, the following observation is immediate for the appropriate constant $\lambda'$.

\begin{observation}\label{obs:numVerticesRep}
$|V(G^*)\setminus P|\leq 2k\cdot 2^{t|T|+2p} + 2(|T|+p) + 2t\cdot 3^{|T|}\cdot 2^{2^{t|T|+2p|T|}\cdot |T|}\leq 2^{2^{\lambda'(t+k^2)|T|}}$.
\end{observation}

We need to prove equivalence between our new instance and the original one. First, we argue that the solutions for the original instance can be modified to obtain solutions for the new instance $I^*$ where the size $|F|$ does not ``grow'' and the ``distributions'' of $W$ among partitions remain the same. For formal argument, we need the following notation. Consider some parity restriction $h': V(G)\rightarrow\widehat{B}$ and a set of partitions $\{({L^e}',{R^e}')\}|_{e\in T}$ of $W$. Now, suppose that there exists an $(h',\{({L^e}',{R^e}')\}|_{e\in T})$-solution for the original instance, and let $(F',\{X'_e,\bar{X}'_e\}|_{e\in T})$ be such a solution. We define $h: T\rightarrow\widehat{B}$ as follows. For all $e\in T$, define $h(e)=B^e=(b^e_1,b^e_2,\ldots,b^e_t)$ such that for all $1\leq i\leq t$, we set $b^e_i=|X'_e\cap Q|\mod 2$.
For all $e\in T$, denote $L^e=W_Q\cap ({L^e}'\cup (U\cap X'_e))$ and $R^e=W_Q\cap ({R^e}'\cup (U\cap \bar{X}'_e))$. Moreover, define $F=F'\cap E(G[Q])$, and for all $e\in T$, define $X_e=X'_e\cap Q$ and $\bar{X}_e=\bar{X}'_e\cap Q$.
Denote $F''=(F'\setminus F)\cup F^{h,\{(L^e,R^e)\}|_{e\in T}}$, and for all $e\in T$, denote $X_e''=(X'_e\setminus X_e)\cup X_e^{h,\{(L^e,R^e)\}|_{e\in T}}$ and $\bar{X}_e''=(\bar{X}'_e\setminus \bar{X}_e)\cup \bar{X}_e^{h,\{(L^e,R^e)\}|_{e\in T}}$. Then, since $L^e=W_Q\cap ({L^e}'\cup (U\cap X'_e))$ and $R^e=W_Q\cap ({R^e}'\cup (U\cap \bar{X}'_e))$, we have that $(F'',\{(X_e'',\bar{X}_e'')\}|_{e\in T})$ is an $(h',\{({L^e}',{R^e}')\}|_{e\in T})$-solution to the original instance. (We remark that it also holds that $(F,\{(X_e,\bar{X}_e)\}|_{e\in T})$ is an $(h,\{(L^e,R^e)\}|_{e\in T})$-solution for the instance $I_Q$ constructed in Section \ref{sec:collecting}.) Therefore, we can assume w.l.o.g.~that $F=F^{h,\{(L^e,R^e)\}|_{e\in T}}$ and $\{(X_e,\bar{X}_e)\}|_{e\in T}=\{(X_e^{h,\{(L^e,R^e)\}|_{e\in T}},\bar{X}_e^{h,\{(L^e,R^e)\}|_{e\in T}})\}|_{e\in T}$.

Now, we construct an $(h',\{({L^e}',{R^e}')\}|_{e\in T})$-solution $(F^*,\{(X^*_e,\bar{X}^*_e)\}|_{e\in T})$ for the new instance $I^*$ as follows. (Note that the use of the term  $(h',\{({L^e}',{R^e}')\}|_{e\in T})$-solution is well-defined since $T$ and the set $W$ are the same in $I^*$ and the original instance.) First, we define $F^*=F'$. This is possible since the endpoints of the edges in $F'$ do not belong to $Z$ (since $F=F^{h,\{(L^e,R^e)\}|_{e\in T}}$), and therefore $F'\subseteq E(G^*)\setminus T$. Next, for every $e\in T$, we define $X^*_e=X'_e\setminus Z$  and $\bar{X}^*_e=\bar{X}'_e\setminus Z$.

\begin{lemma}\label{lem:equivRepOne}
It holds that $(F^*,\{(X^*_e,\bar{X}^*_e)\}|_{e\in T})$ is an $(h',\{({L^e}',{R^e}')\}|_{e\in T})$-solution for $I^*$.
\end{lemma}

\begin{proof}
All of the following arguments implicitly rely on the fact that $(F',\{X'_e,\bar{X}'_e\}|_{e\in T})$ is an $(h',\{({L^e}',{R^e}')\}|_{e\in T})$-solution. First, note that $|F^*|=|F'|\leq k$. Now, consider some terminal $e\in T$. Since $X^*_e\subseteq X'_e$ and $\bar{X}^*_e\subseteq \bar{X}'_e$, it holds that $\cont(e,X^*_e)\cap\{e\}=\cont(e,X'_e)\cap\{e\}=\{e\}\cap E(G^*)$, $\cont(e,X^*_e)\setminus\{e\}\subseteq F'=F^*$, $W^e_1\setminus Z\subseteq X^*_e$ and $W^e_2\setminus Z\subseteq \bar{X}^*_e$. Moreover, since $L^e\cup R^e\subseteq V(G^*)$, we deduce that $L^e\subseteq X^*_e$ and $R^e\subseteq \bar{X}^*_e$. Denote $h'(e)=(b^e_1,b^e_2,\ldots,b^e_t)$. Denote ${\cal U}=\{U\in{\cal U}: v_U\in X^*_e\}$. Notice that $X^*_e\cup(\bigcup_{U\in {\cal U}}U)=X'_e$ and that for all $U\in{\cal U}$, it holds that $X^*_e\cap U=\{v_U\}$. Thus, since the size of each set in ${\cal U}$ is odd (where vertices in the same set have the same type), we have that for all $1\leq i\leq r$, it holds tat $|X^*_e\cap V_i|=|X'_e\cap V_i|=b^e_i\mod 2$.

To conclude that the lemma is correct, it remains to show that none of the newly added edges belongs to $\cont(e,X^*_e)$, since then we get that $\cont(e,X^*_e)\setminus\{e\}\subseteq F^*$ (rather than only $(\cont(e,X^*_e)\setminus\{e\})\cap E(G)\subseteq F^*$). Consider some newly added edge $e'$. Let $\widehat{e}\in E(G)$ be the edge examined when we added $e'$ to $G^*$. We consider two cases.
\begin{itemize}
\item We have that $v\in Z$ and $u\notin Z$ are the two endpoints of $e'$. Since $\widehat{e}$ has an endpoint in $Z$, it holds that $\widehat{e}\notin\cont(e,X'_e)$. Then, the two endpoints of $e'$ are $v_U$ and $u$ for an appropriate set $U\in{\cal U}$. Notice that since $(X_e,\bar{X}_e)=(X_e^{h,\{(L^e,R^e)\}|_{e\in T}},\bar{X}_e^{h,\{(L^e,R^e)\}|_{e\in T}})$, we have that either $v,v_U\in X^*_e$ or $v,v_U\in \bar{X}^*_e$, and we also have that $u\in(X^*_e\cap X'_e)\cup(\bar{X}^*_e\cap \bar{X}'_e)$. Thus, we deduce that $e'\notin\cont(e,X'_e)$.

\item We have that $v\in Z$ and $u\in Z$ are the two endpoints of $e'$. Since $\widehat{e}$ has an endpoint in $Z$, it holds that $\widehat{e}\notin\cont(e,X'_e)$. Then, the two endpoints of $e'$ are $v_U$ and $v_{U'}$ for the appropriate sets $U,U'\in{\cal U}$ and $U\neq U'$. Notice that since $(X_e,\bar{X}_e)=(X_e^{h,\{(L^e,R^e)\}|_{e\in T}},\bar{X}_e^{h,\{(L^e,R^e)\}|_{e\in T}})$, we have that either $v,v_U\in X^*_e$ or $v,v_U\in \bar{X}^*_e$, and we also have that either $u,v_{U'}\in X^*_e$ or $u,v_{U'}\in \bar{X}^*_e$. Thus, we deduce that $e'\notin\cont(e,X'_e)$.
\end{itemize}
\end{proof}

Second, we show how to modify solutions for $I^*$ to obtain solutions for the original instance where the size $|F|$ does not ``grow'' and the ``distributions'' of $W$ among partitions remain the same. 
For formal argument, we need the following notation. Consider some parity restriction $h': V(G^*)\rightarrow\widehat{B}$ and a set of partitions $\{({L^e}',{R^e}')\}|_{e\in T}$ of $W$. Now, suppose that there exists an $(h',\{({L^e}',{R^e}')\}|_{e\in T})$-solution for $I^*$, and let $(F',\{X'_e,\bar{X}'_e\}|_{e\in T})$ be an optimal such solution. We define $h: T\rightarrow\widehat{B}$ as follows. For all $e\in T$, define $h(e)=B^e=(b^e_1,b^e_2,\ldots,b^e_t)$ such that for all $1\leq i\leq t$, we set $b^e_i=|X'_e\cap Q|\mod 2$. For all $e\in T$, denote $L^e=W_Q\cap ({L^e}'\cup(U\cap X'_e))$ and $R^e=W_Q\cap ({R^e}'\cup (U\cap \bar{X}'_e))$. Moreover, define $F=F'\cap E(G[Q])$. For all $e\in T$, define ${\cal U}_e=\{U\in{\cal U}: v_U\in X'_e\}$, $U_e=\bigcup_{U\in{\cal U}_e}U$,  $\bar{\cal U}_e=\{U\in{\cal U}: v_U\in \bar{X}'_e\}$ and $\bar{U}_e=\bigcup_{U\in\bar{\cal U}_e}U$. Accordingly, for all $e\in T$, define $X_e=(X'_e\cap Q)\cup U_e$ and $\bar{X}_e=(\bar{X}'_e\cap Q)\cup\bar{U}_e$.

\begin{lemma}\label{lem:inEquivRepTwo}
It holds that $(F,\{(X_e,\bar{X}_e)\}|_{e\in T})$ is an $(h,\{(L^e,R^e)\}|_{e\in T})$-solution for $I_Q$.
\end{lemma}

\begin{proof}
All of the following arguments implicitly rely on the fact that $(F',\{X'_e,\bar{X}'_e\}|_{e\in T})$ is an $(h',\{({L^e}',{R^e}')\}|_{e\in T})$-solution. First, note that $|F|\leq|F'|\leq k$. Now, consider some terminal $e\in T$. Since $V_2\subseteq V^*\subseteq V(G^*)$, it holds that $\cont(e,X_e)\cap\{e\}=\cont(e,X'_e)\cap\{e\}=\{e\}\cap E(G[Q])$. Moreover, by definition $L^e=W_Q\cap ({L^e}'\cup(U\cap X'_e))$ and $R^e=W_Q\cap ({R^e}'\cup (U\cap \bar{X}'_e))$, and therefore $L^e\subseteq X_e$ and $R^e\subseteq \bar{X}_e$. Denote $h'(e)=(b^e_1,b^e_2,\ldots,b^e_t)$, and consider some $1\leq i\leq t$. Then, since for all $U\in{\cal U}_e$, we have $|X_e\cap U|=1$, and since the size of each set in ${\cal U}$ is odd (where vertices in the same set have the same type), we have that $|X_e\cap V_i|=|(X_e\cup U_e)\cap V_i|=|X'_e\cap V_i|=b^e_i\mod 2$.

Next, we claim that $W^e_1\cap Q\subseteq X_e$ and $W^e_2\cap Q\subseteq \bar{X}_e$. We only show that $W^e_1\cap Q\subseteq X_e$, since the proof of the second claim is symmetric. We first note that $W^e_1\cap Q\cap V(G^*) \subseteq X'_e\cap Q\subseteq X_e$. Now, consider a vertex $v\in W^e_1\cap Q\setminus V(G^*)$. Then, there exists a set $U\in{\cal U}$ such that $v\in U$. By Definitions \ref{def:redPair} and \ref{def:redSet}, it holds that $v_U\in W^e_1$. However, $v_U\in V(G^*)$, and therefore $v_U\in X'_e$. This means that $v\in U_e$, and thus we get that $v\in X_e$.

It remains to show that $\cont(e,X_e)\setminus\{e\}\subseteq F$. Since $F=F'\cap E(G[Q])$, $X'_e\cap Q\subseteq X_e$ and $\bar{X}'_e\cap Q\subseteq \bar{X}_e$, it is sufficient to show that every edge $\widehat{e}\in E(G[Q])\setminus E(G^*)$ satisfies $\widehat{e}\notin\cont(e,X_e)$. For this purpose, consider some $\widehat{e}\in E(G[Q])\setminus E(G^*)$. Then, it has an endpoint, $v$, that belongs to $Z$. Let $u$ denote the other endpoint (which might also belong to $Z$). Let $U$ denote the set in $\cal U$ such that $v\in U$. We consider three cases.
\begin{itemize}
\item First, suppose that $u\notin Z\cup\{v_U\}$.
Then, by construction, in $E(G^*)$ there are $k+1$ edges $e''$ between $v_U$ and $u$ such that $f_e(\widehat{e})=f_e(e'')$. Since $|F'|\leq k$, for at least one edge among them, $e'$, it holds that $e'\notin\cont(e,X_e)$. Then, since either $v,v_U\in X_e$ or $v,v_U\in\bar{X}_e$, and because $u\in (X_e\cap X'_e)\cup(\bar{X}_e\cap \bar{X}'_e)$, we deduce that $\widehat{e}\notin\cont(e,X_e)$.

\item Now, suppose that $u\in U$. Recall that there exist $\widehat{h}:T\rightarrow\widehat{B}$ and a set of partitions $\{(\widehat{L}^e,\widehat{R}^e)\}|_{e\in T}$ of $W_Q$ for which $S(\widehat{h},\{(\widehat{L}^e,\widehat{R}^e)\}|_{e\in T})=(F^{\widehat{h},\{(\widehat{L}^e,\widehat{R}^e)\}|_{e\in T}},\{(X_e^{\widehat{h},\{(\widehat{L}^e,\widehat{R}^e)\}|_{e\in T}},$ $\bar{X}_e^{\widehat{h},\{(\widehat{L}^e,\widehat{R}^e)\}|_{e\in T}})\}|_{e\in T})$, that is, the answer is not $\nil$. In particular, by Definition \ref{def:redPair}, we have that $\widehat{e}\notin\cont(e,X^{\widehat{h},\{(\widehat{L}^e,\widehat{R}^e)\}|_{e\in T}}_e)$ and $v,u\in X^{\widehat{h},\{(\widehat{L}^e,\widehat{R}^e)\}|_{e\in T}}_e$. This implies that $f_e(\widehat{e})=0$. Notice that either $v,u\in X_e$ or $v,u\in\bar{X}_e$, and therefore $\widehat{e}\notin\cont(e,X_e)$.

\item Finally, suppose that $u\in Z\setminus U$. Let $U'$ be the set in ${\cal U}$ such that $u\in U$. Then, by construction, in $E(G^*)$ there are $k+1$ edges $e''$ between $v_U$ and $v_{U'}$ such that $f_e(\widehat{e})=f_e(e'')$. Since $|F'|\leq k$, for at least one edge among them, $e'$, it holds that $e'\notin\cont(e,X_e)$. Then, since either $v,v_U\in X_e$ or $v,v_U\in\bar{X}_e$, and also either $u,v_{U'}\in X_e$ or $u,v_{U'}\in\bar{X}_e$, we deduce that $\widehat{e}\notin\cont(e,X_e)$.
\end{itemize}
\end{proof}

Denote $F''=F\cup (F'\cap E(G[P\cup U]))$, and note that $|F''|\leq|F'|$ (since $|F|\leq|F'\cap E(G[Q])|$). For all $e\in T$, denote $X_e''=X_e\cup (X'_e\cap P)$ and $\bar{X}_e''=\bar{X}_e\cup (\bar{X}'_e\cap P)$. Then, by Lemma \ref{lem:inEquivRepTwo}, since $L^e=W_Q\cap ({L^e}'\cup(U\cap X'_e))$ and $R^e=W_Q\cap ({R^e}'\cup (U\cap \bar{X}'_e))$, we have the following result.

\begin{lemma}\label{lem:equivRepTwo}
It holds that $(F'',\{(X_e'',\bar{X}_e'')\}|_{e\in T})$ is an $(h',\{({L^e}',{R^e}')\}|_{e\in T})$-solution to the original instance.
\end{lemma}

Next, by exhaustive application of Lemma \ref{lem:multiplicity}, we get that between every for of vertices there are at most $k2^{|T|}$ non-terminal edges. Thus, by Observation \ref{obs:numVerticesRep}, we have the following observation for a large enough $\lambda$ (here, $\lambda=3\lambda'$ is sufficient).

\begin{observation}\label{obs:numEdgesRep}
$|E(G^*)\setminus E(G[P])|\leq k2^{|T|}\cdot (|V(G^*)\setminus P|)^2\leq 2^{2^{\lambda(t+k^2)|T|}}$.
\end{observation}

\subsubsection{Solving the Resulting Instance}\label{sec:solveEnd}

Observe that the number of vertices of the graph associated with the instance constructed in Section \ref{sec:replace} is smaller than $|V(G)|$ since $|V(G^*)\cap Q|=q<|Q|$. Thus, \alg{Recurs} can call itself recursively to solve the instance correctly. By using Lemmas \ref{lem:equivRepOne} and \ref{lem:equivRepTwo}, we obtain a solution to the original problem. The running time of the recursive call is $\tau(|V(G^*)|,|E(G^*)|,k,t,|T|)$. Therefore, for the appropriate constants $\delta''$ and $\delta$, the total running time is bounded by
\[\begin{array}{l}
\medskip
\tau(|V(G^*)|,|E(G^*)|,k,t,|T|)+(n+m+ 2^{2^{(t+k^2)|T|}})^{\delta''}\\

\medskip
\hspace{5em} +\ \tau(|Q|,|E(G[Q])|,k,t,|T|) +(n+m+k+t+|T|)^{\delta'}\\

\leq \tau(|V(G^*)|,|E(G^*)|,k,t,|T|) + \tau(|Q|,|E(G[Q])|,k,t,|T|) + (n+m+2^{2^{(t+k^2)|T|}})^{\delta}
\end{array}
\]

By Observations \ref{obs:numVerticesRep} and \ref{obs:numEdgesRep}, the running time is bounded by $\tau(2^{2^{\lambda(t+k^2)|T|}}+|P|,2^{2^{\lambda(t+p)|T|}}+|E(G[P\cup U])|,k,t,|T|) + \tau(|Q|,|E(G[Q])|,k,t,|T|) + (n+m+2^{2^{(t+k^2)|T|}})^{\delta}$. We next show that this is bounded by $\tau(n,m,k,t,|T|)=2^{2^{\alpha(t+k^2)|T|}}\cdot (n+m)^{\alpha}$.

Notice that $|P|+|Q|=n$ and $|E(G[P\cup U])|+|E(G[Q])|=m$. Let us denote $|Q|=\widehat{n}$ and $|E(G[Q])|=\widehat{m}$. 
Therefore, we have that
\[\begin{array}{l}
\medskip
\tau(2^{2^{\lambda(t+k^2)|T|}}+|P|,2^{2^{\lambda(t+k^2)|T|}}+|E(G[P\cup U])|,k,t,|T|) + \tau(|Q|,|E(G[Q])|,k,t,|T|)\\

\medskip
\hspace{5em} +\ (n+m+2^{2^{(t+k^2)|T|}})^{\delta}\\

\medskip
= \tau(2^{2^{\lambda(t+k^2)|T|}}+n-\widehat{n},2^{2^{\lambda(t+k^2)|T|}}+m-\widehat{m},k,t,|T|) + \tau(\widehat{n},\widehat{m},k,t,|T|) + (n+m+2^{2^{(t+k^2)|T|}})^{\delta}\\

\medskip
= 2^{2^{\alpha(t+k^2)|T|}}\cdot((2\cdot 2^{2^{\lambda(t+k^2)|T|}}+n-\widehat{n}+m-\widehat{m})^{\alpha} + (\widehat{n}+\widehat{m})^{\alpha}) + (n+m+2^{2^{(t+k^2)|T|}})^{\delta}

\end{array}\]

Recall that $\min\{n,n-\widehat{n}\}\geq q+1=2^{2^{\lambda(t+k^2)|T|}}+1$, and since we ensured that our graphs is connected, it holds that  $\min\{m,m-\widehat{m}\}\geq q=2^{2^{\lambda(t+k^2)|T|}}$. Thus, we have that 
\[\begin{array}{l}
\medskip
2^{2^{\alpha(t+k^2)|T|}}\cdot((2\cdot 2^{2^{\lambda(t+k^2)|T|}}+n-\widehat{n}+m-\widehat{m})^{\alpha} + (\widehat{n}+\widehat{m})^{\alpha}) + (n+m+2^{2^{(t+k^2)|T|}})^{\delta}\\

\medskip
\leq 2^{2^{\alpha(t+k^2)|T|}}\cdot((n+m-1)^{\alpha} + (2\cdot 2^{2^{\lambda(t+k^2)|T|}}+1)^{\alpha}) + (n+m+2^{2^{(t+k^2)|T|}})^{\delta}\\

\leq \tau(n,m,k,t,|T|) - 2^{2^{\alpha(t+k^2)|T|}}\cdot((n+m-1)^{\alpha-1} - (2\cdot 2^{2^{\lambda(t+k^2)|T|}}+1)^{\alpha})  + (n+m+2^{2^{(t+k^2)|T|}})^{\delta}
\end{array}\]

It remains to show that
\[2^{2^{\alpha(t+k^2)|T|}}\cdot(2\cdot 2^{2^{\lambda(t+k^2)|T|}}+1)^{\alpha} + (n+m+2^{2^{(t+k^2)|T|}})^{\delta}\leq 2^{2^{\alpha(t+k^2)|T|}}\cdot(n+m-1)^{\alpha-1}\]

By choosing a large enough $\alpha$ in advance, we have that it is sufficient to show that
\[2^{2^{\alpha(t+k^2)|T|}}\cdot(2\cdot 2^{2^{\lambda(t+k^2)|T|}}+1)^{\alpha}\leq 2^{2^{\alpha(t+k^2)|T|}}\cdot(n+m-2)^{\alpha-2}\]

That is, it is sufficient to show that
\[(2\cdot 2^{2^{\lambda(t+k^2)|T|}}+1)^3 \leq n+m-2\]

Recall that $n>s=q^4=(2^{2^{\lambda(t+k^2)|T|}})^4$, else the current instance should have already been handled in Section \ref{sec:small}. Thus, we obtain that indeed it holds that $(2\cdot 2^{2^{\lambda(t+k^2)|T|}}+1)^2 \leq n+m-2$.
This concludes the proof of Lemma \ref{lem:aedc}, which by our earlier discussions, also concludes the proof of Theorem \ref{thmdual}. 

\section{Hardness}\label{section:hardness}
In this section we give our hardness result for \spcgraph{}. First, we show that the problem is \classW{1}-hard when parameterized by $r+k+|T|$.

\begin{theorem}\label{thm:w-hard}
\spcgraph{} is \classW{1}-hard when parameterized by $r+k+|T|$.
\end{theorem}

\begin{proof}
We reduce from the {\sc Multicolored Clique} problem:

\defparproblem{{\sc Multicolored Clique}}%
{A graph $G$ and a partition $X_1, \ldots,X_k$ of $V(G)$.}%
{$k$}
{Is there a clique $K$ of $G$ with $|K\cap X_i|=1$ for $i\in\{1,\ldots,k\}$?}

This problem is well-known to be \classW{1}-complete (see~\cite{FellowsHRV09,Pietrzak03}). Let $(G,X_1,\ldots,X_k)$ be an instance of {\sc Multicolored Clique}. We assume without loss of generality that $k\geq 2$ and each $X_i$ is an independent set. 

First, we construct the graph $G'$ as follows.
\begin{itemize}
\item Construct a copy of $G$.
\item For each $i\in \{1,\ldots,k\}$, construct a vertex $u_i$ and make it adjacent to each vertex of $X_i$.
\item Make $\{u_1,\ldots,u_k\}$ a clique.
\item Construct $k$ isolated vertices $x_1,\ldots,x_k$ and $k(k-1)/2$ isolated vertices $y_{ij}$ for $1\leq i<j\leq k$.
\end{itemize}
Denote by $S$ the set of edges $\{u_iu_j\mid 1\leq i<j\leq k\}$.

Next, we construct the $(|V(G')|\times |E(G')|)$-matrix $P$ over $GF(2)$. Recall that as the rows of $I(G')$ correspond to the vertices of $G'$ and the columns of $I(G')$ correspond to the edges of $G'$, we index the elements of $P$ by vertices and edges of $G'$ and denote the elements by $p_{v,e}$ for $v\in V(G')$ and $e\in E(G')$.
\begin{itemize}
\item For each $v\in V(G)$ and $e\in E(G')$, set $p_{v,e}=0$.
\item For $i\in \{1,\ldots,k\}$ and $e\in E(G')\setminus S$, set $p_{x_i,e}=1$ if $e=x_iz$ for $z\in X_i$ and $p_{v,e}=0$ otherwise.
\item For $i\in \{1,\ldots,k\}$ and $e=u_su_t\in S$ for $1\leq s<t\leq k$, set $p_{x_i,e}=1$ if $i=s$ or $i=t$  and $p_{x_i,e}=0$ otherwise.
\item For $1\leq i<j\leq k$ and $e=vw\in E(G)$, set $p_{y_{ij},e}=1$ if $v\in X_i$, $w\in X_j$ or $v\in X_j$, $w\in X_i$, and $p_{y_{ij},e}=0$ otherwise.  
\item For $1\leq i<j\leq k$ and $e\in E(G')\setminus (E(G)\cup T)$, set $p_{y_{ij},e}=0$.  
\item For $1\leq i<j\leq k$ and $e=u_su_t\in S$ for $1\leq s<t\leq k$, set $p_{y_{ij},e}=1$ if $i=s$, $j=t$, and $p_{y_{ij},e}=0$ otherwise.
\end{itemize}

Let $A=I(G')+P$. Recall that we denote the elements of $A$ by $a_{v,e}$ and index them by  vertices and edges of $G'$ as in $I(G')$ and $P$. Denote by $A^e$ the column of $A$ corresponding to $e\in E(G')$. We set $T=\{A^e\mid e\in S\}$. 
Finally, let $k'=k(k+1)/2=k+k(k-1)/2$. 
Let $M$ be the binary matroid represented by $A$.

We claim that $(G,X_1,\ldots,X_k)$ is a yes-instance of {\sc Multicolored Clique} if and only if $(G,P,k')$ is a yes-instance of \spcgraph{}.

Suppose that $K=\{v_1,\ldots,v_k\}$, where $v_i\in X_i$ for $i\in \{1,\ldots,k\}$, is a clique of $G$. Let $F=\{A^{u_iv_i}\mid 1\leq i\leq k\}\cup\{A^{v_iv_j}\mid 1\leq i<j\leq k\}$.   Let $1\leq i<j\leq k$. It is straightforward to verify that $A^{u_iu_j}=A^{u_iv_i}+A^{v_iv_j}+A^{u_jv_j}$. Hence, $\{A^{u_iu_j},A^{u_iv_i},A^{v_iv_j},A^{u_jv_j}\}$ is a cycle of $M$. It implies that $F$ spans each $A^{u_iu_j}\in T$ and, therefore, $F$ spans $T$ in $M$. Since $|F|=k'$,  $(G,P,k')$ is a yes-instance of \spcgraph{}.

Assume that $(G,P,k')$ is a yes-instance of \spcgraph{}. Then there is a set of columns $F$ of $A$ such that $F\cap T=\emptyset$, $|F|\leq k'$ and $F$ spans $T$ in $M$. It means that for every $A^t\in T$, there is a set columns $R\subseteq F$ such that $A^t=\sum_{A^e\in R}A^e$. Consider $A^{t}\in T$ for $t=u_iu_j\in S$. Since $a_{x_i,t}=1$, $R$ contains some $A^{e}$ with $a_{x_i,e}=1$. It implies that $R$ contains $A^e$ for $e=u_iv$ for $v\in X_i$. Because $a_{y_{ij},t}=1$, $R$ contains some $A^{e}$ with $a_{y_{ij},e}=1$ and, therefore, $R$ contains $A^e$ for $e=vw$ for $v\in X_i$ and $w\in X_j$. Since these property hold for all $t\in S$, we obtain that
\begin{itemize}
\item for every $i\in\{1,\ldots,k\}$, $F$ contains at least one column $A^e$ for $e=u_iv$ for $v\in X_i$,
\item for every $1\leq i<j\leq k$, $F$ contains at least one column $A^e$ for $e=vw$ for $v\in X_i$ and $w\in X_j$.
\end{itemize}
Since $|F|\leq k'$, 
\begin{itemize}
\item[(i)] for every $i\in\{1,\ldots,k\}$, $F$ contains exactly one column $A^e$ for $e=u_iv$ for $v\in X_i$,
\item[(ii)] for every $1\leq i<j\leq k$, $F$ contains exactly one column $A^e$ for $e=vw$ for $v\in X_i$ and $w\in X_j$.
\end{itemize}

Let $K=\{v_1,\ldots,v_k\}$ be the set of vertices of $G'$ such that $A^{u_iv_i}\in F$. We show that $K$ is a clique of $G$. Let $1\leq i<j\leq k$. We have that $F$ contains exactly one column $A^{vw}$ for some $v\in X_i$ and $w\in X_j$. Let $t=u_iu_j$. There is a set columns $R\subseteq F$ such that $A^t=\sum_{A^e\in R}A^e$. Recall that $A^t$ has the following non-zero elements:
$a_{u_i,t}$, $a_{u_j,t}$, $a_{x_i,t}$, $a_{x_j,t}$ and $a_{y_{ij},t}$, and all other elements are zeros. By (i) and (ii), we conclude that $R=\{A^{u_iv_i},A^{vw},A^{u_jv_j}\}$. Since, $A^t=A^{u_iv_i}+A^{vw}+A^{u_jv_j}$, $v_i=v$ and $v_j=w$. Hence, $v_iv_j\in E(G)$. It implies that $K$ is a clique. 

This completes the correctness proof for our  reduction.  Notice that since $P$ has $k+k(k-1)/2=k'$ non-zero rows, $\rank(P)\leq k'$. Clearly, $|T|=k(k-1)/2$. Hence, 
\spcgraph{} is \classW{1}-hard when parameterized by $r+k+|T|$.
\end{proof}

Now we prove that \spcgraph{} is \classParaNP-complete when parameterized by $r+|T|$.

\begin{theorem}\label{thm:paraNP-hard}
\spcgraph{} is \classNP-complete when restricted to the instances with $r\leq 2$ and $|T|\leq 2$.
\end{theorem}

\begin{proof}
We reduce from the \textsc{3-Dimensional Matching} problem that is well-known to be \classNP-complete~\cite{GareyJ79}: 

\defproblemu{3-Dimensional Matching}%
{Set $S\subseteq  X\times Y\times Z$, where $X$, $Y$, and $Z$ are disjoint sets having the same number $q$ of elements.}%
{Does $S$ contain a \emph{matching}, i.e., a subset $S'\subseteq S$ such that $|S'|=q$ and no two elements of $S'$ agree in any coordinate?}
    
\noindent
Let  $X=\{x_1,\ldots,x_q\}$, $Y=\{y_1,\ldots,y_q\}$, $Z=\{z_1,\ldots,z_q\}$ and $M=\{m_1,\ldots,m_p\}$.  

We construct the multigraph $G$ as follows.
\begin{itemize}
\item Assume that the set $S$ contains $p$ elements, and construct 4 sets of vertices  $X=\{x_1,\ldots,x_q\}$, $Y=\{y_1,\ldots,y_q\}$, $Z=\{z_1,\ldots,z_q\}$ and $S=\{s_1,\ldots,s_p\}$.
\item For each triple $s_r\in S\subseteq  X\times Y\times Z$ such that $s_r=(x_h,y_i,z_j)$, make the vertex $s_r$ adjacent to $x_h$, $y_i$ and $z_j$. 
\item Construct 2 vertices $a$, $b$ and loops $aa$, $bb$. 
\end{itemize}

Next, we construct the $(|V(G)|\times |E(G)|)$-matrix $P$ over $GF(2)$. As before, we index the elements of $P$ by vertices and edges of $G$ and denote the elements by $p_{v,e}$ for $v\in V(G)$ and $e\in E(G)$.
\begin{itemize}
\item For each $v\in V(G)$ and $e\in E(G)\setminus \{aa,bb\}$, set $p_{v,e}=0$.
\item For each $i\in \{1,\ldots,p\}$ and $e\in \{aa,bb\}$, set $p_{s_i,e}=0$.
\item For $i\in \{1,\ldots,q\}$, set $p_{x_i,a}=1$, $p_{y_i,a}=1$ and $p_{z_i,a}=0$.
\item For $i\in \{1,\ldots,q\}$, set $p_{x_i,b}=1$, $p_{y_i,b}=0$ and $p_{z_i,b}=1$.
\end{itemize}
Notice that since $P$ has 2 non-zero columns, $rank(P)\leq 2$.

Let $A=I(G)+P$. We denote the elements of $A$ by $a_{v,e}$ and index them by  vertices and edges of $G$ as in $I(G)$ and $P$. Let also $A^e$ be the column of $A$ corresponding to $e\in E(G)$. Set  $T=\{A^{aa},A^{bb}\}$.
Finally, let $k=3q$.
Let $M$ be the binary matroid represented by $A$.

We claim that $(X,Y,Z,S)$ is a yes-instance of \textsc{3-Dimensional Matching}
 if and only if $(G,P,k)$ is a yes-instance of \spcgraph{}.

Assume that the set of triples $S$ contains a matching $S\rq{}\subseteq S$. For the corresponding set of vertices $S'$ of $G$, consider $H=G[X\cup Y\cup Z\cup S']$. Let $F=\{A^e\mid e\in E(H)\}$. We have that $T\subseteq\spn(F)$. In particular, for the sets of edges $E_X$, $E_Y$ and $E_Z$ of $H$ incident to the vertices of $X$, $Y$ and $Z$ respectively, $A^{aa}=\sum_{e\in E_X\cup E_Y}A^e$ and $A^{bb}=\sum_{e\in E_X\cup E_Z}A^e$. Since $|E(H)|=3q$,  $(G,P,k)$ is a yes-instance of \spcgraph{}.

Suppose that $(G,P,k)$ is a yes-instance of \spcgraph{}. Let $F$ be a set of columns of $A$ such that $|F|=k$ and $T\subseteq\spn(F)$. Let $R$ be the set of edges of $G$ corresponding to the columns of $F$. 
Since $p_{x_i,a}=1$, $p_{y_i,a}=1$ and $p_{z_i,b}=1$ for $i\in \{1,\ldots,q\}$, for each $i\in \{1,\ldots,q\}$, $R$ contains edges incident to $x_i$, $y_i$ and $z_i$ respectively. Because $|F|=q$, we obtain that  for each $i\in \{1,\ldots,q\}$, $R$ contains exactly one edge incident to $x_i$, exactly one edge incident to $y_i$ and exactly one edge incident to $z_i$. Denote by $E_X,E_Y,E_Z\subseteq R$ the sets of edges incident to $X$, $Y$ and $Z$ respectively. Notice that these sets are matchings in $G$. 
Let $F_{a}\subseteq F$ be such that $\sum_{e\in F_{a}}=A^{aa}$ and let $E_{a}=\{e\in R\mid e\in F_{aa}\}$. 
 Let $H_a=(V(G), E_a)$.
 Since  $p_{x_i,a}=1$ and  $p_{y_i,a}=1$ for $i\in\{1,\ldots,q\}$, $E_X\cup E_Y\subseteq E_a$.  Since  $p_{z_i,a}=0$ for $i\in \{1,\ldots,q\}$ and $p_{v,e}=0$ for 
 $v\in V(G)$ and $e\in E(G)\setminus \{aa,bb\}$, we have that the vertices of $V(G)\setminus (X\cup Y)$ have even degrees in $H_a$. This implies that the matchings $E_X$ and $E_Y$ match the vertices of $X$ and $Y$ to the same subset $S'\subseteq S$. By the same arguments for the terminal $A^{bb}$, we obtain that the  $E_Z$ matches the vertices of $Z$ to  $S'$. By the construction of $G$, we have that the corresponding set of triples $S'\subseteq S\subseteq X\times Y\times Z$ is a matching, that is, $(X,Y,Z,S)$ is a yes-instance of \textsc{3-Dimensional Matching}. 
 \end{proof}

\section{Conclusion}\label{sec:conclusion}
In this paper we established the fixed-parameter tractability of   \spcgraph{} and \spcdual. 
We also know that on the class of binary matroids \SSp is not tractable. So where lies the tractability border for \SSp? 
Our positive results on perturbed matroids, combined with the structure theorem of Geelen,   Gerards,   and Whittle~\cite{GeelenGW15}, rise a natural question: could  the tractability of \SSp  be extended to any proper minor-closed class $\mathcal{M}$ of binary matroids?
Let us note that while we formulate  \SSp only on binary matroids, it can be naturally defined on any class of  matroids. In particular, the parameterized complexity of \SSp on proper minor-closed classes of matroids representable over a finite field is open.

Finally, two concrete open questions. First, what is the parameterized complexity of \spcgraph when 
$|T|$ is a constant and the  parameter is $r+k$? Second, we know that \spcgraph is NP-complete even when $|T|=2$ and $r\leq 2$ (Theorem~\ref{thm:paraNP-hard}). On the other hand, for $r=0$ the problem is in P for any fixed number of terminals (it is actually FPT parameterized by $|T|$). What about the case   $|T|=2$ and $r=1$?

\paragraph{Acknowledgement.} We thank Jim Geelen for valuable insights regarding matroid minors.

\end{document}